\documentclass[12pt]{article}
\usepackage{cite}
\usepackage{amsmath,amssymb,amsfonts,mathtools}
\usepackage{algorithmic}
\usepackage{graphicx}
\usepackage{textcomp}
\usepackage{xcolor}
\usepackage{amsthm}
\usepackage{mathabx}
\usepackage{xurl}
\usepackage{footmisc}
\usepackage{hyperref}
\usepackage{caption}
\usepackage{tablefootnote}
\usepackage{array}
\usepackage{float}
\usepackage{enumerate}
\usepackage{longtable}
\usepackage{wasysym}
\usepackage{indentfirst}

\usepackage{xpatch}
\makeatletter
\AtBeginDocument{\xpatchcmd{\@thm}{\thm@headpunct{.}}{\thm@headpunct{}}{}{}}
\makeatother

\usepackage{soul}

\newcommand{\ovariant}{\text{\st{$\mathtt{o}$}}}
\newcommand\ChangeRT[1]{\noalign{\hrule height #1}}  
\newcommand{\thickhline}{\noalign{\global\arrayrulewidth0.03cm}\hline
                      \noalign{\global\arrayrulewidth0.01cm}}

\DeclareCaptionType{equ}[][]
\theoremstyle{plain}

\newtheorem{theorem}{Theorem}[section]
\newtheorem{lemma}[theorem]{Lemma}
\newtheorem{prop}[theorem]{Proposition}
\newtheorem{cor}[theorem]{Corollary}

\theoremstyle{definition}

\newtheorem{example}[theorem]{Example}
\newtheorem{remark}[theorem]{Remark}

\begin{document}

\title{Construction of self-orthogonal codes over a commutative non-unitary ring  of order 25\footnotemark[1]  \footnotemark[2]}

\author{Jon-Lark Kim   \\ Department of Mathematics and \\ Institute for Mathematical and Data Sciences  \\ Sogang University, Seoul, Korea \\
		{\tt jlkim@sogang.ac.kr } \\
			\\ Marvin Olavides\footnotemark[3]  \\ Department of Mathematics and \\ Institute for Mathematical and Data Sciences  \\ Sogang University, Seoul, Korea \\
		{\tt mmolavides@gmail.com}   \\
		\\ Young Gun Roe \\ KNU Research Institute for Mathematical Sciences \\ Kangwon National University, Chuncheon, Korea  \\
		{\tt ygroe@naver.com}
	}
\date{}
	
\maketitle

\footnotetext[1]{J.-L. Kim was supported in part by the BK21 FOUR (Fostering Outstanding Universities for Research) funded by the Ministry of Education (MOE, Korea) and National Research Foundation of Korea (NRF) under Grant No. 4120240415042 and by Basic Science Research Program through the National Research Foundation of Korea (NRF) funded by the Ministry of Science and ICT under Grant No. RS-2025-24534992.}

\footnotetext[2]{Y.G. Roe was supported by Basic Science Research Program through the National Research Foundation of Korea (NRF) funded by the Ministry of Education (RS-2025-25415913).}

\footnotetext[3]{Corresponding author}

\begin{abstract}
\noindent Codes over non-unitary rings have been studied recently. In particular, codes over the commutative non-unitary ring $I_p$ (in the classification of Fine) of order $p^2$ where $p$ is a prime are being considered.  For $p=2$ (resp. $p=3$), three categories of codes over $I_p$ have been studied: self-orthogonal codes, quasi self-dual codes, and self-dual codes over $I_p$. Using some related mass formulas and building-up constructions, classifications of these codes have been done up to the permutation equivalence (resp. the monomial equivalence) for certain small lengths. In this paper, we take the prime $p=5$ and consider the ring $I_5$. We introduce the notion of linear codes over $I_5$. We also define the same three categories of linear $I_5$-codes, study the structures of these $I_5$-codes and relate them to their associated residue and torsion codes. 
We classify the three categories of codes completely in lengths at most $4$ up to the monomial equivalence for a given type $\{  k_1 , k_2 \}$. Moreover, in the paper of Alahmadi et al. regarding the mass formula for self-orthogonal codes over $I_p$, mistakes in the classification of quasi self-dual codes over $I_5$ had been made such as incorrect automorphism group order of some codes or inconsistency with the mass formula for self-orthogonal codes over $I_p$ for length $n=2$ and type $\{ 1 , 0 \}$ and for length $n=3$ and type $\{ 1, 1 \}$. We correct and improve such results.
\end{abstract}



\section{Introduction}

Coding theory over rings has evolved into a significant area of modern algebraic coding theory, broadening classical concepts that were once confined to finite fields. The discovery by Hammons et al. \cite{HammonsKerdPrepCodes} that nonlinear binary codes such as the Kerdock and Preparata codes can be viewed as binary images under a Gray map of linear codes over $\mathbb{Z}_4$ opened an entirely new direction in the study of codes on more general algebraic structures. Since then, many works have revealed deep links between ring-based codes and topics such as lattices, modular forms, and cryptography \cite{ConwaySloane, DoughGabHarSol, HuffKimSole}.

Within this  framework, self-orthogonal and self-dual codes occupy a central role. They are known for their excellent distance properties and for their close relationships with combinatorial designs and lattices. Alahmadi et al. \cite{E3} mentioned that the classification of self-dual codes over unitary rings and finite fields has rested on two pillars: an algorithm to generate short length codes and a mass formula to signal the completion of the classification \cite{BalmBetNem, Kim_GFq}. One of the construction methods is the so called building-up method. By using a recursion on generator matrices, the building-up method creates a self-dual code of length $n+h$ from a self-dual code of length $n$. This is also called the propagation rule of order $h$. This method was first introduced for binary codes \cite{Kim_extremal} and later was extended to other fields and rings \cite{KimLeeEffConst, KimLeeEucHer, KimLeeMDSSD, LeeHLeeY}.

Codes over non-unitary rings have also been studied recently. In fact, some of the non-unitary rings considered were the rings of order $p^2$ where $p$ is prime in the classification of Fine \cite{Fine}. Seven of these rings are non-unitary and one of them is denoted by $E_p$. Apart from being non-unitary, $E_p$ is also non-commutative.  Alahmadi et al. \cite{TypeIV_E2} introduced the notion of linear codes over the ring $E_2$ and the notions of quasi self-dual codes (self-orthogonal codes of size $2^n$ where $n$ is the length of the code) and type IV codes (quasi self-dual codes with even weights) over $E_2$. Quasi self-dual codes and type IV codes over $E_2$ were  constructed using the building-up method  up to length $n = 12$ \cite{QSD_E2}. The building-up method was also used to classify self-orthogonal codes, one-sided self-dual codes, and self-dual codes over $E_3$ in lengths at most $7$ up to the monomial equivalence. Alahmadi et al. \cite{MassFormula_Ep} gave a mass formula for self-orthogonal codes over the general ring $E_p$ and classified self-orthogonal codes and self-dual codes over $E_p$ where $p = 3$, $5$, and $7$ in short lengths. A classification of self-orthogonal codes, one-sided self-dual codes, and self-dual codes over $E_3$ in lengths at most $7$ and up to the monomial equivalence was done using some related mass formulas and building-up constructions \cite{E3}. One can map codes over $E_p$ to additive codes over $\mathbb F_{p^2}$. Additive codes were originally introduced over $\mathbb{F}_4$ in the context of quantum error correction \cite{CalderbankRainsShorSloane}. Subsequently, the notion was extended to additive codes over $\mathbb{F}_{p^2}$ for odd primes $p$, preserving many of the same structural properties~\cite{Huffman}. This motivates our work.


Another non-unitary ring from the classification of Fine \cite{Fine} is denoted by $I_p$. Unlike $E_p$, $I_p$ is a commutative ring. Similar notions from codes over the ring $E_2$ were used to define linear codes over the commutative non-unitary ring $I_2$ including the notions of quasi self-dual codes, type IV codes, and quasi type IV codes (quasi self-dual codes with even torsion) over $I_2$  \cite{QuasiTypeIV_I2}. In a separate paper, classification of self-orthogonal codes over $I_2$ up to length $n=6$ \cite{BuildUp_I2} was given.  Kim et al. \cite{KimJLYGRoe_I2} extended this classification up to length $n=8$. Alahmadi et al. \cite{MassFormula_Ip} gave a mass formula for self-orthogonal codes over $I_p$ and classified self-orthogonal codes, quasi self-dual codes, and self-dual codes over $I_p$ where $p=3$, $5$, and $7$ for lengths up to $3$. With this mass formula and building-up method, self-orthogonal codes, quasi self-dual codes, and self-dual codes over the ring $I_3$ were classified completely up to length $n=4$ and partially up to length $n=5$ \cite{I3}.

Non-unitary rings such as $E_p$ and  $I_p$ provide a natural framework in which classical notions of duality interact with genuinely new structural phenomena arising from the presence of residue and torsion codes. In this setting, self-duality is often too restrictive, making quasi self-dual codes the appropriate analog. The building-up constructions considered are motivated by the need for explicit generation methods that preserve self-orthogonality or self-duality while respecting the algebraic structure of the ring, especially since mass formulas alone do not yield concrete representatives. These constructions enable systematic classification under the monomial equivalence and provide a practical framework for extending classical recursive techniques to codes over non-unitary rings.

Although $p = 3$ is the smallest odd prime, the case $p = 5$ is the first in which the algebraic and combinatorial structure of $I_p$-codes becomes sufficiently rich to produce a broad range of monomially inequivalent self-orthogonal and quasi self-dual codes. Over $\mathbb{F}_5$, the larger unit group and the associated quadratic and orthogonality relations give rise to a greater variety of admissible configurations than in the case $p = 3$. As a result, building-up constructions over $I_5$ generate substantially more diverse families of codes. In addition, the ring $I_5$ admits a natural reduction map onto $\mathbb{F}_5$, leading to a residue--torsion decomposition that connects codes over $I_5$ with classical linear codes over finite fields. This multilevel viewpoint allows problems over $I_5$ to be studied through well-understood structures over $\mathbb{F}_5$. Combined with explicit recursive constructions, it provides an effective framework for building and classifying longer codes from shorter ones. Consequently, the case $p = 5$ offers a suitable balance between structural complexity and computational tractability, making it particularly appropriate for a detailed classification study.

In this paper, by taking $p=5$, we study codes over the commutative non-unitary ring of order $25$ denoted by $I_5$. We modify some techniques to produce self-orthogonal and self-dual
codes as well as quasi self-dual codes over the ring $I_5$.
Following \cite{MassFormula_Ip}, quasi self-dual codes over $I_5$ are defined as self-orthogonal codes over $I_5$ of length $n$ and size $5^n$. We derive propagation rules of certain orders for self-orthogonal codes, quasi self-dual codes, and self-dual codes over $I_5$. We also classify completely self-orthogonal, quasi self-dual, and self-dual codes over $I_5$ up to length $4$. The classification is made under a monomial action. The mass formulas for these three types of codes are used as stopping criteria for code generation. Although the classification of quasi self-dual codes over $I_5$ has already been made \cite{MassFormula_Ip}, we point out that there are some mistakes in the classification. In particular, for length $n=2$ and type $\{ 1 , 0 \}$, although the result agree with the mass formula for quasi self-dual codes over $I_5$, some of the codes obtained are monomially equivalent to each other and some have wrong order $|\operatorname{Aut} (\mathcal{C})|$ of the automorphism group. For length $n=3$ and type $\{ 1, 1 \}$, the result does not agree with the related mass formula. Thus, we aim to correct these mistakes.

This paper is organized as follows. The next section presents the preliminary notions and notations needed in the latter sections. The propagation rules for the construction of self-orthogonal, quasi self-dual, and self-dual codes are derived in Section 3. Section 4 presents the complete classification for these three categories for a given a length $n \leq 4$ and type $\{  k_1 , k_2 \}$. It is also in this classification where the results obtained by Alahmadi et al. \cite{MassFormula_Ip} are corrected.  Section 5 gives the conclusion of the study and recommendations for further research.

\section{Preliminaries}

\subsection{Linear Codes over {\boldmath $\mathbb{F}_5$}}

The finite field of order $5$ is denoted by $\mathbb{F}_5$ and the vector space of $n$-tuples over $\mathbb{F}_{5}$ is denoted by $\mathbb{F}_{5}^{n}$. A linear code $\mathcal{C}$ over  $\mathbb{F}_{5}$ is a $k$-dimensional subspace of $\mathbb{F}_{5}^{n}$ and we call $\mathcal{C}$ an $[n, k]_{5}$ code if it has length $n$ and dimension $k$. Given $\mathbf{x} = (x_1 , \dotsc , x_n)$ and $\mathbf{y} = ( y_1 , \dotsc , y_n) \in \mathbb{F}_5^n$, the standard inner product is denoted by $( \mathbf{x} , \mathbf{y} ) \coloneqq \sum_{i=1}^{n}  x_i y _i$. The dual of a  linear code over $\mathbb{F}_5$ is denoted by $\mathcal{C}^{\perp}$ and is defined as
\begin{align*}
\mathcal{C}^{\perp} \coloneqq \left\{  \mathbf{y} \in \mathbb{F}_{5}^{n} \mid \forall \mathbf{x} \in \mathcal{C},(\mathbf{x}, \mathbf{y})=0 \right\}  .
\end{align*}
A code $\mathcal{C}$ is self-orthogonal if  $\mathcal{C} \subseteq \mathcal{C}^{\perp}$ and is self-dual if $\mathcal{C} = \mathcal{C}^{\perp}$.

The (Hamming) weight of a vector $\mathbf{x} \in \mathbb{F}_5^n$ is the number of its nonzero coordinates and is denoted by $\operatorname{wt}(\mathbf{x})$. For $\mathbf{x}, \mathbf{y} \in \mathbb{F}_{5}^{n}$, their  Hamming distance is defined by $d(\mathbf{x}, \mathbf{y}) = \operatorname{wt}(\mathbf{x}-\mathbf{y})$. The (minimum) distance $d(\mathcal{C})$ of a linear code $\mathcal{C}$ is
\begin{align*}
d(\mathcal{C}) = \min \{d(\mathbf{x}, \mathbf{y}) \mid \mathbf{x}, \mathbf{y} \in \mathcal{C}, \mathbf{x} \neq \mathbf{y}\}=\min \{ \operatorname{wt}  (\mathbf{c}) \mid \mathbf{c} \in \mathcal{C}, \mathbf{c} \neq \mathbf{0}\}  .
\end{align*}

A  linear code over $\mathbb{F}_5$ is called an $[n, k, d]_{5}$ code if it has length $n$, dimension $k$, and distance $d$. Using  \texttt{MAGMA} \cite{MAGMA} notation, we define the weight distribution of a code $\mathcal{C}$ to be the sequence
$$\left[   \langle 0 , 1  \rangle  ,  \dotsc ,  \langle i , A_i  \rangle   ,  \dotsc  ,  \langle  n  ,  A_n  \rangle    \right]$$
where $A_{i}$ denotes the number of codewords in $\mathcal{C}$ of weight $i$.

\subsection{The Ring {\boldmath $I_5$}}

Based on the classification of rings in \cite{Fine}, we define the ring $I_{5}$ on two generators $\mathtt{a}$ and $\mathtt{b}$ by the relations
$$
I_{5} \coloneqq \left\langle \mathtt{a} , \mathtt{b}   \mid 5 \mathtt{a}  = 5 \mathtt{b} = 0, \mathtt{a}^{2} = \mathtt{b},  \mathtt{a b}=0   \right\rangle  .
$$
Thus,  $I_{5}$ has  characteristic $5$ and consists of $25$ elements, i.e.,
$$	
I_5 = \left\lbrace  0 ,  \mathtt{a} , \mathtt{b}  ,    \mathtt{c} ,    \mathtt{d} ,   \mathtt{e} ,    \mathtt{f} ,   \mathtt{g} ,     \mathtt{h} ,     \mathtt{i} ,     \mathtt{j} ,     \mathtt{k} ,     \mathtt{l} ,     \mathtt{m} ,     \mathtt{n} ,   \ovariant  ,     \mathtt{p} ,     \mathtt{q} ,     \mathtt{r} ,     \mathtt{s} ,     \mathtt{t} ,     \mathtt{u} ,     \mathtt{v} ,     \mathtt{w} ,     \mathtt{x}        \right\rbrace
$$
where
\begin{gather*}
\mathtt{c} = \mathtt{a} + \mathtt{b} ,  \quad  \mathtt{d} = 2 \mathtt{b} ,  \quad   \mathtt{e}  =  2 \mathtt{a}  ,    \quad   \mathtt{f}  = \mathtt{e} + \mathtt{b}  ,  \quad   \mathtt{g} = \mathtt{a} + \mathtt{d}  ,  \quad   \mathtt{h} = \mathtt{e} + \mathtt{d} , \\
\mathtt{i} = 3 \mathtt{b} , \quad  \mathtt{j} =  3 \mathtt{a} , \quad  \mathtt{k} = \mathtt{j} + \mathtt{b}  ,  \quad   \mathtt{l} = \mathtt{a} + \mathtt{i} ,  \quad    \mathtt{m} = \mathtt{j} + \mathtt{i} ,  \quad   \mathtt{n}  =  4  \mathtt{b}  ,  \\
\ovariant = 4 \mathtt{a}  ,  \quad    \mathtt{p} = \ovariant + \mathtt{b} ,  \quad  \mathtt{q} = \mathtt{a} + \mathtt{n} ,  \quad   \mathtt{r} = \ovariant + \mathtt{n}  ,  \quad   \mathtt{s} = \mathtt{j} + \mathtt{d}  ,  \quad   \mathtt{t} = \ovariant + \mathtt{d}  ,    \\
\mathtt{u} = \mathtt{e} + \mathtt{i}  ,  \quad    \mathtt{v} = \ovariant +  \mathtt{i}  ,  \quad   \mathtt{w}  = \mathtt{e} + \mathtt{n}  ,  \quad   \mathtt{x}  =  \mathtt{j} + \mathtt{n}  .
\end{gather*}

Hence, each element of $I_5$ can be written as $c_{ij} = i \mathtt{a} + j \mathtt{b}$ where $0 \leq i , j < 5$. The addition and the multiplication tables for the ring $I_5$ are given in Tables \ref{add table} and \ref{mult table}.

\begin{table}[H]
\caption{Addition table for the ring $I_5$}  \label{add table}
\centering
\scalebox{0.9}{
\begin{tabular}{ c ! {\vrule width 1.5pt} c|c|c|c|c|c|c|c|c|c|c|c|c|c|c|c|c|c|c|c|c|c|c|c|c}

$+$ & $0$ & $\mathtt{a}$ & $\mathtt{b}$ & $\mathtt{c}$ & $\mathtt{d}$ & $\mathtt{e}$ & $\mathtt{f}$ & $\mathtt{g}$ & $\mathtt{h}$ & $\mathtt{i}$ & $\mathtt{j}$ & $\mathtt{k}$ & $\mathtt{l}$ & $\mathtt{m}$ & $\mathtt{n}$ & $\ovariant$ & $\mathtt{p}$ & $\mathtt{q}$ & $\mathtt{r}$ & $\mathtt{s}$ & $\mathtt{t}$ & $\mathtt{u}$ & $\mathtt{v}$ & $\mathtt{w}$ & $\mathtt{x}$ \\
\ChangeRT{1.5pt}
$0$ & $0$ & $\mathtt{a}$ & $\mathtt{b}$ & $\mathtt{c}$ & $\mathtt{d}$ & $\mathtt{e}$ & $\mathtt{f}$ & $\mathtt{g}$ & $\mathtt{h}$ & $\mathtt{i}$ & $\mathtt{j}$ & $\mathtt{k}$ & $\mathtt{l}$ & $\mathtt{m}$ & $\mathtt{n}$ & $\ovariant$ & $\mathtt{p}$ & $\mathtt{q}$ & $\mathtt{r}$ & $\mathtt{s}$ & $\mathtt{t}$ & $\mathtt{u}$ & $\mathtt{v}$ & $\mathtt{w}$ & $\mathtt{x}$ \\ \hline
$\mathtt{a}$ & $\mathtt{a}$ & $\mathtt{e}$ & $\mathtt{c}$ & $\mathtt{f}$ & $\mathtt{g}$ & $\mathtt{j}$ & $\mathtt{k}$ & $\mathtt{h}$ & $\mathtt{s}$ & $\mathtt{l}$ & $\ovariant$ & $\mathtt{p}$ & $\mathtt{u}$ & $\mathtt{v}$ & $\mathtt{q}$ & $0$ & $\mathtt{b}$ & $\mathtt{w}$ & $\mathtt{n}$ & $\mathtt{t}$ & $\mathtt{d}$ & $\mathtt{m}$ & $\mathtt{i}$ & $\mathtt{x}$ & $\mathtt{r}$ \\ \hline
$\mathtt{b}$ & $\mathtt{b}$ & $\mathtt{c}$ & $\mathtt{d}$ & $\mathtt{g}$ & $\mathtt{i}$ & $\mathtt{f}$ & $\mathtt{h}$ & $\mathtt{l}$ & $\mathtt{u}$ & $\mathtt{n}$ & $\mathtt{k}$ & $\mathtt{s}$ & $\mathtt{q}$ & $\mathtt{x}$ & $0$ & $\mathtt{p}$ & $\mathtt{t}$ & $\mathtt{a}$ & $\ovariant$ & $\mathtt{m}$ & $\mathtt{v}$ & $\mathtt{w}$ & $\mathtt{r}$ & $\mathtt{e}$ & $\mathtt{j}$ \\ \hline
$\mathtt{c}$ & $\mathtt{c}$ & $\mathtt{f}$ & $\mathtt{g}$ & $\mathtt{h}$ & $\mathtt{l}$ & $\mathtt{k}$ & $\mathtt{s}$ & $\mathtt{u}$ & $\mathtt{m}$ & $\mathtt{q}$ & $\mathtt{p}$ & $\mathtt{t}$ & $\mathtt{w}$ & $\mathtt{r}$ & $\mathtt{a}$ & $\mathtt{b}$ & $\mathtt{d}$ & $\mathtt{e}$ & $0$ & $\mathtt{v}$ & $\mathtt{i}$ & $\mathtt{x}$ & $\mathtt{n}$ & $\mathtt{j}$ & $\ovariant$ \\ \hline
$\mathtt{d}$ & $\mathtt{d}$ & $\mathtt{g}$ & $\mathtt{i}$ & $\mathtt{l}$ & $\mathtt{n}$ & $\mathtt{h}$ & $\mathtt{u}$ & $\mathtt{q}$ & $\mathtt{w}$ & $0$ & $\mathtt{s}$ & $\mathtt{m}$ & $\mathtt{a}$ & $\mathtt{j}$ & $\mathtt{b}$ & $\mathtt{t}$ & $\mathtt{v}$ & $\mathtt{c}$ & $\mathtt{p}$ & $\mathtt{x}$ & $\mathtt{r}$ & $\mathtt{e}$ & $\ovariant$ & $\mathtt{f}$ & $\mathtt{k}$ \\ \hline
$\mathtt{e}$ & $\mathtt{e}$ & $\mathtt{j}$ & $\mathtt{f}$ & $\mathtt{k}$ & $\mathtt{h}$ & $\ovariant$ & $\mathtt{p}$ & $\mathtt{s}$ & $\mathtt{t}$ & $\mathtt{u}$ & $0$ & $\mathtt{b}$ & $\mathtt{m}$ & $\mathtt{i}$ & $\mathtt{w}$ & $\mathtt{a}$ & $\mathtt{c}$ & $\mathtt{x}$ & $\mathtt{q}$ & $\mathtt{d}$ & $\mathtt{g}$ & $\mathtt{v}$ & $\mathtt{l}$ & $\mathtt{r}$ & $\mathtt{n}$ \\ \hline
$\mathtt{f}$ & $\mathtt{f}$ & $\mathtt{k}$ & $\mathtt{h}$ & $\mathtt{s}$ & $\mathtt{u}$ & $\mathtt{p}$ & $\mathtt{t}$ & $\mathtt{m}$ & $\mathtt{v}$ & $\mathtt{w}$ & $\mathtt{b}$ & $\mathtt{d}$ & $\mathtt{x}$ & $\mathtt{n}$ & $\mathtt{e}$ & $\mathtt{c}$ & $\mathtt{g}$ & $\mathtt{j}$ & $\mathtt{a}$ & $\mathtt{i}$ & $\mathtt{l}$ & $\mathtt{r}$ & $\mathtt{q}$ & $\ovariant$ & $0$ \\ \hline
$\mathtt{g}$ & $\mathtt{g}$ & $\mathtt{h}$ & $\mathtt{l}$ & $\mathtt{u}$ & $\mathtt{q}$ & $\mathtt{s}$ & $\mathtt{m}$ & $\mathtt{w}$ & $\mathtt{x}$ & $\mathtt{a}$ & $\mathtt{t}$ & $\mathtt{v}$ & $\mathtt{e}$ & $\ovariant$ & $\mathtt{c}$ & $\mathtt{d}$ & $\mathtt{i}$ & $\mathtt{f}$ & $\mathtt{b}$ & $\mathtt{r}$ & $\mathtt{n}$ & $\mathtt{j}$ & $0$ & $\mathtt{k}$ & $\mathtt{p}$ \\ \hline
$\mathtt{h}$ & $\mathtt{h}$ & $\mathtt{s}$ & $\mathtt{u}$ & $\mathtt{m}$ & $\mathtt{w}$ & $\mathtt{t}$ & $\mathtt{v}$ & $\mathtt{x}$ & $\mathtt{r}$ & $\mathtt{e}$ & $\mathtt{d}$ & $\mathtt{i}$ & $\mathtt{j}$ & $0$ & $\mathtt{f}$ & $\mathtt{g}$ & $\mathtt{l}$ & $\mathtt{k}$ & $\mathtt{c}$ & $\mathtt{n}$ & $\mathtt{q}$ & $\ovariant$ & $\mathtt{a}$ & $\mathtt{p}$ & $\mathtt{b}$ \\ \hline
$\mathtt{i}$ & $\mathtt{i}$ & $\mathtt{l}$ & $\mathtt{n}$ & $\mathtt{q}$ & $0$ & $\mathtt{u}$ & $\mathtt{w}$ & $\mathtt{a}$ & $\mathtt{e}$ & $\mathtt{b}$ & $\mathtt{m}$ & $\mathtt{x}$ & $\mathtt{c}$ & $\mathtt{k}$ & $\mathtt{d}$ & $\mathtt{v}$ & $\mathtt{r}$ & $\mathtt{g}$ & $\mathtt{t}$ & $\mathtt{j}$ & $\ovariant$ & $\mathtt{f}$ & $\mathtt{p}$ & $\mathtt{h}$ & $\mathtt{s}$ \\ \hline
$\mathtt{j}$ & $\mathtt{j}$ & $\ovariant$ & $\mathtt{k}$ & $\mathtt{p}$ & $\mathtt{s}$ & $0$ & $\mathtt{b}$ & $\mathtt{t}$ & $\mathtt{d}$ & $\mathtt{m}$ & $\mathtt{a}$ & $\mathtt{c}$ & $\mathtt{v}$ & $\mathtt{l}$ & $\mathtt{x}$ & $\mathtt{e}$ & $\mathtt{f}$ & $\mathtt{r}$ & $\mathtt{w}$ & $\mathtt{g}$ & $\mathtt{h}$ & $\mathtt{i}$ & $\mathtt{u}$ & $\mathtt{n}$ & $\mathtt{q}$ \\ \hline
$\mathtt{k}$ & $\mathtt{k}$ & $\mathtt{p}$ & $\mathtt{s}$ & $\mathtt{t}$ & $\mathtt{m}$ & $\mathtt{b}$ & $\mathtt{d}$ & $\mathtt{v}$ & $\mathtt{i}$ & $\mathtt{x}$ & $\mathtt{c}$ & $\mathtt{g}$ & $\mathtt{r}$ & $\mathtt{q}$ & $\mathtt{j}$ & $\mathtt{f}$ & $\mathtt{h}$ & $\ovariant$ & $\mathtt{e}$ & $\mathtt{l}$ & $\mathtt{u}$ & $\mathtt{n}$ & $\mathtt{w}$ & $0$ & $\mathtt{a}$ \\ \hline
$\mathtt{l}$ & $\mathtt{l}$ & $\mathtt{u}$ & $\mathtt{q}$ & $\mathtt{w}$ & $\mathtt{a}$ & $\mathtt{m}$ & $\mathtt{x}$ & $\mathtt{e}$ & $\mathtt{j}$ & $\mathtt{c}$ & $\mathtt{v}$ & $\mathtt{r}$ & $\mathtt{f}$ & $\mathtt{p}$ & $\mathtt{g}$ & $\mathtt{i}$ & $\mathtt{n}$ & $\mathtt{h}$ & $\mathtt{d}$ & $\ovariant$ & $0$ & $\mathtt{k}$ & $\mathtt{b}$ & $\mathtt{s}$ & $\mathtt{t}$ \\ \hline
$\mathtt{m}$ & $\mathtt{m}$ & $\mathtt{v}$ & $\mathtt{x}$ & $\mathtt{r}$ & $\mathtt{j}$ & $\mathtt{i}$ & $\mathtt{n}$ & $\ovariant$ & $0$ & $\mathtt{k}$ & $\mathtt{l}$ & $\mathtt{q}$ & $\mathtt{p}$ & $\mathtt{c}$ & $\mathtt{s}$ & $\mathtt{u}$ & $\mathtt{w}$ & $\mathtt{t}$ & $\mathtt{h}$ & $\mathtt{a}$ & $\mathtt{e}$ & $\mathtt{b}$ & $\mathtt{f}$ & $\mathtt{d}$ & $\mathtt{g}$ \\ \hline
$\mathtt{n}$ & $\mathtt{n}$ & $\mathtt{q}$ & $0$ & $\mathtt{a}$ & $\mathtt{b}$ & $\mathtt{w}$ & $\mathtt{e}$ & $\mathtt{c}$ & $\mathtt{f}$ & $\mathtt{d}$ & $\mathtt{x}$ & $\mathtt{j}$ & $\mathtt{g}$ & $\mathtt{s}$ & $\mathtt{i}$ & $\mathtt{r}$ & $\ovariant$ & $\mathtt{l}$ & $\mathtt{v}$ & $\mathtt{k}$ & $\mathtt{p}$ & $\mathtt{h}$ & $\mathtt{t}$ & $\mathtt{u}$ & $\mathtt{m}$ \\ \hline
$\ovariant$ & $\ovariant$ & $0$ & $\mathtt{p}$ & $\mathtt{b}$ & $\mathtt{t}$ & $\mathtt{a}$ & $\mathtt{c}$ & $\mathtt{d}$ & $\mathtt{g}$ & $\mathtt{v}$ & $\mathtt{e}$ & $\mathtt{f}$ & $\mathtt{i}$ & $\mathtt{u}$ & $\mathtt{r}$ & $\mathtt{j}$ & $\mathtt{k}$ & $\mathtt{n}$ & $\mathtt{x}$ & $\mathtt{h}$ & $\mathtt{s}$ & $\mathtt{l}$ & $\mathtt{m}$ & $\mathtt{q}$ & $\mathtt{w}$ \\ \hline
$\mathtt{p}$ & $\mathtt{p}$ & $\mathtt{b}$ & $\mathtt{t}$ & $\mathtt{d}$ & $\mathtt{v}$ & $\mathtt{c}$ & $\mathtt{g}$ & $\mathtt{i}$ & $\mathtt{l}$ & $\mathtt{r}$ & $\mathtt{f}$ & $\mathtt{h}$ & $\mathtt{n}$ & $\mathtt{w}$ & $\ovariant$ & $\mathtt{k}$ & $\mathtt{s}$ & $0$ & $\mathtt{j}$ & $\mathtt{u}$ & $\mathtt{m}$ & $\mathtt{q}$ & $\mathtt{x}$ & $\mathtt{a}$ & $\mathtt{e}$ \\ \hline
$\mathtt{q}$ & $\mathtt{q}$ & $\mathtt{w}$ & $\mathtt{a}$ & $\mathtt{e}$ & $\mathtt{c}$ & $\mathtt{x}$ & $\mathtt{j}$ & $\mathtt{f}$ & $\mathtt{k}$ & $\mathtt{g}$ & $\mathtt{r}$ & $\ovariant$ & $\mathtt{h}$ & $\mathtt{t}$ & $\mathtt{l}$ & $\mathtt{n}$ & $0$ & $\mathtt{u}$ & $\mathtt{i}$ & $\mathtt{p}$ & $\mathtt{b}$ & $\mathtt{s}$ & $\mathtt{d}$ & $\mathtt{m}$ & $\mathtt{v}$ \\ \hline
$\mathtt{r}$ & $\mathtt{r}$ & $\mathtt{n}$ & $\ovariant$ & $0$ & $\mathtt{p}$ & $\mathtt{q}$ & $\mathtt{a}$ & $\mathtt{b}$ & $\mathtt{c}$ & $\mathtt{t}$ & $\mathtt{w}$ & $\mathtt{e}$ & $\mathtt{d}$ & $\mathtt{h}$ & $\mathtt{v}$ & $\mathtt{x}$ & $\mathtt{j}$ & $\mathtt{i}$ & $\mathtt{m}$ & $\mathtt{f}$ & $\mathtt{k}$ & $\mathtt{g}$ & $\mathtt{s}$ & $\mathtt{l}$ & $\mathtt{u}$ \\ \hline
$\mathtt{s}$ & $\mathtt{s}$ & $\mathtt{t}$ & $\mathtt{m}$ & $\mathtt{v}$ & $\mathtt{x}$ & $\mathtt{d}$ & $\mathtt{i}$ & $\mathtt{r}$ & $\mathtt{n}$ & $\mathtt{j}$ & $\mathtt{g}$ & $\mathtt{l}$ & $\ovariant$ & $\mathtt{a}$ & $\mathtt{k}$ & $\mathtt{h}$ & $\mathtt{u}$ & $\mathtt{p}$ & $\mathtt{f}$ & $\mathtt{q}$ & $\mathtt{w}$ & $0$ & $\mathtt{e}$ & $\mathtt{b}$ & $\mathtt{c}$ \\ \hline
$\mathtt{t}$ & $\mathtt{t}$ & $\mathtt{d}$ & $\mathtt{v}$ & $\mathtt{i}$ & $\mathtt{r}$ & $\mathtt{g}$ & $\mathtt{l}$ & $\mathtt{n}$ & $\mathtt{q}$ & $\ovariant$ & $\mathtt{h}$ & $\mathtt{u}$ & $0$ & $\mathtt{e}$ & $\mathtt{p}$ & $\mathtt{s}$ & $\mathtt{m}$ & $\mathtt{b}$ & $\mathtt{k}$ & $\mathtt{w}$ & $\mathtt{x}$ & $\mathtt{a}$ & $\mathtt{j}$ & $\mathtt{c}$ & $\mathtt{f}$ \\ \hline
$\mathtt{u}$ & $\mathtt{u}$ & $\mathtt{m}$ & $\mathtt{w}$ & $\mathtt{x}$ & $\mathtt{e}$ & $\mathtt{v}$ & $\mathtt{r}$ & $\mathtt{j}$ & $\ovariant$ & $\mathtt{f}$ & $\mathtt{i}$ & $\mathtt{n}$ & $\mathtt{k}$ & $\mathtt{b}$ & $\mathtt{h}$ & $\mathtt{l}$ & $\mathtt{q}$ & $\mathtt{s}$ & $\mathtt{g}$ & $0$ & $\mathtt{a}$ & $\mathtt{p}$ & $\mathtt{c}$ & $\mathtt{t}$ & $\mathtt{d}$ \\ \hline
$\mathtt{v}$ & $\mathtt{v}$ & $\mathtt{i}$ & $\mathtt{r}$ & $\mathtt{n}$ & $\ovariant$ & $\mathtt{l}$ & $\mathtt{q}$ & $0$ & $\mathtt{a}$ & $\mathtt{p}$ & $\mathtt{u}$ & $\mathtt{w}$ & $\mathtt{b}$ & $\mathtt{f}$ & $\mathtt{t}$ & $\mathtt{m}$ & $\mathtt{x}$ & $\mathtt{d}$ & $\mathtt{s}$ & $\mathtt{e}$ & $\mathtt{j}$ & $\mathtt{c}$ & $\mathtt{k}$ & $\mathtt{g}$ & $\mathtt{h}$ \\ \hline
$\mathtt{w}$ & $\mathtt{w}$ & $\mathtt{x}$ & $\mathtt{e}$ & $\mathtt{j}$ & $\mathtt{f}$ & $\mathtt{r}$ & $\ovariant$ & $\mathtt{k}$ & $\mathtt{p}$ & $\mathtt{h}$ & $\mathtt{n}$ & $0$ & $\mathtt{s}$ & $\mathtt{d}$ & $\mathtt{u}$ & $\mathtt{q}$ & $\mathtt{a}$ & $\mathtt{m}$ & $\mathtt{l}$ & $\mathtt{b}$ & $\mathtt{c}$ & $\mathtt{t}$ & $\mathtt{g}$ & $\mathtt{v}$ & $\mathtt{i}$ \\ \hline
$\mathtt{x}$ & $\mathtt{x}$ & $\mathtt{r}$ & $\mathtt{j}$ & $\ovariant$ & $\mathtt{k}$ & $\mathtt{n}$ & $0$ & $\mathtt{p}$ & $\mathtt{b}$ & $\mathtt{s}$ & $\mathtt{q}$ & $\mathtt{a}$ & $\mathtt{t}$ & $\mathtt{g}$ & $\mathtt{m}$ & $\mathtt{w}$ & $\mathtt{e}$ & $\mathtt{v}$ & $\mathtt{u}$ & $\mathtt{c}$ & $\mathtt{f}$ & $\mathtt{d}$ & $\mathtt{h}$ & $\mathtt{i}$ & $\mathtt{l}$ \\  
\end{tabular}
}

\end{table}

\begin{table}[H]
\caption{Multiplication table for the ring $I_5$}  \label{mult table}
\centering
\scalebox{0.9}{
\begin{tabular}{c ! {\vrule width 1.5pt} c|c|c|c|c|c|c|c|c|c|c|c|c|c|c|c|c|c|c|c|c|c|c|c|c}
$\times$ & $0$ & $\mathtt{a}$ & $\mathtt{b}$ & $\mathtt{c}$ & $\mathtt{d}$ & $\mathtt{e}$ & $\mathtt{f}$ & $\mathtt{g}$ & $\mathtt{h}$ & $\mathtt{i}$ & $\mathtt{j}$ & $\mathtt{k}$ & $\mathtt{l}$ & $\mathtt{m}$ & $\mathtt{n}$ & $\ovariant$ & $\mathtt{p}$ & $\mathtt{q}$ & $\mathtt{r}$ & $\mathtt{s}$ & $\mathtt{t}$ & $\mathtt{u}$ & $\mathtt{v}$ & $\mathtt{w}$ & $\mathtt{x}$ \\
\ChangeRT{1.5pt}
$0$      & $0$ & $0$ & $0$ & $0$ & $0$ & $0$ & $0$ & $0$ & $0$ & $0$ & $0$ & $0$ & $0$ & $0$ & $0$ & $0$ & $0$ & $0$ & $0$ & $0$ & $0$ & $0$ & $0$ & $0$ & $0$ \\ \hline
$\mathtt{a}$      & $0$ & $\mathtt{b}$ & $0$ & $\mathtt{b}$ & $0$ & $\mathtt{d}$ & $\mathtt{d}$ & $\mathtt{b}$ & $\mathtt{d}$ & $0$ & $\mathtt{i}$ & $\mathtt{i}$ & $\mathtt{b}$ & $\mathtt{i}$ & $0$ & $\mathtt{n}$ & $\mathtt{n}$ & $\mathtt{b}$ & $\mathtt{n}$ & $\mathtt{i}$ & $\mathtt{n}$ & $\mathtt{d}$ & $\mathtt{n}$ & $\mathtt{d}$ & $\mathtt{i}$ \\ \hline
$\mathtt{b}$      & $0$ & $0$ & $0$ & $0$ & $0$ & $0$ & $0$ & $0$ & $0$ & $0$ & $0$ & $0$ & $0$ & $0$ & $0$ & $0$ & $0$ & $0$ & $0$ & $0$ & $0$ & $0$ & $0$ & $0$ & $0$ \\ \hline
$\mathtt{c}$      & $0$ & $\mathtt{b}$ & $0$ & $\mathtt{b}$ & $0$ & $\mathtt{d}$ & $\mathtt{d}$ & $\mathtt{b}$ & $\mathtt{d}$ & $0$ & $\mathtt{i}$ & $\mathtt{i}$ & $\mathtt{b}$ & $\mathtt{i}$ & $0$ & $\mathtt{n}$ & $\mathtt{n}$ & $\mathtt{b}$ & $\mathtt{n}$ & $\mathtt{i}$ & $\mathtt{n}$ & $\mathtt{d}$ & $\mathtt{n}$ & $\mathtt{d}$ & $\mathtt{i}$ \\ \hline
$\mathtt{d}$      & $0$ & $0$ & $0$ & $0$ & $0$ & $0$ & $0$ & $0$ & $0$ & $0$ & $0$ & $0$ & $0$ & $0$ & $0$ & $0$ & $0$ & $0$ & $0$ & $0$ & $0$ & $0$ & $0$ & $0$ & $0$ \\ \hline
$\mathtt{e}$      & $0$ & $\mathtt{d}$ & $0$ & $\mathtt{d}$ & $0$ & $\mathtt{n}$ & $\mathtt{n}$ & $\mathtt{d}$ & $\mathtt{n}$ & $0$ & $\mathtt{b}$ & $\mathtt{b}$ & $\mathtt{d}$ & $\mathtt{b}$ & $0$ & $\mathtt{i}$ & $\mathtt{i}$ & $\mathtt{d}$ & $\mathtt{i}$ & $\mathtt{b}$ & $\mathtt{i}$ & $\mathtt{n}$ & $\mathtt{i}$ & $\mathtt{n}$ & $\mathtt{b}$ \\ \hline
$\mathtt{f}$      & $0$ & $\mathtt{d}$ & $0$ & $\mathtt{d}$ & $0$ & $\mathtt{n}$ & $\mathtt{n}$ & $\mathtt{d}$ & $\mathtt{n}$ & $0$ & $\mathtt{b}$ & $\mathtt{b}$ & $\mathtt{d}$ & $\mathtt{b}$ & $0$ & $\mathtt{i}$ & $\mathtt{i}$ & $\mathtt{d}$ & $\mathtt{i}$ & $\mathtt{b}$ & $\mathtt{i}$ & $\mathtt{n}$ & $\mathtt{i}$ & $\mathtt{n}$ & $\mathtt{b}$ \\ \hline
$\mathtt{g}$      & $0$ & $\mathtt{b}$ & $0$ & $\mathtt{b}$ & $0$ & $\mathtt{d}$ & $\mathtt{d}$ & $\mathtt{b}$ & $\mathtt{d}$ & $0$ & $\mathtt{i}$ & $\mathtt{i}$ & $\mathtt{b}$ & $\mathtt{i}$ & $0$ & $\mathtt{n}$ & $\mathtt{n}$ & $\mathtt{b}$ & $\mathtt{n}$ & $\mathtt{i}$ & $\mathtt{n}$ & $\mathtt{d}$ & $\mathtt{n}$ & $\mathtt{d}$ & $\mathtt{i}$ \\ \hline
$\mathtt{h}$      & $0$ & $\mathtt{d}$ & $0$ & $\mathtt{d}$ & $0$ & $\mathtt{n}$ & $\mathtt{n}$ & $\mathtt{d}$ & $\mathtt{n}$ & $0$ & $\mathtt{b}$ & $\mathtt{b}$ & $\mathtt{d}$ & $\mathtt{b}$ & $0$ & $\mathtt{i}$ & $\mathtt{i}$ & $\mathtt{d}$ & $\mathtt{i}$ & $\mathtt{b}$ & $\mathtt{i}$ & $\mathtt{n}$ & $\mathtt{i}$ & $\mathtt{n}$ & $\mathtt{b}$ \\ \hline
$\mathtt{i}$      & $0$ & $0$ & $0$ & $0$ & $0$ & $0$ & $0$ & $0$ & $0$ & $0$ & $0$ & $0$ & $0$ & $0$ & $0$ & $0$ & $0$ & $0$ & $0$ & $0$ & $0$ & $0$ & $0$ & $0$ & $0$ \\ \hline
$\mathtt{j}$      & $0$ & $\mathtt{i}$ & $0$ & $\mathtt{i}$ & $0$ & $\mathtt{b}$ & $\mathtt{b}$ & $\mathtt{i}$ & $\mathtt{b}$ & $0$ & $\mathtt{n}$ & $\mathtt{n}$ & $\mathtt{i}$ & $\mathtt{n}$ & $0$ & $\mathtt{d}$ & $\mathtt{d}$ & $\mathtt{i}$ & $\mathtt{d}$ & $\mathtt{n}$ & $\mathtt{d}$ & $\mathtt{b}$ & $\mathtt{d}$ & $\mathtt{b}$ & $\mathtt{n}$ \\ \hline
$\mathtt{k}$      & $0$ & $\mathtt{i}$ & $0$ & $\mathtt{i}$ & $0$ & $\mathtt{b}$ & $\mathtt{b}$ & $\mathtt{i}$ & $\mathtt{b}$ & $0$ & $\mathtt{n}$ & $\mathtt{n}$ & $\mathtt{i}$ & $\mathtt{n}$ & $0$ & $\mathtt{d}$ & $\mathtt{d}$ & $\mathtt{i}$ & $\mathtt{d}$ & $\mathtt{n}$ & $\mathtt{d}$ & $\mathtt{b}$ & $\mathtt{d}$ & $\mathtt{b}$ & $\mathtt{n}$ \\ \hline
$\mathtt{l}$      & $0$ & $\mathtt{b}$ & $0$ & $\mathtt{b}$ & $0$ & $\mathtt{d}$ & $\mathtt{d}$ & $\mathtt{b}$ & $\mathtt{d}$ & $0$ & $\mathtt{i}$ & $\mathtt{i}$ & $\mathtt{b}$ & $\mathtt{i}$ & $0$ & $\mathtt{n}$ & $\mathtt{n}$ & $\mathtt{b}$ & $\mathtt{n}$ & $\mathtt{i}$ & $\mathtt{n}$ & $\mathtt{d}$ & $\mathtt{n}$ & $\mathtt{d}$ & $\mathtt{i}$ \\ \hline
$\mathtt{m}$      & $0$ & $\mathtt{i}$ & $0$ & $\mathtt{i}$ & $0$ & $\mathtt{b}$ & $\mathtt{b}$ & $\mathtt{i}$ & $\mathtt{b}$ & $0$ & $\mathtt{n}$ & $\mathtt{n}$ & $\mathtt{i}$ & $\mathtt{n}$ & $0$ & $\mathtt{d}$ & $\mathtt{d}$ & $\mathtt{i}$ & $\mathtt{d}$ & $\mathtt{n}$ & $\mathtt{d}$ & $\mathtt{b}$ & $\mathtt{d}$ & $\mathtt{b}$ & $\mathtt{n}$ \\ \hline
$\mathtt{n}$      & $0$ & $0$ & $0$ & $0$ & $0$ & $0$ & $0$ & $0$ & $0$ & $0$ & $0$ & $0$ & $0$ & $0$ & $0$ & $0$ & $0$ & $0$ & $0$ & $0$ & $0$ & $0$ & $0$ & $0$ & $0$ \\ \hline
$\ovariant$      & $0$ & $\mathtt{n}$ & $0$ & $\mathtt{n}$ & $0$ & $\mathtt{i}$ & $\mathtt{i}$ & $\mathtt{n}$ & $\mathtt{i}$ & $0$ & $\mathtt{d}$ & $\mathtt{d}$ & $\mathtt{n}$ & $\mathtt{d}$ & $0$ & $\mathtt{b}$ & $\mathtt{b}$ & $\mathtt{n}$ & $\mathtt{b}$ & $\mathtt{d}$ & $\mathtt{b}$ & $\mathtt{i}$ & $\mathtt{b}$ & $\mathtt{i}$ & $\mathtt{d}$ \\ \hline
$\mathtt{p}$      & $0$ & $\mathtt{n}$ & $0$ & $\mathtt{n}$ & $0$ & $\mathtt{i}$ & $\mathtt{i}$ & $\mathtt{n}$ & $\mathtt{i}$ & $0$ & $\mathtt{d}$ & $\mathtt{d}$ & $\mathtt{n}$ & $\mathtt{d}$ & $0$ & $\mathtt{b}$ & $\mathtt{b}$ & $\mathtt{n}$ & $\mathtt{b}$ & $\mathtt{d}$ & $\mathtt{b}$ & $\mathtt{i}$ & $\mathtt{b}$ & $\mathtt{i}$ & $\mathtt{d}$ \\ \hline
$\mathtt{q}$      & $0$ & $\mathtt{b}$ & $0$ & $\mathtt{b}$ & $0$ & $\mathtt{d}$ & $\mathtt{d}$ & $\mathtt{b}$ & $\mathtt{d}$ & $0$ & $\mathtt{i}$ & $\mathtt{i}$ & $\mathtt{b}$ & $\mathtt{i}$ & $0$ & $\mathtt{n}$ & $\mathtt{n}$ & $\mathtt{b}$ & $\mathtt{n}$ & $\mathtt{i}$ & $\mathtt{n}$ & $\mathtt{d}$ & $\mathtt{n}$ & $\mathtt{d}$ & $\mathtt{i}$ \\ \hline
$\mathtt{r}$      & $0$ & $\mathtt{n}$ & $0$ & $\mathtt{n}$ & $0$ & $\mathtt{i}$ & $\mathtt{i}$ & $\mathtt{n}$ & $\mathtt{i}$ & $0$ & $\mathtt{d}$ & $\mathtt{d}$ & $\mathtt{n}$ & $\mathtt{d}$ & $0$ & $\mathtt{b}$ & $\mathtt{b}$ & $\mathtt{n}$ & $\mathtt{b}$ & $\mathtt{d}$ & $\mathtt{b}$ & $\mathtt{i}$ & $\mathtt{b}$ & $\mathtt{i}$ & $\mathtt{d}$ \\ \hline
$\mathtt{s}$      & $0$ & $\mathtt{i}$ & $0$ & $\mathtt{i}$ & $0$ & $\mathtt{b}$ & $\mathtt{b}$ & $\mathtt{i}$ & $\mathtt{b}$ & $0$ & $\mathtt{n}$ & $\mathtt{n}$ & $\mathtt{i}$ & $\mathtt{n}$ & $0$ & $\mathtt{d}$ & $\mathtt{d}$ & $\mathtt{i}$ & $\mathtt{d}$ & $\mathtt{n}$ & $\mathtt{d}$ & $\mathtt{b}$ & $\mathtt{d}$ & $\mathtt{b}$ & $\mathtt{n}$ \\ \hline
$\mathtt{t}$      & $0$ & $\mathtt{n}$ & $0$ & $\mathtt{n}$ & $0$ & $\mathtt{i}$ & $\mathtt{i}$ & $\mathtt{n}$ & $\mathtt{i}$ & $0$ & $\mathtt{d}$ & $\mathtt{d}$ & $\mathtt{n}$ & $\mathtt{d}$ & $0$ & $\mathtt{b}$ & $\mathtt{b}$ & $\mathtt{n}$ & $\mathtt{b}$ & $\mathtt{d}$ & $\mathtt{b}$ & $\mathtt{i}$ & $\mathtt{b}$ & $\mathtt{i}$ & $\mathtt{d}$ \\ \hline
$\mathtt{u}$      & $0$ & $\mathtt{d}$ & $0$ & $\mathtt{d}$ & $0$ & $\mathtt{n}$ & $\mathtt{n}$ & $\mathtt{d}$ & $\mathtt{n}$ & $0$ & $\mathtt{b}$ & $\mathtt{b}$ & $\mathtt{d}$ & $\mathtt{b}$ & $0$ & $\mathtt{i}$ & $\mathtt{i}$ & $\mathtt{d}$ & $\mathtt{i}$ & $\mathtt{b}$ & $\mathtt{i}$ & $\mathtt{n}$ & $\mathtt{i}$ & $\mathtt{n}$ & $\mathtt{b}$ \\ \hline
$\mathtt{v}$      & $0$ & $\mathtt{n}$ & $0$ & $\mathtt{n}$ & $0$ & $\mathtt{i}$ & $\mathtt{i}$ & $\mathtt{n}$ & $\mathtt{i}$ & $0$ & $\mathtt{d}$ & $\mathtt{d}$ & $\mathtt{n}$ & $\mathtt{d}$ & $0$ & $\mathtt{b}$ & $\mathtt{b}$ & $\mathtt{n}$ & $\mathtt{b}$ & $\mathtt{d}$ & $\mathtt{b}$ & $\mathtt{i}$ & $\mathtt{b}$ & $\mathtt{i}$ & $\mathtt{d}$ \\ \hline
$\mathtt{w}$      & $0$ & $\mathtt{d}$ & $0$ & $\mathtt{d}$ & $0$ & $\mathtt{n}$ & $\mathtt{n}$ & $\mathtt{d}$ & $\mathtt{n}$ & $0$ & $\mathtt{b}$ & $\mathtt{b}$ & $\mathtt{d}$ & $\mathtt{b}$ & $0$ & $\mathtt{i}$ & $\mathtt{i}$ & $\mathtt{d}$ & $\mathtt{i}$ & $\mathtt{b}$ & $\mathtt{i}$ & $\mathtt{n}$ & $\mathtt{i}$ & $\mathtt{n}$ & $\mathtt{b}$ \\ \hline
$\mathtt{x}$      & $0$ & $\mathtt{i}$ & $0$ & $\mathtt{i}$ & $0$ & $\mathtt{b}$ & $\mathtt{b}$ & $\mathtt{i}$ & $\mathtt{b}$ & $0$ & $\mathtt{n}$ & $\mathtt{n}$ & $\mathtt{i}$ & $\mathtt{n}$ & $0$ & $\mathtt{d}$ & $\mathtt{d}$ & $\mathtt{i}$ & $\mathtt{d}$ & $\mathtt{n}$ & $\mathtt{d}$ & $\mathtt{b}$ & $\mathtt{d}$ & $\mathtt{b}$ & $\mathtt{n}$ \\
\end{tabular}
}

\end{table}

We can infer from Table \ref{mult table} that this ring is commutative without unity and has a unique maximal ideal $J = \{ j  \mathtt{b}: 0 \leq j < 5 \}  = \{0 ,  \mathtt{b}   , \mathtt{d} , \mathtt{i} , \mathtt{n}   \}$ with residue field $I_{5} / J \cong \mathbb{F}_{5}$. With this, we can write $I_5$ as
\begin{align*}
I_5 = \{  \mathtt{a} x + \mathtt{b} y \mid  x , y   \in \mathbb{F}_5    \} .
\end{align*}

Define a natural action of $\mathbb{F}_5$ on the ring $I_5$ using the rule
\begin{gather*}
r 0 = 0 r = 0 ,  \quad   r 1 = 1 r = r ,  \quad   r 2 = r + r = 2r ,  \\
r 3 = r + r + r = 3 r , \quad  r 4 = r + r + r + r  = 4 r
\end{gather*}
for all $r \in I_5$. Note that for all $r \in I_{5}, x, y \in \mathbb{F}_{5}$, this action is \textit{distributive} in the sense that $r(x +_5 y)=r x+r y$, where $+_5$ denotes the addition in $\mathbb{F}_{5}$. Given $\mathbf{x} = ( x_1 , \dotsc , x_n ) \in \mathbb{F}_{5}^{n}$ and $\mathbf{r} = ( r_1 , \dotsc , r_n ) \in I_{5}^{n}$, we will be flexible in the use of  inner product notation $(\mathbf{x}, \mathbf{r})$ to indicate
\begin{align*}
(\mathbf{x}, \mathbf{r})  \coloneqq   x_{1} r_{1}+\cdots+x_{n} r_{n}  .
\end{align*}

We define the reduction map modulo $J$ as $\pi \colon I_{5} \longrightarrow I_{5} / J \simeq \mathbb{F}_{5}$ by
\begin{alignat*}{6}
&\pi (0) &&= \pi ( \mathtt{b} ) &&= \pi ( \mathtt{d} ) &&= \pi (  \mathtt{i}  ) &&= \pi (  \mathtt{n}  )  &&=  0  ,   \\
&\pi ( \mathtt{a}  ) &&= \pi (  \mathtt{c}  ) &&= \pi ( \mathtt{g}  ) &&= \pi ( \mathtt{l}  ) &&= \pi (  \mathtt{q}  )  &&=  1  ,    \\
&\pi ( \mathtt{e}  ) &&= \pi (   \mathtt{f}   ) &&= \pi (  \mathtt{h}  ) &&= \pi (  \mathtt{u}  ) &&= \pi (  \mathtt{w}  )  &&=  2  ,    \\
&\pi (  \mathtt{j}  ) &&= \pi ( \mathtt{k}  ) &&= \pi ( \mathtt{m}  ) &&= \pi ( \mathtt{s} ) &&= \pi ( \mathtt{x}  )  &&=  3  ,    \\
&\pi ( \ovariant  ) &&= \pi ( \mathtt{p}  ) &&= \pi ( \mathtt{r} ) &&= \pi ( \mathtt{t}  ) &&= \pi ( \mathtt{v}  )  &&=  4  .
\end{alignat*}

The map $\pi$ is extended in the natural way to a map from $I_{5}^{n}$ to $\mathbb{F}_{5}^{n}$.

\subsection{Codes over {\boldmath $I_5$} }

A linear code $\mathcal{C}$ over $I_5$ (or simply called an $I_5$-code) of length $n$ is defined as an $I_{5}$-submodule of $I_{5}^{n}$. It can be regarded as the $I_{5}$-span of the rows of a so-called generator matrix. Two  codes over $\mathbb{F}_5$ of length $n$ can be associated canonically with the $I_5$-code $\mathcal{C}$. The residue code $\operatorname{res}(\mathcal{C})$ is defined as
\begin{align*}
\operatorname{res} (\mathcal{C})  \coloneqq  \left\lbrace   \pi ( \mathbf{y} )  \mid  \mathbf{y}  \in \mathcal{C}  \right \rbrace  =   \pi (\mathcal{C} )
\end{align*}
and the torsion code  $\operatorname{tor}(\mathcal{C})$  is defined as
\begin{align*}
\operatorname{tor}(\mathcal{C}) \coloneqq \left\{\mathbf{x} \in \mathbb{F}_{5}^{n}  \mid   \mathtt{b}  \mathbf{x} \in \mathcal{C}  \right\}  .
\end{align*}

Note that $\operatorname{res}(\mathcal{C})$ and $\operatorname{tor}(\mathcal{C})$ are both linear codes over $\mathbb{F}_5$ and $\operatorname{res}(\mathcal{C}) \subseteq \operatorname{tor}(\mathcal{C})$. An $I_5$-code $\mathcal{C}$ is said to be of type $\left \{  k_{1}, k_{2}  \right  \}$ if $\operatorname{res} (\mathcal{C})$ has dimension $k_1$ and $\operatorname{tor} (\mathcal{C})$ has dimension $k_1 + k_2$. We say that an $I_5$-code $\mathcal{C}$ is free if and only if $k_2 = 0$. Equivalently, an $I_5$-code $\mathcal{C}$ is free if and only if $\operatorname{res} (\mathcal{C}) = \operatorname{tor} (\mathcal{C})$. Let $\pi_\mathcal{C}$ be the restriction of $\pi$ to $\mathcal{C}$. Given an $I_5$-code $\mathcal{C}$ of type $\{ k_1 , k_2\}$, the first isomorphism theorem applied to $\pi_\mathcal{C}$ gives
$$
|\mathcal{C}|=|\operatorname{res}(\mathcal{C})||\operatorname{tor}(\mathcal{C})|=5^{2 k_{1}+k_{2}} .
$$

Define the inner product of $\mathbf{x}=\left(x_{1}, \ldots, x_{n}\right)$, $\mathbf{y}=\left(y_{1}, \ldots, y_{n}\right) \in I_{5}^{n}$ as $(\mathbf{x}, \mathbf{y})   \coloneqq  \sum_{i=1}^{n} x_i y_i$. The dual code $\mathcal{C}^{\perp}$ of the code $\mathcal{C}$ is the module defined as

$$
\mathcal{C}^{\perp}=\left\{\mathbf{y} \in I_{5}^{n} \mid \forall \mathbf{x} \in \mathcal{C},(\mathbf{x}, \mathbf{y})=0\right\}  .
$$

An $I_{5}$-code $\mathcal{C}$ is self-orthogonal if for all $\mathbf{x}, \mathbf{y} \in \mathcal{C},(\mathbf{x}, \mathbf{y})=0$. If $\mathcal{C}=\mathcal{C}^{\perp}$, we say that $\mathcal{C}$ is self-dual. We adopt the terminology used in \cite{I3} to mean that an $I_{5}$-code is a quasi self-dual code of length $n$ if it is self-orthogonal and of size $5^{n}$.

Two $I_{5}$-codes $\mathcal{C}$ and $\mathcal{C}'$ are said to be monomially equivalent if there is an $n \times n$ monomial matrix $M$ (with exactly one nonzero entry in each row and column and the nonzero entries are in $\{1, -1 \}$) such that $\mathcal{C}' = \{ \mathbf{c} M \, : \, \mathbf{c} \in \mathcal{C}    \}$.

\begin{theorem}    \label{thm: c=au+bv for any c in C}
If $\mathcal{C}$ is a linear code of length $n$ over $I_5$, then the following hold:
\begin{enumerate}
\item  Every codeword $\mathbf{c} \in \mathcal{C}$ can be written as $\mathbf{c} = \mathtt{a} \mathbf{u} + \mathtt{b} \mathbf{v}$ for some $\mathbf{u} \in \operatorname{res}(\mathcal{C})$ and $\mathbf{v} \in \mathbb{F}_{5}^{n}$.  \label{itm: c=au+bv}

\item If $\mathbf{u} \in \operatorname{res}(\mathcal{C})$, then $\mathtt{a} \mathbf{u} + \mathtt{b} \mathbf{v}$ is a codeword in $\mathcal{C}$ for some $\mathbf{v} \in \mathbb{F}_{5}^{n}$.  \label{itm: au+bv in C}
\end{enumerate}
\end{theorem}

\begin{proof}
Let $\mathbf{c} \in \mathcal{C}$. We can write $\mathbf{c}$ in a $\mathtt{b}$-adic decomposition form as $\mathbf{c}=\mathtt{a} \mathbf{u} + \mathtt{b} \mathbf{v}$ where $\mathbf{u}, \mathbf{v} \in \mathbb{F}_{5}^{n}$. Since $\pi (\mathbf{c})=\pi (\mathtt{a} \mathbf{u}+\mathtt{b} \mathbf{v})=\mathbf{u}$, it follows that $\mathbf{u} \in \operatorname{res}(\mathcal{C})$. This proves the first part.

Now let $\mathbf{u} \in \operatorname{res}(\mathcal{C})$. Then there exists $\mathbf{c} \in \mathcal{C}$ such that $\pi(\mathbf{c})=\mathbf{u}$. We can write $\mathbf{c}$ in a $\mathtt{b}$-adic decomposition form as $\mathbf{c}=\mathtt{a} \mathbf{w}+\mathtt{b} \mathbf{v}$ where $\mathbf{w}, \mathbf{v} \in \mathbb{F}_{5}^{n}$. Observe that $\pi (\mathtt{a} \mathbf{w}+\mathtt{b} \mathbf{v})=\mathbf{w}$. On the other hand, $\pi(\mathtt{a} \mathbf{w}+\mathtt{b} \mathbf{v})=\pi(\mathbf{c})=\mathbf{u}$. Hence, $\mathbf{w}=\mathbf{u}$ and so $\mathtt{a} \mathbf{u}+\mathtt{b} \mathbf{v}$ is a codeword in $\mathcal{C}$. This proves the second part.
\end{proof}

\begin{theorem}  \label{thm: minimum distance of I5 code C}
If $\mathcal{C}$ is a non-trivial linear code over $I_5$, then the minimum distance of $\mathcal{C}$ equals the minimum distance of $\operatorname{tor}(\mathcal{C})$.
\end{theorem}

\begin{proof}
Let $d$ be the minimum distance of $\mathcal{C}$ and let $d_{t}$ be the minimum distance of $\operatorname{tor} (\mathcal{C})$. Then there exists a nonzero $\mathbf{t} \in \operatorname{tor} (\mathcal{C})$ such that $\mathrm{wt}(\mathbf{t})=d_{t}$. Since $\mathtt{b} \operatorname{tor} (\mathcal{C}) \subseteq \mathcal{C}$ and $\operatorname{wt} ( \mathtt{b} \mathbf{t}) = \operatorname{wt}(\mathbf{t})=d_{t}$, we have $d \leq d_{t}$.

Now, we prove that $d \geq d_{t}$. Let $\mathbf{x} \in \mathcal{C}$ such that $\mathrm{wt} (\mathbf{x})=d$. By Theorem \ref{thm: c=au+bv for any c in C}, $\mathbf{x} = \mathtt{a} \mathbf{u} + \mathtt{b} \mathbf{v}$ where $\mathbf{u} \in \operatorname{res}(\mathcal{C})$ and $\mathbf{v} \in \mathbb{F}_{5}^{n}$. Since $\mathcal{C}$ is a non-trivial linear code, we have the following three cases depending on $\mathbf{u}$ and $\mathbf{v}$ :
\begin{itemize}
  \item If $\mathbf{u} = \mathbf{0}$ and $\mathbf{v} \neq \mathbf{0}$, then we have $\mathbf{v} \in \operatorname{tor} (\mathcal{C})$ and $\mathrm{wt} (\mathbf{x}) = \mathrm{wt}( \mathtt{b} \mathbf{v} ) = \mathrm{wt} (\mathbf{v}) \geq d_{t}$.
  \item If $\mathbf{u} \neq \mathbf{0}$ and $\mathbf{v} = \mathbf{0}$, then $\mathrm{wt} (\mathbf{x}) =  \mathrm{wt}( \mathtt{a} \mathbf{u}) = \mathrm{wt} (\mathbf{u})$.
  \item If $\mathbf{u}, \mathbf{v} \neq \mathbf{0}$, then $\mathrm{wt} (\mathbf{x}) \geq \mathrm{wt} ( \mathtt{a} \mathbf{x}) = \mathrm{wt}( \mathtt{b} \mathbf{u}) = \mathrm{wt}(\mathbf{u})$.
\end{itemize}

Since $\mathbf{u} \in \operatorname{res}(\mathcal{C}) \subseteq \operatorname{tor}(\mathcal{C})$, it follows that $\operatorname{wt}(\mathbf{u}) \geq d_{t}$. Thus, in all cases, $d=\mathrm{wt}(\mathbf{x}) \geq d_{t}$.

Since $d \leq d_{t}$ and $d \geq d_{t}$, it follows that $d=d_{t}$.
\end{proof}

\begin{theorem}  \label{thm: residue and torsion of dual code Cperp}
If $\mathcal{C}$ is a  linear code over $I_5$, then $\operatorname{res} ( \mathcal{C}^{\perp} ) = \operatorname{res} ( \mathcal{C} )^{\perp}$ and $\operatorname{tor}  ( \mathcal{C}^{\perp} ) = \mathbb{F}_{5}^{n}$.
\end{theorem}

\begin{proof}
Let $\mathbf{u} \in \operatorname{res} (\mathcal{C}^{\perp} )$. By Theorem \ref{thm: c=au+bv for any c in C}, $\mathtt{a} \mathbf{u} + \mathtt{b} \mathbf{v}$ is a codeword in $\mathcal{C}^{\perp}$ for some $\mathbf{v} \in \mathbb{F}_{5}^{n}$. Let $\mathbf{x} \in \operatorname{res}(\mathcal{C})$. By Theorem \ref{thm: c=au+bv for any c in C}, $\mathtt{a} \mathbf{x} + \mathtt{b} \mathbf{y}$ is a codeword in $\mathcal{C}$ for some $\mathbf{y} \in \mathbb{F}_{5}^{n}$. By definition of $\mathcal{C}^{\perp}$,
\begin{align*}
0=(\mathtt{a} \mathbf{u} + \mathtt{b} \mathbf{v})  \cdot  ( \mathtt{a} \mathbf{x} + \mathtt{b} \mathbf{y}) = \mathtt{b}  (  \mathbf{u} \cdot \mathbf{x}  )  .
\end{align*}
Hence, $\mathbf{u} \cdot \mathbf{x}=0$ which implies that $\mathbf{u} \in \operatorname{res} (\mathcal{C})^{\perp}$. Therefore, $\operatorname{res} (\mathcal{C}^{\perp} ) \subseteq \operatorname{res}(\mathcal{C})^{\perp}$.

Now, assume $\mathbf{u} \in \operatorname{res}(\mathcal{C})^{\perp}$. Let $\mathbf{c} \in \mathcal{C}$. By Theorem \ref{thm: c=au+bv for any c in C}, $\mathbf{c}=$ $\mathtt{a} \mathbf{x} + \mathtt{b} \mathbf{y}$ where $\mathbf{x} \in \operatorname{res}(\mathcal{C})$ and $\mathbf{y} \in \mathbb{F}_{5}^{n}$. Observe that
\begin{align*}
\mathtt{a} \mathbf{u} \cdot \mathbf{c} = \mathtt{a} \mathbf{u} \cdot( \mathtt{a} \mathbf{x} + \mathtt{b} \mathbf{y}) = \mathtt{b} ( \mathbf{u} \cdot \mathbf{x} ) = 0  .
\end{align*}
Hence, $\mathtt{a} \mathbf{u} \in \mathcal{C}^{\perp}$ and $\pi( \mathtt{a} \mathbf{u}) = \mathbf{u}$ which yields $\mathbf{u} \in \operatorname{res} (\mathcal{C}^{\perp} )$. Therefore, $\operatorname{res}(\mathcal{C})^{\perp} \subseteq \operatorname{res}  (\mathcal{C}^{\perp} )$.

Now, we show that $\mathbb{F}_{5}^{n} \subseteq \operatorname{tor} (\mathcal{C}^{\perp} )$. Let $\mathbf{u} \in \mathbb{F}_{5}^{n}$ and let $\mathbf{c} \in \mathcal{C}$. By Theorem \ref{thm: c=au+bv for any c in C}, $\mathbf{c} = \mathtt{a} \mathbf{x} + \mathtt{b} \mathbf{y}$ where $\mathbf{x} \in \operatorname{res}(\mathcal{C})$ and $\mathbf{y} \in \mathbb{F}_{5}^{n}$. Observe that
\begin{align*}
\mathbf{c} \cdot \mathtt{b} \mathbf{u}=( \mathtt{a} \mathbf{x} + \mathtt{b} \mathbf{y}) \cdot \mathtt{b} \mathbf{u} = 0  .
\end{align*}
Hence, $\mathtt{b} \mathbf{u} \in \mathcal{C}^{\perp}$ which gives $\mathbf{u} \in \operatorname{tor}  (\mathcal{C}^{\perp} )$. Therefore, $\mathbb{F}_{5}^{n} = \operatorname{tor} (\mathcal{C}^{\perp} )$.
\end{proof}

\begin{theorem}[Theorem 2 in \cite{MassFormula_Ip}]  \label{thm: dual of a code over I5}
If $\mathcal{C}$ is a  linear code over $I_5$, then $\mathcal{C}^\perp = \mathtt{a}  \operatorname{res} (\mathcal{C})^\perp  \oplus  \mathtt{b} \mathbb{F}_5^n$.
\end{theorem}

The following lemma is a characteristic $5$ version of Lemma 1 in \cite{I3}.
\begin{lemma}
Let $\mathcal{C}$ be a linear code over $I_{5}$ of type $\left\{k_{1}, k_{2}\right\}$ and length $n$. Then a generator matrix $G$ of the code is of the form
$$
\begin{bmatrix*}[c]
\mathtt{a} I_{k_{1}} & \mathtt{a} X & Y \\
0 & \mathtt{b} I_{k_{2}} & \mathtt{b} Z
\end{bmatrix*}
$$
where $Y$ is a matrix with entries in $I_{5}, X$ and $Z$ are matrices with entries in $\mathbb{F}_{5}$, and $I_{k_{1}}, I_{k_{2}}$ are identity matrices.
\end{lemma}

\begin{theorem}\label{thm:multilevel_generator_matrix_I5}
Let $\mathcal{C}_1$ and $\mathcal{C}_2$ be linear codes over $\mathbb{F}_5$ such that
$\mathcal{C}_1 \subseteq \mathcal{C}_2$, and let $G_1$ and $G_2$ be generator matrices of
$\mathcal{C}_1$ and $\mathcal{C}_2$, respectively. Define
\begin{align*}
G=
\begin{pmatrix}
\mathtt a G_1\\
\mathtt b G_2
\end{pmatrix}.
\end{align*}
Let $\mathcal{C}$ be the $I_5$-submodule of $I_5^n$ generated by the rows of $G$, that is,
the smallest $I_5$-submodule of $I_5^n$ containing these rows. Then $\mathcal{C}$ is a linear code over $I_5$ of length $n$, and
\begin{align*}
\mathcal{C}
=
\mathtt a\mathcal{C}_1+\mathtt b\mathcal{C}_2
=
\{ \mathtt a\mathbf{u}+\mathtt b\mathbf{v} \,  :  \,  \mathbf{u}\in\mathcal{C}_1, \, 
\mathbf{v}\in\mathcal{C}_2 \}.
\end{align*}
In particular, for the $I_5$-code $\mathcal{C}$ constructed above, we have $\operatorname{res} (\mathcal{C}) = \mathcal{C}_1$ and $\operatorname{tor}(\mathcal{C}) = \mathcal{C}_2$. Moreover, the rows of $G$ form a generating set of $\mathcal{C}$ as an $I_5$-module, but
are not necessarily $I_5$-linearly independent.
\end{theorem}

\begin{proof}
Let $\mathcal{C}$ be the smallest $I_5$-submodule of $I_5^n$ containing the rows of $G$.
Since the rows of $G$ are precisely the vectors $\mathtt a\mathbf{u}$ with $\mathbf{u}$ a row of $G_1$
and the vectors $\mathtt b\mathbf{v}$ with $\mathbf{v}$ a row of $G_2$, we have
$\mathtt a\mathcal{C}_1\subseteq \mathcal{C}$ and $\mathtt b\mathcal{C}_2\subseteq \mathcal{C}$.
Hence, $\mathtt a \mathcal{C}_1 + \mathtt b \mathcal{C}_2 \subseteq \mathcal{C}$. Conversely, let $r \in I_5$ and $\mathbf{u} \in \mathcal{C}_1$. Writing $r = \mathtt{a} x + \mathtt b y$ with
$x , y \in \mathbb{F}_5$, we have
\begin{align*}
r (\mathtt{a}  \mathbf{u}) = ( \mathtt{a} x + \mathtt{b} y )\mathtt{a} \mathbf{u} = \mathtt{a}^2 x \mathbf{u} = \mathtt{b}  x  \mathbf{u} \in \mathtt{b} \mathcal{C}_2,
\end{align*}
because $\mathbf{u}\in \mathcal{C}_1\subseteq \mathcal{C}_2$.
Also, for every $\mathbf{v}\in \mathcal{C}_2$, we have $r(\mathtt b\mathbf{v})=\mathbf{0}$.
Therefore $\mathcal{C}\subseteq \mathtt a\mathcal{C}_1+\mathtt b\mathcal{C}_2$.
Thus, $\mathcal{C} = \mathtt a\mathcal{C}_1+\mathtt b\mathcal{C}_2$.

Finally, for any $\mathbf{c}=\mathtt a\mathbf{u}+\mathtt b\mathbf{v}\in\mathcal{C}$,
we have $\pi(\mathbf{c})=\mathbf{u}\in\mathcal{C}_1$, which shows
$\operatorname{res}(\mathcal{C})\subseteq\mathcal{C}_1$.
Since $\mathtt a\mathbf{u}\in\mathcal{C}$ for all $\mathbf{u}\in\mathcal{C}_1$,
the reverse inclusion holds and $\operatorname{res}(\mathcal{C})=\mathcal{C}_1$.
Similarly, by the definition of torsion and the equality
$\mathcal{C}=\mathtt a\mathcal{C}_1+\mathtt b\mathcal{C}_2$, we obtain
$\operatorname{tor}(\mathcal{C})=\mathcal{C}_2$.

The final statement follows from the fact that the rows of $G$ generate
$\mathcal{C}$ as an $I_5$-module, although they may be $I_5$-linearly dependent.
\end{proof}

\subsection{Codes over {\boldmath $\mathbb{F}_{25}$}}

An additive code $\mathcal{C}$ of length $n$ over $\mathbb{F}_{25}$ is defined as an $\mathbb{F}_5$-additive subgroup of $\mathbb{F}_{25}^{n}$. Therefore, $\mathcal{C}$ contains $5^{k}$ codewords for some integer $0 \leq k \leq 2 n$ and is called an $(n, 5^{k} )$ code. Moreover, if $\mathcal{C}$ has a minimum distance $d$, we call $\mathcal{C}$ an $(n, 5^{k}, d )$ code. An additive code $\mathcal{C}$ over $\mathbb{F}_{25}$ can be represented by a $k \times n$ matrix $G$, called a generator matrix, with entries from $\mathbb{F}_{25}$ whose rows span $\mathcal{C}$. That is, $\mathcal{C}$ is the $\mathbb{F}_{5}$-span of the rows of $G$.

Let $\omega \in \mathbb{F}_{25}$ such that $\omega^2 +3 = 0$. Then $\mathbb{F}_{25} = \mathbb{F}_5 [ \omega ]$. The trace map $\operatorname{Tr} \colon \mathbb{F}_{25} \to \mathbb{F}_5$ is defined by $\operatorname{Tr} (x) = x + x^5$. Following \cite{MassFormula_Ip}, every linear $I_{5}$-code $\mathcal{C}$ is attached with an additive code $\phi(\mathcal{C})$ over $\mathbb{F}_{25}$ by the map $\phi \colon I_{5} \to \mathbb{F}_{25}$ where
\begin{gather*}
\phi (0) = 0 ,   \ \     \phi ( \mathtt{a} ) =  3 ,      \ \    \phi (   \mathtt{b}  ) =    \omega  ,  \ \   \phi (  \mathtt{c}  )  =   \omega  + 3  ,   \  \     \phi ( \mathtt{d} )  =  2\omega , \  \     \phi ( \mathtt{e} )  =  1 ,    \\
 \phi ( \mathtt{f} ) = \omega + 1 ,   \ \     \phi ( \mathtt{g} ) =  2\omega + 3 ,    \  \    \phi (   \mathtt{h}  ) =    2\omega + 1  ,  \ \   \phi ( \mathtt{i} )  =  3 \omega ,   \  \   \phi ( \mathtt{j} )  =  4 ,    \\
 \phi ( \mathtt{k} ) = \omega + 4 ,   \ \     \phi ( \mathtt{l} ) =  3\omega + 3 ,    \  \    \phi (   \mathtt{m}  ) =    3\omega + 4  ,  \ \   \phi ( \mathtt{n} )  =  4  \omega ,   \  \   \phi ( \mathtt{o} )  =  2 ,    \\
 \phi ( \mathtt{p} ) = \omega + 2 ,   \ \     \phi ( \mathtt{q} ) =  4\omega + 3 ,    \  \    \phi (   \mathtt{r}  ) =    4\omega + 2  ,  \ \   \phi ( \mathtt{s} )  =  2  \omega +  4 ,   \  \   \phi ( \mathtt{t} )  =  2  \omega  +  2 ,    \\
 \phi ( \mathtt{u} ) = 3 \omega + 1 ,   \ \     \phi ( \mathtt{v} ) =  3 \omega + 2 ,    \  \    \phi (   \mathtt{w}  ) =    4\omega + 1  ,  \ \   \phi ( \mathtt{x} )  =  4  \omega +  4  .
\end{gather*}
This map is extended naturally from $I_5^n$ to $\mathbb{F}_{25}^n$. The parameters $(n, 5^{k}, d)$ of an $I_{5}$-code are identified with those of its image under $\phi$. For all $\mathbf{x} \in I_{5}^{n}$, we have $\operatorname{Tr}(\phi(\mathbf{x}))= \pi (\mathbf{x})$,  so $\operatorname{res}(\mathcal{C})=\operatorname{Tr}(\phi(\mathcal{C}))$. Similarly, $\operatorname{tor}(\mathcal{C}) = \mathbb{F}_{5}^{n} \cap \phi(\mathcal{C})$, i.e., $\operatorname{tor}(\mathcal{C})$ is the subfield subcode of $\phi(\mathcal{C})$.

\subsection{Mass Formulas}

We recall the mass formulas for  codes over $\mathbb{F}_5$ and $I_5$ mentioned in \cite{Kim_GFq, MassFormula_Ip}.

\begin{theorem} [Theorem 4.7 in \cite{Kim_GFq}]   \label{thm: mass formula for SO codes over F5}
Let $\varphi_{n, k}$ be the number of self-orthogonal codes over $\mathbb{F}_5$ of length $n$ and dimension $k$.
\begin{enumerate}
\item  If $n \geq 3$ is odd, then
$$
\varphi_{n, k}=\frac{\prod_{i=0}^{k-1}  ( 5^{n-1-2 i}-1 )   }{\prod_{i=1}^{k}   ( 5^{i}-1 )  }   .
$$

\item  If $n \geq 2$ is even, then  
$$
\varphi_{n, k}= \begin{cases}\dfrac{\left(5^{n-k}-5^{n / 2-k}+5^{n / 2}-1\right) \prod_{i=1}^{k-1}  ( 5^{n-2 i}-1 )  }{\prod_{i=1}^{k}   ( 5^{i}-1 )  } & \!\!\!\!\! , \text { if }(-1)^{\frac{n}{2}} \text { is square }    \\[1.0em]
 \dfrac{\left(5^{n-k}+5^{n / 2-k}-5^{n / 2}-1\right) \prod_{i=1}^{k-1} ( 5^{n-2 i}-1 )}{\prod_{i=1}^{k} ( 5^{i}-1 )} & \!\!\!\!\!  , \text { if }(-1)^{\frac{n}{2}} \text { is not square.}\end{cases}
$$
\end{enumerate}
\end{theorem}

For  $k \leq n$, the \textit{Gaussian binomial coefficient}  $\begin{bmatrix*}[c]
n  \\  k
\end{bmatrix*}_{q}$    is defined as
$$
\begin{bmatrix*}[c]
n  \\  k
\end{bmatrix*}_{q} \coloneqq   \frac{\left(q^{n}-1\right)\left(q^{n}-q\right) \cdots\left(q^{n}-q^{k-1}\right)}{\left(q^{k}-1\right)\left(q^{k}-q\right) \cdots\left(q^{k}-q^{k-1}\right)}
$$
which counts the number of $k$-dimensional subspaces in an $n$-dimensional vector space over $\mathbb{F}_{q}$.

Let $\mathcal{N}_{\operatorname{SO}}   \!   \left(n, k_1 , k_2 \right)$ denote the number of distinct self-orthogonal codes over $I_5$ of length $n$ and type $\{ k_1 , k_2 \}$. Similarly, define $\mathcal{N}_{\operatorname{QSD}}   \!   \left(n, k_1 , k_2 \right)$ to be the number of distinct quasi self-dual codes over $I_5$ of length $n$ and type $\{ k_1 , k_2 \}$. We also define $\mathcal{N}_{\operatorname{SD}}  \! \left( n, \frac{n}{2} , \frac{n}{2} \right)$ to be the number of distinct self-dual codes over $I_5$ of length $n$ and type $\left\lbrace \frac{n}{2} , \frac{n}{2} \right\rbrace$. The following theorem utilizes Theorem 7 in \cite{MassFormula_Ip} for the case when $p = 5$.

\begin{theorem}[Theorem 7 in \cite{MassFormula_Ip}]   \label{thm: mass formula SO codes over I5}
For all lengths $n$ and for the type $\left\{ k_1, k_2 \right\}$, the number of distinct self-orthogonal codes over $I_5$ is
\begin{align*}
\mathcal{N}_{\operatorname{SO}}   \!   \left(n, k_1 , k_2 \right)
=
\varphi_{n, k_1 }  \begin{bmatrix*}[c]
n-k_1  \\  k_2
\end{bmatrix*}_{5}   5^{ k_1 \left(n - k_1 - k_2 \right) }  .     
\end{align*}
\end{theorem}

The following corollaries follow from Theorem \ref{thm: mass formula SO codes over I5} and the usual counting technique under the group action.

\begin{cor}  \label{cor: mass formula for SO codes over I5}
For a given length $n$ and type $\left\{k_{1}, k_{2}\right\}$ with $0 \leq k_{1}, k_{2} \leq n$, we have
\begin{align*}
\sum_{\mathcal{C}} \frac{1}{| \operatorname{Aut} (\mathcal{C})|} =  \frac{\mathcal{N}_{\operatorname{SO}}    \!   \left(n, k_{1}, k_{2}\right)}{2^{n} n!}
\end{align*}
where $\mathcal{C}$ runs over distinct representatives of equivalence classes under a monomial action of self-orthogonal codes over $I_5$ of length $n$ and type $\left\{k_{1}, k_{2}\right\}$.    \label{cor: mass formula SO codes over I5}
\end{cor}

\begin{cor}    \label{cor: mass formula for QSD codes over I5}
For all codes of lengths $n$ of type $\{k_1, k_2 \}$ with $0 \leq k_1 , k_2 \leq n$, we have
\begin{align*}
\sum_{\mathcal{C}} \frac{1}{ | \operatorname{Aut} (\mathcal{C})  | } = \frac{\mathcal{N}_{\operatorname{QSD}} (n , k_1 , k_2)}{2^n   n!}  = \frac{ \varphi_{n, k_1}   \begin{bmatrix*}[c]
k_1 + k_2  \\   k_2
\end{bmatrix*}_{5}  5^{k_1^{2}}}{2^n n!}
\end{align*}
where $\mathcal{C}$ runs over distinct representatives of equivalence classes under monomial column permutations of quasi self-dual codes over $I_5$ of length $n$ and type $\left\{k_{1}, k_{2} \right\}$.
\end{cor}
\begin{proof}
The result follows immediately from Theorem \ref{thm: mass formula SO codes over I5} and Corollary \ref{cor: mass formula SO codes over I5} and using the fact that  quasi self-dual codes are self-orthogonal codes with $k_2 = n - 2 k_1$.
\end{proof}

\begin{cor}   \label{cor: mass formula for free SO codes over I5}
For all codes of lengths $n$ and of type $\{ k_1 , 0 \}$, we have
\begin{align*}
\sum_{\mathcal{C}}   \frac{1}{| \operatorname{Aut} (\mathcal{C})|}  =  \frac{\mathcal{N}_{\operatorname{SO}}    \!   \left(n, k_{1}, 0  \right)}{2^{n} n!}  =  \frac{ \varphi_{n, k_{1}} 5^{ k_1  (n - k_1)   }    }{  2^n n!  }
\end{align*}
where $\mathcal{C}$ runs over distinct representatives of equivalence classes under a monomial action of free self-orthogonal codes over $I_5$ of length $n$.
\end{cor}

\begin{cor}  \label{cor: mass formula for SD codes over I5}
For a given integer $n \geq 2$, we have the identity
$$
\sum_{\mathcal{C}}   \frac{1}{| \operatorname{Aut} (\mathcal{C})|}  =   \frac{\mathcal{N}_{\mathrm{SD}}  \! \left(n, \frac{n}{2}, \frac{n}{2}  \right)}{2^{n} n!}   =   \frac{\varphi_{n, \frac{n}{2} } }{2^n n!}
$$
where $\mathcal{C}$ runs over distinct representatives of equivalence classes under a monomial action of self-dual codes over $I_5$ of length $n$ and type $\left\lbrace \frac{n}{2} , \frac{n}{2} \right\rbrace$.
\end{cor}

\section{Constructions}   \label{sec: constructions}

\subsection{Construction of Self-Orthogonal Codes}
The following theorem characterizes self-orthogonal codes over $I_5$ which will be important in our construction.

\begin{theorem}  \label{thm: SO iff resC SO}
A linear code $\mathcal{C}$ of length $n$ over $I_{5}$ is self-orthogonal if and only if $\operatorname{res} (\mathcal{C})$ is a  self-orthogonal code over $\mathbb{F}_{5}$.
\end{theorem}

\begin{proof}
We make use of the proof of Theorem \ref{thm: c=au+bv for any c in C} to give a similar proof. It can be shown that every codeword $\mathbf{i} \in \mathcal{C}$ can be written as $\mathbf{i} = \mathtt{a} \mathbf{u} + \mathtt{b} \mathbf{v}$ for some $\mathbf{u} \in \operatorname{res}(\mathcal{C})$ and $\mathbf{v} \in \mathbb{F}_{5}^{n}$. If $\mathbf{u} \in \operatorname{res}(\mathcal{C})$, then $\mathtt{a} \mathbf{u} + \mathtt{b} \mathbf{v} \in \mathcal{C}$ for some $\mathbf{v} \in \mathbb{F}_{5}^{n}$. If $\mathbf{i}_{1}, \mathbf{i}_{2} \in \mathcal{C}$, then we can write $\mathbf{i}_{1}$ and $\mathbf{i}_{2}$ in $\mathtt{b}$-adic decomposition form as $\mathbf{i}_{1} = \mathtt{a} \mathbf{x}_{1} + \mathtt{b} \mathbf{y}_{1}$ and $\mathbf{i}_{2} = \mathtt{a} \mathbf{x}_{2} + \mathtt{b} \mathbf{y}_{2}$ where $\mathbf{x}_{1}, \mathbf{x}_{2} \in \operatorname{res} (\mathcal{C})$ and $\mathbf{y}_{1}, \mathbf{y}_{2} \in \mathbb{F}_{5}^{n}$. Then we have
$$
\left( \mathtt{a} \mathbf{x}_{1} + \mathtt{b} \mathbf{y}_{1}  ,  \mathtt{a} \mathbf{x}_{2} +  \mathtt{b} \mathbf{y}_{2}\right) = \mathtt{a}^{2}\left(\mathbf{x}_{1}, \mathbf{x}_{2}\right)+  \mathtt{ab}  \left(\mathbf{x}_{1}, \mathbf{y}_{2}\right) +  \mathtt{ab} \left(\mathbf{x}_{2}, \mathbf{y}_{1}\right) +  \mathtt{b}^{2}\left(\mathbf{y}_{1}, \mathbf{y}_{2}\right) =  \mathtt{b}  \left(\mathbf{x}_{1}, \mathbf{x}_{2}\right) .
$$
By the relation between $\mathcal{C}$ and $\operatorname{res}(\mathcal{C})$, the result follows.
\end{proof}

Now, we do some building-up constructions of self-orthogonal codes. The following theorem gives a propagation rule of order $2$ for self-orthogonal codes.

\begin{theorem}   \label{thm: SO code of length n+2 (1 generator)}
Let $\mathcal{C}_0$ be a  self-orthogonal code over $I_5$ of length $n$ with generator matrix $G_0 = (\mathbf{r}_i)$ where $(\mathbf{r}_i)$ is the $i$th row of $G_0$ for $i =  1 , 2 , \dotsc , m$. Let $\mathbf{x} \in \mathbb{F}_5^n$ and let $\mu , \nu \in J = \{  0 , \mathtt{b} , \mathtt{d} , \mathtt{i} , \mathtt{n}  \}$, not both zero, such that $2 \mu + \nu = 0$. Then the code $\mathcal{C}$ with generator matrix
\begin{align*}
{G}  =  \left(    \begin{array}{cc|c}
0  &   \mu   &   \nu  \mathbf{x}   \\  \hline
(\mathbf{x} , \mathbf{r}_1)   &   2  (\mathbf{x} , \mathbf{r}_1)     &   \mathbf{r}_1    \\
\vdots   &   \vdots   &   \vdots   \\
(\mathbf{x} , \mathbf{r}_m)   &   2  (\mathbf{x} , \mathbf{r}_m)     &   \mathbf{r}_m
\end{array}    \right)
\end{align*}
is a self-orthogonal code of length $n + 2$.
\end{theorem}

\begin{proof}
It suffices to show that the rows of $G$ are orthogonal to themselves and to one another. Let  $\mathbf{y}_{0} =  ( \begin{array}{ccc}
0  &   \mu   &   \nu  \mathbf{x}  \end{array}  )$ be the first row of $G$, and let $\mathbf{y}_{i} =  ( \begin{array}{ccc}
(\mathbf{x} , \mathbf{r}_i)   &   2  (\mathbf{x} , \mathbf{r}_i)     &   \mathbf{r}_i   \end{array}  )$ be the $(i+1)$th row of $G$ for $i = 1 , 2 , \dotsc , m$. Then
\begin{align*}
(\mathbf{y}_0 , \mathbf{y}_0 )  &=  \mu^2 + \nu^2   (  \mathbf{x} , \mathbf{x} )  =  0  ,  \\
(\mathbf{y}_i , \mathbf{y}_j )  &=    ( \mathbf{x}  ,  \mathbf{r}_i  )   ( \mathbf{x}  ,  \mathbf{r}_j  ) +  4  ( \mathbf{x}  ,  \mathbf{r}_i  )   ( \mathbf{x}  ,  \mathbf{r}_j  )  +  ( \mathbf{r}_i  ,  \mathbf{r}_j  )  =   0 , \text{ and}   \\
(\mathbf{y}_0 , \mathbf{y}_i )  &=  2 \mu  ( \mathbf{x}  ,  \mathbf{r}_i  )   +  \nu  ( \mathbf{x}  ,  \mathbf{r}_i  )   =  ( 2 \mu + \nu  )( \mathbf{x}  ,  \mathbf{r}_i  )   =   0
\end{align*}
since $\mu, \nu \in J$. Thus, $\mathcal{C}$ is a self-orthogonal code of length $n + 2$.
\end{proof}

The next theorem gives a propagation rule of order $4$ for self-orthogonal codes.

\begin{theorem}  \label{thm: SO code of length n+4 (1 generator)}
Let $\mathcal{C}_0$ be a self-orthogonal code over $I_5$ of length $n$ with generator matrix $G_0 = (\mathbf{r}_i)$ where $(\mathbf{r}_i)$ is the $i$th row of $G_0$ for $i =  1 , 2 , \dotsc , m$. Let $\mathbf{x} \in \mathbb{F}_5^n$ and $\alpha , \beta , \gamma, \delta  \in I_5$, not all zero, such that $\alpha + \beta + 2 \gamma + \delta  = 0$. Then the code $\mathcal{C}$ with generator matrix
			\begin{align*}
			G = \left(
			\begin{array}{cccc|c}
			\alpha             & \beta      & \gamma   &          0                  & \delta   \mathbf{x} \\   \hline
			(\mathbf{x} , \mathbf{r}_1) & (\mathbf{x} , \mathbf{r}_1) &  2 (\mathbf{x} , \mathbf{r}_1)  & 2 (\mathbf{x} , \mathbf{r}_1)    & \mathbf{r}_1      \\
			\vdots            & \vdots            & \vdots                  & \vdots    & \vdots     \\
			(\mathbf{x} , \mathbf{r}_m) & (\mathbf{x} , \mathbf{r}_m) &  2 (\mathbf{x} , \mathbf{r}_m)   &  2(\mathbf{x} , \mathbf{r}_m)   & \mathbf{r}_m
			\end{array}
			\right)
			\end{align*}
			is a self-orthogonal code of length $n + 4$  if one of the following conditions is satisfied:
			\begin{enumerate}[(i)]
			\item  $(\mathbf{x} , \mathbf{x}) = 1$ and $\alpha^2 + \beta^2 + \gamma^2  + \delta^2  = 0$,
			\item  $(\mathbf{x} , \mathbf{x}) = 2$ and $\alpha^2 + \beta^2 + \gamma^2  + 2 \delta^2  = 0$,
			\item  $(\mathbf{x} , \mathbf{x}) = 3$ and $\alpha^2 + \beta^2 + \gamma^2  + 3 \delta^2  = 0$, or
			\item  $(\mathbf{x} , \mathbf{x}) = 4$ and $\alpha^2 + \beta^2 + \gamma^2  + 4  \delta^2 = 0$.
			\end{enumerate}
\end{theorem}

\begin{proof}
Let
\begin{align*}
\mathbf{y}_{0}= ( \begin{array}{ccccc}
\alpha       & \beta      & \gamma   &          0                  & \delta  \mathbf{x}
\end{array}  )
\end{align*}
be the first row of  $G$ and let
\begin{align*}
\mathbf{y}_i = ( \begin{array}{ccccc}
(\mathbf{x} , \mathbf{r}_i) & (\mathbf{x} , \mathbf{r}_i) & 2 (\mathbf{x} , \mathbf{r}_i)  & 2 (\mathbf{x} , \mathbf{r}_i)     & \mathbf{r}_i
\end{array}  )
\end{align*}
be the $(i+1)$th row of $G$ for $i = 1 , 2 , \dotsc , m$.
\begin{enumerate}[(i)]
\item  If $(\mathbf{x} , \mathbf{x}) = 1$ and $\alpha^2 + \beta^2 + \gamma^2  + \delta^2 = 0$, then
\begin{align*}
(\mathbf{y}_0 , \mathbf{y}_0 )  &=  \alpha^2 + \beta^2 + \gamma^2 + \delta^2  ( \mathbf{x} , \mathbf{x}  )  = \alpha^2 + \beta^2 + \gamma^2 + \delta^2   =  0  ,  \\
(\mathbf{y}_0 , \mathbf{y}_i )  &=  \alpha  ( \mathbf{x}  ,  \mathbf{r}_i  )  +  \beta ( \mathbf{x}  ,  \mathbf{r}_i  )  +  2 \gamma  ( \mathbf{x}  ,  \mathbf{r}_i  )  +  \delta  ( \mathbf{x}  ,  \mathbf{r}_i  )    =  (  \alpha + \beta + 2 \gamma + \delta  )  ( \mathbf{x}  ,  \mathbf{r}_i  )  =  0  ,  \text{ and}   \\
(\mathbf{y}_i , \mathbf{y}_j )  &=  10  ( \mathbf{x}  ,  \mathbf{r}_i  )   ( \mathbf{x}  ,  \mathbf{r}_j  )  +  ( \mathbf{r}_i  ,  \mathbf{r}_j  )  =   0
\end{align*}
for all $i , j = 1 , 2 , \dotsc , m$.

\item  If $(\mathbf{x} , \mathbf{x}) = 2$ and $\alpha^2 + \beta^2 + \gamma^2  + 2 \delta^2  = 0$, then
\begin{align*}
(\mathbf{y}_0 , \mathbf{y}_0 )  &=  \alpha^2 + \beta^2 + \gamma^2 + \delta^2   ( \mathbf{x} , \mathbf{x}  )  = \alpha^2 + \beta^2 + \gamma^2 + 2 \delta^2  =  0  , \\
(\mathbf{y}_0 , \mathbf{y}_i )  &=  \alpha  ( \mathbf{x}  ,  \mathbf{r}_i  )  +  \beta ( \mathbf{x}  ,  \mathbf{r}_i  )  + 2 \gamma  ( \mathbf{x}  ,  \mathbf{r}_i  )  +  \delta  ( \mathbf{x}  ,  \mathbf{r}_i  )
=  (  \alpha + \beta + 2 \gamma + \delta   )  ( \mathbf{x}  ,  \mathbf{r}_i  )  =  0  ,  \text{ and}  \\
(\mathbf{y}_i , \mathbf{y}_j )  &=  10  ( \mathbf{x}  ,  \mathbf{r}_i  )  ( \mathbf{x}  ,  \mathbf{r}_j  )   +  ( \mathbf{r}_i  ,  \mathbf{r}_j  )  =   0
\end{align*}
for all $i , j = 1 , 2 , \dotsc , m$.

\item  If $(\mathbf{x} , \mathbf{x}) = 3$ and $\alpha^2 + \beta^2 + \gamma^2  + 3 \delta^2  = 0$, then
\begin{align*}
(\mathbf{y}_0 , \mathbf{y}_0 )  &=  \alpha^2 + \beta^2 + \gamma^2 + \delta^2   ( \mathbf{x} , \mathbf{x}  )  = \alpha^2 + \beta^2 + \gamma^2 + 3 \delta^2  =  0  ,  \\
(\mathbf{y}_0 , \mathbf{y}_i )  &=  \alpha  ( \mathbf{x}  ,  \mathbf{r}_i  )  +  \beta ( \mathbf{x}  ,  \mathbf{r}_i  )  +  2 \gamma  ( \mathbf{x}  ,  \mathbf{r}_i  )  +  \delta  ( \mathbf{x}  ,  \mathbf{r}_i  )    =  (  \alpha + \beta + 2 \gamma + \delta  )  ( \mathbf{x}  ,  \mathbf{r}_i  )  =  0   ,  \text{ and}   \\
(\mathbf{y}_i , \mathbf{y}_j )  &=  10  ( \mathbf{x}  ,  \mathbf{r}_i  )   ( \mathbf{x}  ,  \mathbf{r}_j  )   +  ( \mathbf{r}_i  ,  \mathbf{r}_j  )  =   0
\end{align*}
for all $i , j = 1 , 2 , \dotsc , m$.

\item  If $(\mathbf{x} , \mathbf{x}) = 4$ and $\alpha^2 + \beta^2 + \gamma^2  + 4 \delta^2 = 0$, then
\begin{align*}
(\mathbf{y}_0 , \mathbf{y}_0 )  &=  \alpha^2 + \beta^2 + \gamma^2 + \delta^2   ( \mathbf{x} , \mathbf{x}  )  = \alpha^2 + \beta^2 + \gamma^2 + 4 \delta^2  =  0  ,  \\
(\mathbf{y}_0 , \mathbf{y}_i )  &=  \alpha  ( \mathbf{x}  ,  \mathbf{r}_i  )  +  \beta ( \mathbf{x}  ,  \mathbf{r}_i  )  +  2 \gamma  ( \mathbf{x}  ,  \mathbf{r}_i  )  +  \delta  ( \mathbf{x}  ,  \mathbf{r}_i  )    =  (  \alpha + \beta + 2 \gamma + \delta  )  ( \mathbf{x}  ,  \mathbf{r}_i  )  =  0  ,   \text{ and}    \\
(\mathbf{y}_i , \mathbf{y}_j )  &=  10  ( \mathbf{x}  ,  \mathbf{r}_i  )  ( \mathbf{x}  ,  \mathbf{r}_j  )  +  ( \mathbf{r}_i  ,  \mathbf{r}_j  )  =   0
\end{align*}
for all $i , j = 1 , 2 , \dotsc , m$.
\end{enumerate}
In all cases, we see that $\mathcal{C}$ is a self-orthogonal code of length $n + 4$.
\end{proof}

\begin{theorem}   \label{thm: SO code of length n+5 (2 generators)}
Let $\mathcal{C}_{0}$ be a self-orthogonal code over $I_{5}$ of length $n$ with generator matrix $G_{0} = \left(\mathbf{r}_{i}\right)$ where $\mathbf{r}_{i}$ is the ith row of $G_{0}$ for $i = 1, 2 , \dotsc, m$. Let $\mathbf{x} \in \mathbb{F}_{5}^{n}$ and $\alpha, \beta, \gamma, \delta, \varepsilon \in I_{5}$, not all zero, such that $\alpha + \beta + \gamma + \delta + \varepsilon = 0$. Let $\iota, \kappa , \lambda , \mu , \nu  \in J = \{0, \mathtt{b}, \mathtt{d} , \mathtt{i} , \mathtt{n} \}$, not all zero, such that $\iota+ \kappa + \lambda + \mu + \nu =0$. Then the code $\mathcal{C}$ with generator matrix
\begin{align*}
G = \left(
\begin{array}{ccccc|c}
\alpha             & \beta      & \gamma   &       \delta    &   0                  & \varepsilon  \mathbf{x} \\
0             &  \iota      &  \kappa   &       \lambda    &   \mu                  & \nu    \mathbf{x} \\   \hline
(\mathbf{x} , \mathbf{r}_1) & (\mathbf{x} , \mathbf{r}_1) & (\mathbf{x} , \mathbf{r}_1)  & (\mathbf{x} , \mathbf{r}_1)  & (\mathbf{x} , \mathbf{r}_1)    & \mathbf{r}_1      \\
\vdots            & \vdots            & \vdots            & \vdots        & \vdots    & \vdots     \\
(\mathbf{x} , \mathbf{r}_m) & (\mathbf{x} , \mathbf{r}_m) & (\mathbf{x} , \mathbf{r}_m)   & (\mathbf{x} , \mathbf{r}_m)   & (\mathbf{x} , \mathbf{r}_m)   & \mathbf{r}_m
\end{array}
\right)
\end{align*}
is a self-orthogonal code of length $n + 5$ if one of the following conditions is satisfied:
	\begin{enumerate}[(i)]
	\item  $(\mathbf{x} , \mathbf{x}) = 1$ and $\alpha^2 + \beta^2 + \gamma^2  + \delta^2 + \varepsilon^2 = 0$,
	\item  $(\mathbf{x} , \mathbf{x}) = 2$ and $\alpha^2 + \beta^2 + \gamma^2  + \delta^2 + 2 \varepsilon^2 = 0$,
	\item  $(\mathbf{x} , \mathbf{x}) = 3$ and $\alpha^2 + \beta^2 + \gamma^2  + \delta^2 + 3 \varepsilon^2 = 0$, or
	\item  $(\mathbf{x} , \mathbf{x}) = 4$ and $\alpha^2 + \beta^2 + \gamma^2  + \delta^2 + 4 \varepsilon^2 = 0$.
	\end{enumerate}
\end{theorem}
\begin{proof}
Let $\mathbf{y}_{0}  = \left(\begin{array}{cccccc}  \alpha       & \beta      & \gamma   &       \delta    &   0                  & \varepsilon  \mathbf{x}    \end{array}  \right)$, $\mathbf{y}_{0}^{\prime} = \left(\begin{array}{cccccc}  0 & \iota & \kappa & \lambda & \mu  &  \nu  \mathbf{x}  \end{array}  \right)$, and
\begin{align*}
\mathbf{y}_i = ( \begin{array}{cccccc}
(\mathbf{x} , \mathbf{r}_i) & (\mathbf{x} , \mathbf{r}_i) & (\mathbf{x} , \mathbf{r}_i)  & (\mathbf{x} , \mathbf{r}_i)  & (\mathbf{x} , \mathbf{r}_i)    & \mathbf{r}_i
\end{array}  )
\end{align*}
for $i = 1,2, \ldots, m$.

\begin{enumerate}[(i)]
\item  If $(\mathbf{x} , \mathbf{x}) = 1$ and $\alpha^2 + \beta^2 + \gamma^2  + \delta^2 + \varepsilon^2 = 0$, then
\begin{align*}
( \mathbf{y}_0 , \mathbf{y}_0 )  &=  \alpha^2 + \beta^2 + \gamma^2 + \delta^2 + \varepsilon^2  (  \mathbf{x} , \mathbf{x} )  =   \alpha^2 + \beta^2 + \gamma^2 + \delta^2 + \varepsilon^2  =  0,  \\
( \mathbf{y}_0 , \mathbf{y}_i )  &= \alpha  ( \mathbf{x}  ,  \mathbf{r}_i  )   +  \beta  ( \mathbf{x}  ,  \mathbf{r}_i  )  +  \gamma  ( \mathbf{x}  ,  \mathbf{r}_i  )  +  \delta  ( \mathbf{x}  ,  \mathbf{r}_i  )  +  \varepsilon  ( \mathbf{x}  ,  \mathbf{r}_i  ) = ( \alpha + \beta +  \gamma  +  \delta   +  \varepsilon )  (  \mathbf{x} , \mathbf{r}_i )   =  0,    \\
( \mathbf{y}_i , \mathbf{y}_j )  &=   5   ( \mathbf{x}, \mathbf{r}_i )  ( \mathbf{x}, \mathbf{r}_j )  +  ( \mathbf{r}_i, \mathbf{r}_j )  =  0 ,  \\
(\mathbf{y}_0^{\prime} , \mathbf{y}_0^{\prime} )  &=  \iota^2 + \kappa^2 + \lambda^2 + \mu^2 + \nu^2  (  \mathbf{x} , \mathbf{x} )  =  0  ,  \\
(\mathbf{y}_0 , \mathbf{y}_0^{\prime} )  &=   \beta \iota  +  \gamma  \kappa  +  \delta  \lambda  +  \varepsilon  \nu  (  \mathbf{x} , \mathbf{x} )  =  0  \text{ since $\beta \iota =  \gamma  \kappa  =  \delta  \lambda  =  \varepsilon  \nu=0$,  and}   \\
(\mathbf{y}_0^{\prime} , \mathbf{y}_i )  &=  \iota  ( \mathbf{x}  ,  \mathbf{r}_i  )   +  \kappa  ( \mathbf{x}  ,  \mathbf{r}_i  )  +  \lambda  ( \mathbf{x}  ,  \mathbf{r}_i  )  +  \mu  ( \mathbf{x}  ,  \mathbf{r}_i  )  +  \nu  ( \mathbf{x}  ,  \mathbf{r}_i  )  =  ( \iota + \kappa + \lambda + \mu + \nu )( \mathbf{x}  ,  \mathbf{r}_i  )   =   0
\end{align*}
for all $i , j = 1 , 2 , \dotsc , m$.

\item  If $(\mathbf{x} , \mathbf{x}) = 2$ and $\alpha^2 + \beta^2 + \gamma^2  + \delta^2 + 2 \varepsilon^2 = 0$, then
\begin{align*}
( \mathbf{y}_0 , \mathbf{y}_0 )  &=  \alpha^2 + \beta^2 + \gamma^2 + \delta^2 + \varepsilon^2  (  \mathbf{x} , \mathbf{x} )  =   \alpha^2 + \beta^2 + \gamma^2 + \delta^2 + 2 \varepsilon^2  =  0,  \\
( \mathbf{y}_0 , \mathbf{y}_i )  &=  \alpha  ( \mathbf{x}  ,  \mathbf{r}_i  )   +  \beta  ( \mathbf{x}  ,  \mathbf{r}_i  )  +  \gamma  ( \mathbf{x}  ,  \mathbf{r}_i  )  +  \delta  ( \mathbf{x}  ,  \mathbf{r}_i  )  +  \varepsilon  ( \mathbf{x}  ,  \mathbf{r}_i  ) = ( \alpha + \beta +  \gamma  +  \delta   +  \varepsilon )  (  \mathbf{x} , \mathbf{r}_i )   =  0,    \\
( \mathbf{y}_i , \mathbf{y}_j )  &=   5   ( \mathbf{x}, \mathbf{r}_i )  ( \mathbf{x}, \mathbf{r}_j )  +  ( \mathbf{r}_i, \mathbf{r}_j )  =  0 ,  \\
(\mathbf{y}_0^{\prime} , \mathbf{y}_0^{\prime} )  &=  \iota^2 + \kappa^2 + \lambda^2 + \mu^2 + \nu^2  (  \mathbf{x} , \mathbf{x} )  =  0  ,  \\
(\mathbf{y}_0 , \mathbf{y}_0^{\prime} )  &=   \beta \iota  +  \gamma  \kappa  +  \delta  \lambda  +  \varepsilon  \nu  (  \mathbf{x} , \mathbf{x} )  =  0  \text{ since $\beta \iota =  \gamma  \kappa  =  \delta  \lambda  =  \varepsilon  \nu=0$,  and}   \\
(\mathbf{y}_0^{\prime} , \mathbf{y}_i )  &=  \iota  ( \mathbf{x}  ,  \mathbf{r}_i  )   +  \kappa  ( \mathbf{x}  ,  \mathbf{r}_i  )  +  \lambda  ( \mathbf{x}  ,  \mathbf{r}_i  )  +  \mu  ( \mathbf{x}  ,  \mathbf{r}_i  )  +  \nu  ( \mathbf{x}  ,  \mathbf{r}_i  )  =  ( \iota + \kappa + \lambda + \mu + \nu )( \mathbf{x}  ,  \mathbf{r}_i  )   =   0
\end{align*}
for all $i , j = 1 , 2 , \dotsc , m$.

\item  If $(\mathbf{x} , \mathbf{x}) = 3$ and $\alpha^2 + \beta^2 + \gamma^2  + \delta^2 + 3 \varepsilon^2 = 0$, then
\begin{align*}
( \mathbf{y}_0 , \mathbf{y}_0 )  &=  \alpha^2 + \beta^2 + \gamma^2 + \delta^2 + \varepsilon^2  (  \mathbf{x} , \mathbf{x} )  =   \alpha^2 + \beta^2 + \gamma^2 + \delta^2 + 3 \varepsilon^2  =  0,  \\
( \mathbf{y}_0 , \mathbf{y}_i )  &=  \alpha  ( \mathbf{x}  ,  \mathbf{r}_i  )   +  \beta  ( \mathbf{x}  ,  \mathbf{r}_i  )  +  \gamma  ( \mathbf{x}  ,  \mathbf{r}_i  )  +  \delta  ( \mathbf{x}  ,  \mathbf{r}_i  )  +  \varepsilon  ( \mathbf{x}  ,  \mathbf{r}_i  ) = ( \alpha + \beta +  \gamma  +  \delta   +  \varepsilon )  (  \mathbf{x} , \mathbf{r}_i )   =  0,    \\
( \mathbf{y}_i , \mathbf{y}_j )  &=   5   ( \mathbf{x}, \mathbf{r}_i )  ( \mathbf{x}, \mathbf{r}_j )  +  ( \mathbf{r}_i, \mathbf{r}_j )  =  0 ,  \\
(\mathbf{y}_0^{\prime} , \mathbf{y}_0^{\prime} )  &=  \iota^2 + \kappa^2 + \lambda^2 + \mu^2 + \nu^2  (  \mathbf{x} , \mathbf{x} )  =  0  ,  \\
(\mathbf{y}_0 , \mathbf{y}_0^{\prime} )  &=   \beta \iota  +  \gamma  \kappa  +  \delta  \lambda  +  \varepsilon  \nu  (  \mathbf{x} , \mathbf{x} )  =  0  \text{ since $\beta \iota =  \gamma  \kappa  =  \delta  \lambda  =  \varepsilon  \nu=0$,  and}   \\
(\mathbf{y}_0^{\prime} , \mathbf{y}_i )  &=  \iota  ( \mathbf{x}  ,  \mathbf{r}_i  )   +  \kappa  ( \mathbf{x}  ,  \mathbf{r}_i  )  +  \lambda  ( \mathbf{x}  ,  \mathbf{r}_i  )  +  \mu  ( \mathbf{x}  ,  \mathbf{r}_i  )  +  \nu  ( \mathbf{x}  ,  \mathbf{r}_i  )  =  ( \iota + \kappa + \lambda + \mu + \nu )( \mathbf{x}  ,  \mathbf{r}_i  )   =   0
\end{align*}
for all $i , j = 1 , 2 , \dotsc , m$.

\item  If $(\mathbf{x} , \mathbf{x}) = 4$ and $\alpha^2 + \beta^2 + \gamma^2  + \delta^2 + 4 \varepsilon^2 = 0$, then
\begin{align*}
( \mathbf{y}_0 , \mathbf{y}_0 )  &=  \alpha^2 + \beta^2 + \gamma^2 + \delta^2 + \varepsilon^2  (  \mathbf{x} , \mathbf{x} )  =   \alpha^2 + \beta^2 + \gamma^2 + \delta^2 + 4 \varepsilon^2  =  0,  \\
( \mathbf{y}_0 , \mathbf{y}_i )  &=  \alpha  ( \mathbf{x}  ,  \mathbf{r}_i  )   +  \beta  ( \mathbf{x}  ,  \mathbf{r}_i  )  +  \gamma  ( \mathbf{x}  ,  \mathbf{r}_i  )  +  \delta  ( \mathbf{x}  ,  \mathbf{r}_i  )  +  \varepsilon  ( \mathbf{x}  ,  \mathbf{r}_i  ) = ( \alpha + \beta +  \gamma  +  \delta   +  \varepsilon )  (  \mathbf{x} , \mathbf{r}_i )   =  0,    \\
( \mathbf{y}_i , \mathbf{y}_j )  &=   5   ( \mathbf{x}, \mathbf{r}_i )  ( \mathbf{x}, \mathbf{r}_j )  +  ( \mathbf{r}_i, \mathbf{r}_j )  =  0 ,  \\
(\mathbf{y}_0^{\prime} , \mathbf{y}_0^{\prime} )  &=  \iota^2 + \kappa^2 + \lambda^2 + \mu^2 + \nu^2  (  \mathbf{x} , \mathbf{x} )  =  0  ,  \\
(\mathbf{y}_0 , \mathbf{y}_0^{\prime} )  &=   \beta \iota  +  \gamma  \kappa  +  \delta  \lambda  +  \varepsilon  \nu  (  \mathbf{x} , \mathbf{x} )  =  0  \text{ since $\beta \iota =  \gamma  \kappa  =  \delta  \lambda  =  \varepsilon  \nu=0$,  and}   \\
(\mathbf{y}_0^{\prime} , \mathbf{y}_i )  &=  \iota  ( \mathbf{x}  ,  \mathbf{r}_i  )   +  \kappa  ( \mathbf{x}  ,  \mathbf{r}_i  )  +  \lambda  ( \mathbf{x}  ,  \mathbf{r}_i  )  +  \mu  ( \mathbf{x}  ,  \mathbf{r}_i  )  +  \nu  ( \mathbf{x}  ,  \mathbf{r}_i  )  =  ( \iota + \kappa + \lambda + \mu + \nu )( \mathbf{x}  ,  \mathbf{r}_i  )   =   0
\end{align*}
for all $i , j = 1 , 2 , \dotsc , m$.
\end{enumerate}

In all cases, we see that $\mathcal{C}$ is self-orthogonal  of length $n + 5$.
\end{proof}

\begin{example}
By Theorem \ref{thm: SO iff resC SO}, we have a self-orthogonal $(2, 5^2, 2)$ code with generator matrix
\begin{align*}
G_0 = \left(
\begin{array}{cccc}
\mathtt{a}  &   \mathtt{e}
\end{array}
\right)  .
\end{align*}
This code is of type $\{ 1 , 0 \}$ and has weight distribution $[ \langle 0 , 1  \rangle ,   \langle 2, 24 \rangle  ]$.

Using Theorem \ref{thm: SO code of length n+2 (1 generator)} with the base matrix $G_0$, $\mu = \mathtt{b},  \nu = \mathtt{i}$, and $\mathbf{x} = (2, 0)$, we obtain a self-orthogonal $(4, 5^3, 2)$ code with generator matrix
\begin{align*}
G_{0,1} = \left(
\begin{array}{cc|cccc}
0                   &   \mathtt{b}  &   \mathtt{b}   &       0    \\
\hline
\mathtt{e}	  &   \ovariant  &   \mathtt{a}  &   \mathtt{e}
\end{array}
\right) .
\end{align*}
This code is of type $\{ 1 , 1 \}$ and has weight distribution $[ \langle 0 , 1  \rangle ,  \langle 2, 4  \rangle  ,   \langle  3 , 8  \rangle, \langle  4 , 112  \rangle]$.

Now, if we use Theorem \ref{thm: SO code of length n+4 (1 generator)} with the base matrix $G_0$, $\alpha = \mathtt{a}$, $\beta = \mathtt{c},  \gamma =  \mathtt{s}$, $\delta = \mathtt{e}$, and $\mathbf{x} = (1, 0)$, we obtain a self-orthogonal $(6, 5^4 , 4)$ code with generator matrix
\begin{align*}
G_{0,2} = \left(
\begin{array}{cccc|cc}
\mathtt{a}  &  \mathtt{c}  &  \mathtt{s}  &  0  &  \mathtt{e}   &    0  \\
\hline
\mathtt{a}  & \mathtt{a}  &   \mathtt{e}  &   \mathtt{e}  &   \mathtt{a}    &  \mathtt{e}
\end{array}
\right) .
\end{align*}
This code is of type $\{ 2 , 0 \}$ and has weight distribution $[ \langle 0 , 1  \rangle , \langle 4, 28 \rangle, \langle 5 , 88  \rangle, \langle 6 , 508  \rangle ]$.

If we instead use Theorem \ref{thm: SO code of length n+5 (2 generators)} with the base matrix $G_0$, $\alpha = \mathtt{a}$, $\beta = \mathtt{e}$, $\gamma = \ovariant$, $\delta = \mathtt{j}$, $\varepsilon = 0$, $\iota = \mathtt{b}$, $\kappa = \mathtt{d}$, $\lambda = 0$, $\mu = \mathtt{i}$, $\nu = \mathtt{n}$, and $\mathbf{x} = (1, 4)$, we obtain a self-orthogonal $(7, 5^5, 4)$ code with generator matrix
\begin{align*}
G_{0,3} = \left(
\begin{array}{ccccc|cc}
\mathtt{a}   &  \mathtt{e}  &   \ovariant  &  \mathtt{j}   &   0   &  0   &  0     \\
0   &  \mathtt{b}  &   \mathtt{d}  &  0   &  \mathtt{i}    &  \mathtt{n}  &  \mathtt{b}   \\
\hline
\ovariant   &  \ovariant  &   \ovariant &  \ovariant   &  \ovariant   &   \mathtt{a}  &   \mathtt{e}
\end{array}
\right) .
\end{align*}
This code is of type $\{ 2 , 1 \}$ and has weight distribution $[ \langle 0 , 1  \rangle , \langle 4 , 32 \rangle,  \langle  5 , 48  \rangle , \langle 6 , 676 \rangle , \langle 7 , 2368  \rangle ]$.
\end{example}

\subsection{Construction of Quasi Self-Dual Codes}

In this section, we present construction methods for quasi self-dual codes over $I_5$. We start with what is known as the \textbf{multilevel construction}.

\begin{theorem}  \label{thm: multilevelconstruction}
Let $\mathcal{C}_{1}$ be a self-orthogonal  code of length $n$ over $\mathbb{F}_5$, and let $\mathcal{C}_{2}$ be a  code of length $n$ over $\mathbb{F}_5$ such that $\mathcal{C}_{1} \subseteq \mathcal{C}_{2}$. The code $\mathcal{C}$ defined by the relation
\begin{align*}
\mathcal{C} = \mathtt{a} \mathcal{C}_{1}  +  \mathtt{b}   \mathcal{C}_{2}
\end{align*}
is self-orthogonal with residue code $\mathcal{C}_{1}$ and torsion code $\mathcal{C}_{2}$. Furthermore, if $\left|\mathcal{C}_{1}\right|\left|\mathcal{C}_{2}\right|=5^{n}$, then $\mathcal{C}$ is quasi self-dual.
\end{theorem}
\begin{proof}
This is just a direct analogue of Theorem 5 in \cite{I3} by extending the case of the ring $I_3$ with characteristic $3$ to $I_5$ with characteristic $5$.
\end{proof}

\begin{prop}[Proposition 6 in \cite{MassFormula_Ip}]
Let $\mathcal{C} = \mathtt{a} \operatorname{res}(\mathcal{C}) +  \mathtt{b}  \operatorname{tor}(\mathcal{C})$ be an $I_{5}$-code. Suppose further that $\operatorname{res} (\mathcal{C})$ is a self-dual $[n, k]$ code over $\mathbb{F}_{5}$. If $\operatorname{tor}(\mathcal{C})=\operatorname{res}(\mathcal{C})$, then $\mathcal{C}$ will be a QSD code.
\end{prop}

We now present some building-up constructions for quasi self-dual codes. The next two theorems give propagation rules of order $2$ for quasi self-dual codes over $I_5$.

\begin{theorem}   \label{thm: QSD code of length n+2 (1 generator)}
Let $\mathcal{C}_0$ be a quasi self-dual code over $I_5$ of length $n$ with generator matrix $G_0 = (\mathbf{r}_i)$ where $(\mathbf{r}_i)$ is the $i$th row of $G_0$ for $i =  1 , 2 , \dotsc , m$. Let $\mathbf{x} \in \mathbb{F}_5^n$ and $\alpha , \beta , \gamma  \in I_5$, not all zero, such that $\alpha + 2 \beta + \gamma = 0$. Then the code $\mathcal{C}$ with generator matrix
\begin{align*}
G  =  \left(    \begin{array}{cc|c}
\alpha  &   \beta   &   \gamma  \mathbf{x}   \\  \hline
(\mathbf{x} , \mathbf{r}_1)   &   2  (\mathbf{x} , \mathbf{r}_1)     &   \mathbf{r}_1    \\
\vdots   &   \vdots   &   \vdots   \\
(\mathbf{x} , \mathbf{r}_m)   &   2  (\mathbf{x} , \mathbf{r}_m)     &   \mathbf{r}_m
\end{array}    \right)
\end{align*}
is a quasi self-dual code of length $n + 2$ if one of the following conditions is satisfied:
\begin{enumerate}[(i)]
\item  $(\mathbf{x} , \mathbf{x}) = 1$ and $\alpha^2 + \beta^2 + \gamma^2  = 0$,   \label{enum: c^2+d^2+e^2=0 (ii)}
\item  $(\mathbf{x} , \mathbf{x}) = 2$ and $\alpha^2 + \beta^2 + 2 \gamma^2  = 0$,
\item  $(\mathbf{x} , \mathbf{x}) = 3$ and $\alpha^2 + \beta^2 + 3 \gamma^2  = 0$, or
\item  $(\mathbf{x} , \mathbf{x}) = 4$ and $\alpha^2 + \beta^2 + 4  \gamma^2  = 0$.
\end{enumerate}
\end{theorem}

\begin{proof}
Let  $\mathbf{y}_{0} =  ( \begin{array}{ccc}
\alpha  &   \beta   &   \gamma  \mathbf{x}  \end{array}  )$ be the first row of  $G$, and let $\mathbf{y}_{i} =  ( \begin{array}{ccc}
(\mathbf{x} , \mathbf{r}_i)   &   2  (\mathbf{x} , \mathbf{r}_i)     &   \mathbf{r}_i   \end{array}  )$ be the $(i+1)$th row of $G$ for $i = 1 , 2 , \dotsc , m$.
\begin{enumerate}[(i)]
\item  If $(\mathbf{x} , \mathbf{x}) = 1$ and $\alpha^2 + \beta^2 + \gamma^2  = 0$, then
\begin{align*}
(\mathbf{y}_0 , \mathbf{y}_0 )  &=  \alpha^2 + \beta^2 + \gamma^2  ( \mathbf{x} , \mathbf{x}  )  = \alpha^2 + \beta^2 + \gamma^2  =  0  ,  \\
(\mathbf{y}_0 , \mathbf{y}_i )  &=  \alpha  ( \mathbf{x}  ,  \mathbf{r}_i  )  +  2 \beta ( \mathbf{x}  ,  \mathbf{r}_i  )  +  \gamma  ( \mathbf{x}  ,  \mathbf{r}_i  )    =  (  \alpha +  2 \beta + \gamma  )  ( \mathbf{x}  ,  \mathbf{r}_i  )  =  0  ,  \text{ and}   \\
(\mathbf{y}_i , \mathbf{y}_j )  &=    ( \mathbf{x}  ,  \mathbf{r}_i  )   ( \mathbf{x}  ,  \mathbf{r}_j  ) +  4  ( \mathbf{x}  ,  \mathbf{r}_i  )   ( \mathbf{x}  ,  \mathbf{r}_j  )  +  ( \mathbf{r}_i  ,  \mathbf{r}_j  )  =   0
\end{align*}
for all $i , j = 1 , 2 , \dotsc , m$.

\item  If $(\mathbf{x} , \mathbf{x}) = 2$ and $\alpha^2 + \beta^2 +2   \gamma^2  = 0$, then
\begin{align*}
(\mathbf{y}_0 , \mathbf{y}_0 )  &=  \alpha^2 + \beta^2 + \gamma^2  ( \mathbf{x} , \mathbf{x}  )  = \alpha^2 + \beta^2 + 2  \gamma^2  =  0  ,  \\
(\mathbf{y}_0 , \mathbf{y}_i )  &=  \alpha  ( \mathbf{x}  ,  \mathbf{r}_i  )  +  2 \beta ( \mathbf{x}  ,  \mathbf{r}_i  )  +  \gamma  ( \mathbf{x}  ,  \mathbf{r}_i  )    =  (  \alpha +  2 \beta + \gamma  )  ( \mathbf{x}  ,  \mathbf{r}_i  )  =  0   ,   \text{ and}  \\
(\mathbf{y}_i , \mathbf{y}_j )  &=    ( \mathbf{x}  ,  \mathbf{r}_i  )   ( \mathbf{x}  ,  \mathbf{r}_j  ) +  4  ( \mathbf{x}  ,  \mathbf{r}_i  )   ( \mathbf{x}  ,  \mathbf{r}_j  )  +  ( \mathbf{r}_i  ,  \mathbf{r}_j  )  =   0
\end{align*}
for all $i , j = 1 , 2 , \dotsc , m$.

\item  If $(\mathbf{x} , \mathbf{x}) = 3$ and $\alpha^2 + \beta^2 + 3   \gamma^2  = 0$, then
\begin{align*}
(\mathbf{y}_0 , \mathbf{y}_0 )  &=  \alpha^2 + \beta^2 + \gamma^2  ( \mathbf{x} , \mathbf{x}  )  = \alpha^2 + \beta^2 + 3  \gamma^2  =  0   , \\
(\mathbf{y}_0 , \mathbf{y}_i )  &=  \alpha  ( \mathbf{x}  ,  \mathbf{r}_i  )  +  2 \beta ( \mathbf{x}  ,  \mathbf{r}_i  )  +  \gamma  ( \mathbf{x}  ,  \mathbf{r}_i  )    =  (  \alpha +  2 \beta + \gamma  )  ( \mathbf{x}  ,  \mathbf{r}_i  )  =  0   ,   \text{ and} \\
(\mathbf{y}_i , \mathbf{y}_j )  &=    ( \mathbf{x}  ,  \mathbf{r}_i  )   ( \mathbf{x}  ,  \mathbf{r}_j  ) +  4  ( \mathbf{x}  ,  \mathbf{r}_i  )   ( \mathbf{x}  ,  \mathbf{r}_j  )  +  ( \mathbf{r}_i  ,  \mathbf{r}_j  )  =   0
\end{align*}
for all $i , j = 1 , 2 , \dotsc , m$.

\item  If $(\mathbf{x} , \mathbf{x}) = 4$ and $\alpha^2 + \beta^2 + 4 \gamma^2  = 0$, then
\begin{align*}
(\mathbf{y}_0 , \mathbf{y}_0 )  &=  \alpha^2 + \beta^2 + \gamma^2  ( \mathbf{x} , \mathbf{x}  )  = \alpha^2 + \beta^2 + 4 \gamma^2  =  0   ,   \\
(\mathbf{y}_0 , \mathbf{y}_i )  &=  \alpha  ( \mathbf{x}  ,  \mathbf{r}_i  )  +  2 \beta ( \mathbf{x}  ,  \mathbf{r}_i  )  +  \gamma  ( \mathbf{x}  ,  \mathbf{r}_i  )    =  (  \alpha +  2 \beta + \gamma  )  ( \mathbf{x}  ,  \mathbf{r}_i  )  =  0   ,  \text{ and}  \\
(\mathbf{y}_i , \mathbf{y}_j )  &=    ( \mathbf{x}  ,  \mathbf{r}_i  )   ( \mathbf{x}  ,  \mathbf{r}_j  ) +  4  ( \mathbf{x}  ,  \mathbf{r}_i  )   ( \mathbf{x}  ,  \mathbf{r}_j  )  +  ( \mathbf{r}_i  ,  \mathbf{r}_j  )  =   0
\end{align*}
for all $i , j = 1 , 2 , \dotsc , m$.

\end{enumerate}
In all cases, we see that $\mathcal{C}$ is a self-orthogonal code. Moreover, $| \mathcal{C} | = 25^1  | \mathcal{C}_0 | = 5^{n+2}$. Thus, $\mathcal{C}$ is a quasi self-dual code of length $n + 2$.
\end{proof}

\begin{theorem}  \label{thm: QSD code of length n+2 (2 generators)}
Let $\mathcal{C}_{0}$ be a quasi self-dual code of length $n$ over $I_{5}$ with generator matrix $G_{0}=\left(\mathbf{r}_{i}\right)$, where $r_{i}$ is the $i$th row of $G_{0}$ for $i = 1 , 2 , \dotsc , m$. If $\mathbf{x} \in \mathbb{F}_{5}^{n}$ and $\mu, \nu \in J = \{  0, \mathtt{b} , \mathtt{d} , \mathtt{i} , \mathtt{n} \}$ are nonzero elements, then the code $\mathcal{C}$ with generator matrix
\begin{align*}
G=\left(\begin{array}{cc|c}
4 \mu & 0  & \mu \mathbf{x} \\
0 & 2 \nu   & \nu \mathbf{x} \\   \hline
(\mathbf{x} , \mathbf{r}_1)   &   2  (\mathbf{x} , \mathbf{r}_1)     &   \mathbf{r}_1    \\
\vdots   &   \vdots   &   \vdots   \\
(\mathbf{x} , \mathbf{r}_m)   &   2  (\mathbf{x} , \mathbf{r}_m)     &   \mathbf{r}_m
\end{array}\right)
\end{align*}
is a quasi self-dual code of length $n + 2$.
\end{theorem}
\begin{proof}
Let $\mathbf{y}_{0} =  \left(\begin{array}{ccc}  4 \mu & 0  & \mu  \mathbf{x}   \end{array}   \right)$, $\mathbf{y}_{0}^{\prime} =  \left(\begin{array}{cccc}  0 & 2 \nu   & \nu \mathbf{x}    \end{array}   \right)$, and $\mathbf{y}_{i} =  ( \begin{array}{ccc}     (\mathbf{x} , \mathbf{r}_i)   &   2  (\mathbf{x} , \mathbf{r}_i)     &   \mathbf{r}_i   \end{array}  )$. For $i , j = 1 , 2 , \dotsc , m$, we have
\begin{align*}
(  \mathbf{y}_0 , \mathbf{y}_0  )  &=  \mu^2 + \mu^2 ( \mathbf{x} , \mathbf{x}  )  =  0,   \\
(  \mathbf{y}_0 , \mathbf{y}_0^{\prime}  )  &= \mu \nu  ( \mathbf{x} , \mathbf{x}  )  =  0,   \\
(  \mathbf{y}_0^{\prime} , \mathbf{y}_0^{\prime}  )  &= 4 \nu^2 + \nu^2 ( \mathbf{x} , \mathbf{x}  )  =  0,   \\
(  \mathbf{y}_0 , \mathbf{y}_i  )  &=  4 \mu ( \mathbf{x} , \mathbf{r}_i  ) + \mu ( \mathbf{x} , \mathbf{r}_i  )  =  0  ,   \\
(  \mathbf{y}_0^{\prime} , \mathbf{y}_i  )  &=  4 \nu ( \mathbf{x} , \mathbf{r}_i  ) + \nu ( \mathbf{x} , \mathbf{r}_i  )  =  0  ,  \text{ and}   \\
(\mathbf{y}_i , \mathbf{y}_j )  &=    ( \mathbf{x}  ,  \mathbf{r}_i  )   ( \mathbf{x}  ,  \mathbf{r}_j  ) +  4  ( \mathbf{x}  ,  \mathbf{r}_i  )   ( \mathbf{x}  ,  \mathbf{r}_j  )  +  ( \mathbf{r}_i  ,  \mathbf{r}_j  )  =   0
\end{align*}
as $\mu , \nu \in J$. Thus, $\mathcal{C}$ is a self-orthogonal code. Furthermore, since $|\mathcal{C}|=5^{2} \left|\mathcal{C}_{0}\right| = 5^{n+2}, \mathcal{C}$ is a quasi self-dual code of length $n+2$.
\end{proof}

\begin{example}
The code $J = \{ 0 , \mathtt{b} ,  \mathtt{d} ,  \mathtt{i}  ,  \mathtt{n}   \}$ is a quasi self-dual $(1, 5^1, 1)$ code and has generator matrix given by $G_1 = \left( \mathtt{b}  \right)$. Note that the code $J$ is of type $\{ 0 , 1 \}$ and has weight distribution $[ \langle 0 , 1  \rangle  ,  \langle  1  ,  4  \rangle  ]$.

Using Theorem \ref{thm: QSD code of length n+2 (1 generator)} from the matrix $G_1$ with $\alpha = 0$, $\beta = \mathtt{a}$, $\gamma = \mathtt{j}$, and $\mathbf{x} = 1$, we obtain a  self-orthogonal $(3, 5^3, 1)$ code with generator matrix
\begin{align*}
G_{1,1} = \left(
\begin{array}{cc|c}
0  &  \mathtt{a}  &  \mathtt{j}  \\
\hline
\mathtt{b}  &  \mathtt{d}  &  \mathtt{b}
\end{array}
\right) .
\end{align*}
This code is of type $\{ 1 , 1 \}$ and has weight distribution $[ \langle  0 , 1  \rangle , \langle  1 , 4   \rangle , \langle  2 , 24 \rangle ,  \langle 3 , 96  \rangle ]$.

Using Theorem \ref{thm: QSD code of length n+2 (2 generators)}, from $G_1$ with $\mu = \mathtt{n}$ and $\nu = \mathtt{i}$ and $\mathbf{x} = 1$, we obtain another quasi self-dual $(3, 5^3, 1)$ code with generator matrix
	\begin{align*}
		G_{1 , 2}  =   \left(
		\begin{array}{cc|c}
		\mathtt{b} &  0    &    \mathtt{n}       \\
		0 &  \mathtt{b}    &    \mathtt{i}       \\
		\hline
		\mathtt{b} &  \mathtt{d}    &    \mathtt{b}
		\end{array}
		\right)
		\end{align*}
and weight distribution $[ \langle 0 , 1 \rangle  ,   \langle  1  ,  12  \rangle  ,  \langle  2 ,  48 \rangle  ,  \langle  3  ,  64  \rangle ]$. This code is of type $\{ 0 , 3 \}$ and is monomially equivalent to the code with generator matrix
$
		\left(
		\begin{array}{ccc}
		\mathtt{b} &  0    &    0       \\
		0 &  \mathtt{b}    &    0       \\
		0 &  0    &    \mathtt{b}
		\end{array}
		\right)
$.
\end{example}

\begin{cor}  \label{cor: QSD code of length n+4 (3 generators)}
Let $\mathcal{C}_0$ be a quasi self-dual code over $I_5$ of length $n$ with generator matrix $G_0 = (\mathbf{r}_i)$ where $(\mathbf{r}_i)$ is the $i$th row of $G_0$ for $i =  1 , 2 , \dotsc , m$. Let $\mathbf{x} \in \mathbb{F}_5^n$ and $\alpha , \beta , \gamma, \delta  \in I_5$, not all zero, such that $\alpha + \beta + 2 \gamma + \delta  = 0$. Let $\iota, \kappa, \mu , \nu  \in J = \{ 0 , \mathtt{b} , \mathtt{d} , \mathtt{i} , \mathtt{n}   \}$, not all zero, such that $\iota + \kappa = \mu + \nu = 0$. Then the code $\mathcal{C}$ with generator matrix
			\begin{align*}
			G = \left(
			\begin{array}{cccc|c}
			\alpha             & \beta      & \gamma   &       0                   & \delta   \mathbf{x} \\
			0                     &   \iota      &         0              &     0                 &    \kappa \mathbf{x}     \\
			0                     &   0     &         \mu                &     0                 &    \nu   \mathbf{x}     \\
			\hline
			(\mathbf{x} , \mathbf{r}_1) & (\mathbf{x} , \mathbf{r}_1) & 2 (\mathbf{x} , \mathbf{r}_1)  & 2 (\mathbf{x} , \mathbf{r}_1)      & \mathbf{r}_1      \\
			\vdots            & \vdots            & \vdots                 & \vdots    & \vdots     \\
			(\mathbf{x} , \mathbf{r}_m) & (\mathbf{x} , \mathbf{r}_m) & 2 (\mathbf{x} , \mathbf{r}_m)   & 2 (\mathbf{x} , \mathbf{r}_m)      & \mathbf{r}_m
			\end{array}
			\right)
			\end{align*}
			is a quasi self-dual code of length $n + 4$  if one of the following conditions is satisfied:
			\begin{enumerate}[(i)]
			\item  $(\mathbf{x} , \mathbf{x}) = 1$ and $\alpha^2 + \beta^2 + \gamma^2  + \delta^2  = 0$,
			\item  $(\mathbf{x} , \mathbf{x}) = 2$ and $\alpha^2 + \beta^2 + \gamma^2  + 2 \delta^2  = 0$,
			\item  $(\mathbf{x} , \mathbf{x}) = 3$ and $\alpha^2 + \beta^2 + \gamma^2  + 3 \delta^2  = 0$, or
			\item  $(\mathbf{x} , \mathbf{x}) = 4$ and $\alpha^2 + \beta^2 + \gamma^2  + 4  \delta^2 = 0$.
			\end{enumerate}
\end{cor}

\begin{proof}
Let $\mathbf{y}_{0}= ( \begin{array}{ccccc}  \alpha       & \beta      & \gamma       &   0                  & \delta  \mathbf{x}     \end{array}  )$, $\mathbf{y}_0' =  \left(  \begin{array}{ccccc}      0                     &   \iota      &         0              &     0                 &    \kappa \mathbf{x}           \end{array}   \right)$, $\mathbf{y}_0'' =  \left(  \begin{array}{ccccc}      0                     &   0     &         \mu                   &     0                 &    \nu   \mathbf{x}            \end{array}   \right)$,
and let $\mathbf{y}_i = \left(  \begin{array}{ccccc}      (\mathbf{x} , \mathbf{r}_i) & (\mathbf{x} , \mathbf{r}_i) & 2 (\mathbf{x} , \mathbf{r}_i)  & 2 (\mathbf{x} , \mathbf{r}_i)     & \mathbf{r}_i           \end{array}   \right) $. We have
\begin{align*}
( \mathbf{y}_0 , \mathbf{y}_0' ) &=  \beta \iota + \delta \kappa  ( \mathbf{x} , \mathbf{x} ) = 0
\end{align*}
since $\beta \iota = \delta \kappa = 0$. Similar computations show that $\mathbf{y}_0$, $\mathbf{y}_0'$, and $\mathbf{y}_0''$ are orthogonal to one another and to themselves since $\iota, \kappa, \mu , \nu \in J$. From Theorem \ref{thm: SO code of length n+4 (1 generator)}, we have already seen that $( \mathbf{y}_0 , \mathbf{y}_i ) = 0$ and $( \mathbf{y}_i , \mathbf{y}_j ) = 0$. Now,
\begin{align*}
( \mathbf{y}_0' , \mathbf{y}_i ) &=  \iota  ( \mathbf{x} , \mathbf{r}_i  )  +   \kappa  ( \mathbf{x} , \mathbf{r}_i  )   =  ( \iota + \kappa  ) ( \mathbf{x} , \mathbf{r}_i  )     =  0  \   \text{ and}   \\
( \mathbf{y}_0'' , \mathbf{y}_i ) &=  \mu  ( \mathbf{x} , \mathbf{r}_i  )  +   \nu  ( \mathbf{x} , \mathbf{r}_i  )   =  ( \mu + \nu  ) ( \mathbf{x} , \mathbf{r}_i  )     =  0 .
\end{align*}
Thus, $\mathcal{C}$ is self-orthogonal. Moreover, since $| \mathcal{C} | = 25^1 5^2 | \mathcal{C}_0 | = 5^{n+4}$, $\mathcal{C}$ is a quasi self-dual code over $I_5$  of length $n+4$.
\end{proof}

\begin{example}
Using Corollary \ref{cor: QSD code of length n+4 (3 generators)} with the matrix $G_{1,2}$, $\alpha = \mathtt{b}$, $\beta = \mathtt{q}$, $\gamma = \mathtt{a}$, $\delta = \mathtt{e}$, $\iota = \mathtt{d}$, $\kappa = \mathtt{i}$, $\mu = \mathtt{b}$, $\nu = \mathtt{n}$, and $\mathbf{x} = ( 1, 1, 0)$, we obtain a quasi self-dual $(7, 5^7, 1)$ code with generator matrix
\begin{align*}
G_{1,2,1}  =
\left(
\begin{array}{cccc|ccc}
\mathtt{b}  &  \mathtt{q}  &  \mathtt{a}  &  0  &  \mathtt{e}   &  \mathtt{e}  &  0  \\
0  &  \mathtt{d}  &  0  &  0  &  \mathtt{i}   &  \mathtt{i}  &  0  \\
0  &  0  &  \mathtt{b}  &  0  &  \mathtt{n}   &  \mathtt{n}  &  0  \\
\hline
\mathtt{b}  &  \mathtt{b}  &  \mathtt{d}  &  \mathtt{d}  &  \mathtt{b}   &  0 &  \mathtt{n}  \\
\mathtt{b}  &  \mathtt{b}  &  \mathtt{d}  &  \mathtt{d}  &  0   &  \mathtt{b} &  \mathtt{i}  \\
\mathtt{i}  &  \mathtt{i}  &  \mathtt{b}  &  \mathtt{b}  &  \mathtt{b}   &  \mathtt{d} &  \mathtt{b}
\end{array}
\right)
\end{align*}
and weight distribution $[  \langle 0 , 1  \rangle ,  \langle  1 , 20  \rangle ,  \langle  2 , 164 \rangle ,  \langle  3 , 720 \rangle ,  \langle  4 ,  1920  \rangle ,  \langle   5  ,  8584   \rangle ,  \langle   6  ,  32620 \rangle ,   \langle   7 ,  34096 \rangle  ]$.
\end{example}

The next theorem is an extension of Theorem 6 in \cite{I3} from the ring $I_3$ to the ring $I_5$ and gives a propagation rule of order $5$.

\begin{theorem}   \label{thm: QSD codes of length n+5}
Let $\mathcal{C}_{0}$ be a quasi self-dual code of length $n$ over $I_{5}$ with generator matrix $G_{0}=\left(\mathbf{r}_{i}\right)$, where $r_{i}$ is the $i$th row of $G_{0}$ for $i = 1 , 2 , \dotsc , m$. If $\mathbf{x} \in \mathbb{F}_{5}^{n}$ and $\iota , \kappa ,  \lambda , \mu , \nu   \in J = \{0, \mathtt{b} , \mathtt{d} , \mathtt{i} , \mathtt{n}  \}$ are nonzero elements, then the code $\mathcal{C}$ with generator matrix
\begin{align*}
G=\left(\begin{array}{ccccc|c}
\iota & 0 & 0 & 0  & 0   & \iota \mathbf{x} \\
0 &  \kappa & 0 & 0 &  0   & \kappa \mathbf{x} \\
0 &  0 & \lambda & 0 &  0   & \lambda  \mathbf{x} \\
0 &  0 & 0 & \mu &  0   & \mu  \mathbf{x} \\
0 &  0 & 0 & 0 &  \nu   & \nu  \mathbf{x} \\     \hline
4 (\mathbf{x} , \mathbf{r}_1)   &   4  (\mathbf{x} , \mathbf{r}_1)  & 4 (\mathbf{x} , \mathbf{r}_1)  &  4 (\mathbf{x} , \mathbf{r}_1)  &  4 (\mathbf{x} , \mathbf{r}_1)    &   \mathbf{r}_1    \\
\vdots   &   \vdots   &   \vdots  &  \vdots   &   \vdots   &   \vdots   \\
4 (\mathbf{x} , \mathbf{r}_m)   &   4  (\mathbf{x} , \mathbf{r}_m)  & 4 (\mathbf{x} , \mathbf{r}_m)  &  4 (\mathbf{x} , \mathbf{r}_m)  &  4 (\mathbf{x} , \mathbf{r}_m)    &   \mathbf{r}_m
\end{array}\right)
\end{align*}
is a quasi self-dual code of length $n + 5$.
\end{theorem}
\begin{proof}
Let $\mathbf{y}_{0} = \left(\begin{array}{cccccc}  \iota & 0 & 0 & 0  & 0   & \iota \mathbf{x}    \end{array}   \right)$, $\mathbf{y}_{0}^{\prime} = \left(\begin{array}{cccccc}   0 &  \kappa & 0 & 0 &  0   & \kappa \mathbf{x}  \end{array}   \right)$, $\mathbf{y}_{0}^{\prime\prime} = \left(\begin{array}{cccccc}   0 &  0 & \lambda & 0 &  0   & \lambda  \mathbf{x}  \end{array}   \right)$,  $\mathbf{y}_{0}^{\prime\prime\prime} = \left(\begin{array}{cccccc}  0 &  0 & 0 & \mu &  0   & \mu  \mathbf{x}  \end{array}   \right)$, and  $\mathbf{y}_{0}^{\prime\prime\prime\prime} = \left(\begin{array}{cccccc}  0 &  0 & 0 & 0 &  \nu   & \nu  \mathbf{x}  \end{array}  \right)$. Observe that
\begin{align*}
(  \mathbf{y}_0 , \mathbf{y}_0^{\prime}  )  =   \iota  \kappa  (  \mathbf{x} ,  \mathbf{x}  )  =  0
\end{align*}
since $\iota \kappa = 0$. Similarly, it can be shown that $\mathbf{y}_0$, $\mathbf{y}_0^\prime$, $\mathbf{y}_0^{\prime\prime}$, $\mathbf{y}_0^{\prime\prime\prime}$, and $\mathbf{y}_0^{\prime\prime\prime\prime}$ are orthogonal to each other and to themselves because $\iota , \kappa ,  \lambda , \mu , \nu   \in J = \{0, \mathtt{b}, \mathtt{d} , \mathtt{i} , \mathtt{n} \}$.

Now, let $\mathbf{y}_{i} =  \left(\begin{array}{cccccc}  4 (\mathbf{x} , \mathbf{r}_i)   &   4  (\mathbf{x} , \mathbf{r}_i)  & 4 (\mathbf{x} , \mathbf{r}_i)  &  4 (\mathbf{x} , \mathbf{r}_i)  &  4 (\mathbf{x} , \mathbf{r}_i)    &   \mathbf{r}_i     \end{array}   \right)$ for $i = 1 , 2 , \dotsc , m$. Then we have
\begin{align*}
(  \mathbf{y}_0 , \mathbf{y}_i  ) =  4 \iota   (\mathbf{x} , \mathbf{r}_i)  +  \iota  (\mathbf{x} , \mathbf{r}_i)  =  0  .
\end{align*}
Similarly, it can be shown that $(  \mathbf{y}_0^{\prime} , \mathbf{y}_i  ) = 0$,  $(  \mathbf{y}_0^{\prime\prime} , \mathbf{y}_i  ) = 0$, $(  \mathbf{y}_0^{\prime\prime\prime} , \mathbf{y}_i  ) = 0$,  and  $(  \mathbf{y}_0^{\prime\prime\prime\prime} , \mathbf{y}_i  ) = 0$. Moreover,
\begin{align*}
(  \mathbf{y}_i , \mathbf{y}_j  )   =  5  (\mathbf{x} , \mathbf{r}_i)  (\mathbf{x} , \mathbf{r}_j)  +   (\mathbf{r}_i , \mathbf{r}_j)  =  0
\end{align*}
for $j = 1 , 2 , \dotsc , m$.
Hence, $\mathcal{C}$ is a self-orthogonal code. Furthermore, since $|\mathcal{C}|=5^{5}\left|\mathcal{C}_{0}\right|=5^{n+5}, \mathcal{C}$ is a quasi self-dual code over $I_5$ of length $n + 5$.
\end{proof}

\subsection{Construction of Self-Dual Codes}

We start with some characterizations of self-dual codes over $I_5$.

\begin{theorem}  \label{thm: SD code iff resC is SD and torC is F_5^n}
A linear code $\mathcal{C}$ of length $n$ over $I_{5}$ is self-dual if and only if $\operatorname{res}(\mathcal{C})$ is a self-dual code over $\mathbb{F}_5$ and $\operatorname{tor}(\mathcal{C})=\mathbb{F}_{5}^{n}$. Furthermore, the code $\mathcal{C}$ is given by the relation $\mathcal{C} = \mathtt{a} \operatorname{res} (\mathcal{C}) + \mathtt{b} \mathbb{F}_5^n$.  Additionally, $|\mathcal{C}| = 5^{\frac{3n}{2}}$.
\end{theorem}

\begin{proof}
For the forward implication, suppose that $\mathcal{C}$ is self-dual, i.e., $\mathcal{C} = \mathcal{C}^\perp$. Then $\operatorname{res} ( \mathcal{C} ) = \operatorname{res}  ( \mathcal{C}^\perp )$ and $\operatorname{tor} ( \mathcal{C} ) = \operatorname{tor} ( \mathcal{C}^\perp  ) $.  By Theorem \ref{thm: residue and torsion of dual code Cperp}, we have  $\operatorname{res} ( \mathcal{C}^\perp ) = \operatorname{res} ( \mathcal{C} )^\perp$ and $\operatorname{tor} ( \mathcal{C} ) = \mathbb{F}_5^n$.

For the reverse implication, suppose that $\operatorname{res} ( \mathcal{C} )$ is self-dual and $\operatorname{tor} ( \mathcal{C} ) =  \mathbb{F}_5^n$. By Theorems \ref{thm: c=au+bv for any c in C} and \ref{thm: dual of a code over I5}, we have
\begin{align*}
\mathcal{C}  \subseteq  \mathtt{a}  \operatorname{res}  ( \mathcal{C} ) + \mathtt{b}  \mathbb{F}_5^n  =  \mathtt{a}  \operatorname{res}  (  \mathcal{C} )^\perp  +  \mathtt{b} \mathbb{F}_5^n  =  \mathcal{C}^\perp  .
\end{align*}
Since $| \mathcal{C} | =  | \operatorname{res}  ( \mathcal{C} ) |   | \operatorname{tor}  ( \mathcal{C} ) |  =  | \operatorname{res}  ( \mathcal{C} )^\perp |  | \mathbb{F}_5^n | =  | \mathcal{C}^\perp |$, it follows that $| \mathcal{C} | = | \mathcal{C}^\perp |$, i.e., $\mathcal{C} = \mathcal{C}^\perp$. Hence, $\mathcal{C}$ is self-dual. Also, $| \mathcal{C} | = | \operatorname{res} ( \mathcal{C}  ) |    | \operatorname{tor} ( \mathcal{C}  ) | = 5^{\frac{n}{2}} 5^n = 5^{\frac{3n}{2}}$.
\end{proof}

\begin{remark}
Theorem \ref{thm: SD code iff resC is SD and torC is F_5^n} tells us that the length of a self-dual code over $I_5$ is always even.
\end{remark}

\begin{cor}
The minimum distance of a non-trivial self-dual code over $I_5$ is $1$.
\end{cor}

\begin{proof}
This is immediate from Theorems \ref{thm: minimum distance of I5 code C} and \ref{thm: SD code iff resC is SD and torC is F_5^n}.
\end{proof}

We now present a building-up construction for self-dual codes over $I_5$.

\begin{theorem}   \label{thm: SD code of length n+6 (5 generators)}
Let $n$ be even and let $\mathcal{C}_{0}$ be a self-dual code of length $n$ over $I_{5}$ with generator matrix $G_{0}=\left(\mathbf{r}_{i}\right)$, where $r_{i}$ is the $i$th row of $G_{0}$ for $i = 1 , 2 , \dotsc , m$. Let $\mathbf{x}_1 , \mathbf{x}_2  \in \mathbb{F}_{5}^{n}$ such that $( \mathbf{x}_1 , \mathbf{x}_2 ) = 0$ and $( \mathbf{x}_i , \mathbf{x}_i ) = 4$ for $i = 1 , 2$. For $1 \leq i \leq m$, define $u_i = ( \mathbf{x}_1 , \mathbf{r}_i ) $ and $v_i = ( \mathbf{x}_2 , \mathbf{r}_i ) $. Then the code $\mathcal{C}$ with generator matrix
\begin{align*}
G=\left(\begin{array}{cccccc|c}
\mathtt{a} & 0 & 0 & 0  & 0  & 0  & \mathtt{a} \mathbf{x}_1    \\
0 &  \mathtt{a} & 0 & 0 &  0  & 0  & \mathtt{a} \mathbf{x}_2    \\
0 &  0 & \mathtt{b} & 0 &  0  & 0  &    \mathtt{n}   \mathbf{x}_2    \\
0 &  0 & 0 & \mathtt{b} &  0 & 0   & \mathtt{b} ( 4 \mathbf{x}_1 +  4 \mathbf{x}_2 )    \\
0 &  0 & 0 & 0 &  \mathtt{b}  & 0  & \mathtt{b} ( 3 \mathbf{x}_1 + 4 \mathbf{x}_2 )    \\
0 &  0 & 0 & 0 & 0  &  \mathtt{b}  & \mathtt{b} ( 2 \mathbf{x}_1 +  \mathbf{x}_2 )    \\
 \hline
4u_1   &   4v_1  &  v_1 & u_1 + v_1  &  2 u_1 + v_1  & 3 u_1 + 4 v_1    &   \mathbf{r}_1    \\
\vdots   &   \vdots   &   \vdots  &  \vdots   &   \vdots   &   \vdots  &  \vdots   \\
4u_m   &   4  v_m  &  v_m  & u_m + v_m  &  2 u_m + v_m  & 3 u_m + 4 v_m     &   \mathbf{r}_m
\end{array}\right)
\end{align*}
is a self-dual code of length $n + 6$.
\end{theorem}

\begin{proof}
It can be shown that $\mathcal{C}$ is self-orthogonal. Now, let $\widehat{  \mathcal{C}_0  }$ be  the span of the last $m$ generators. Furthermore, let $D$ be the $\mathbb{F}_{5}$-span of $\left\{\mathbf{v}_{1}, \mathbf{v}_{2}, \mathbf{v}_{3}, \mathbf{v}_{4} , \mathbf{v}_5 , \mathbf{v}_{6}  \right\}$ where $\mathbf{v}_1 = (1 , 0 , 0 , 0 , 0 , 0 , \mathbf{x}_{1}  )$, $\mathbf{v}_2 = (0 , 1 , 0 , 0 , 0 , 0 ,  \mathbf{x}_{2} )$, $\mathbf{v}_3 = (0 , 0 , 1 , 0 , 0 , 0 , 4 \mathbf{x}_{2} )$,  $\mathbf{v}_4 = ( 0,0,0,1, 0 , 0 , ( 4 \mathbf{x}_{1} + 4 \mathbf{x}_{2} )  )$,  $\mathbf{v}_{5} = ( 0 , 0 , 0 , 0 , 1 , 0 , ( 3 \mathbf{x}_{1} + 4 \mathbf{x}_{2} )  )$, and $\mathbf{v}_6 = ( 0 , 0 , 0 , 0 , 0 , 1 ,  (  2 \mathbf{x}_1    +  \mathbf{x}_2 )   )$. Thus, the construction of $\operatorname{tor}(\mathcal{C})$ in the theorem is equivalent to
\begin{align*}
\operatorname{tor}  (\mathcal{C}) = D + \operatorname{tor} ( \widehat{\mathcal{C}_{0}} )  .
\end{align*}
By Theorem \ref{thm: SD code iff resC is SD and torC is F_5^n}, it follows that $\operatorname{tor}\left(\mathcal{C}_{0}\right) = \mathbb{F}_5^n$ since $\mathcal{C}_0$ is self-dual. Thus, we can infer that there is a ratio of $5^6$ on the sizes since $|D|=5^6$ and $\operatorname{tor}(\mathcal{C})=\mathbb{F}_{5}^{n+6}$.

Now, it suffices to show that $\operatorname{res}(\mathcal{C})$ is a self-dual  code over $\mathbb{F}_5$. Since $\mathcal{C}$ is self-orthogonal, $\operatorname{res}(\mathcal{C})$ is also self-orthogonal, i.e., $\operatorname{res}(\mathcal{C}) \subseteq \operatorname{res}(\mathcal{C})^{\perp}$. By Theorem \ref{thm: SD code iff resC is SD and torC is F_5^n}, $| \mathcal{C}_{0} |=5^{\frac{3 n}{2}}$  and because $G$ has six additional generators, it follows that $| \mathcal{C}|=5^{3} 25^{3} | \mathcal{C}_{0} | = 5^{\frac{3(n+6)}{2}}$. Thus, $|\operatorname{res}(\mathcal{C})|= | \operatorname{res}(\mathcal{C})^{\perp} |=5^{\frac{n+6}{2}}$ and so $\operatorname{res}(\mathcal{C})$ is a self-dual code over $\mathbb{F}_5$.

Hence, $\mathcal{C} = \mathtt{a} \operatorname{res}(\mathcal{C}) + \mathtt{b} \mathbb{F}_{5}^{n+6}$ and the code $\mathcal{C}$ obtained from the self-dual code $\mathcal{C}_{0}$ by the building-up construction is also self-dual.
\end{proof}

The following theorem is a special case of Lemma 12 in \cite{MassFormula_Ip} when dealing with $p=5$.

\begin{theorem}
Two self-dual codes $\mathcal{C}_1$ and $\mathcal{C}_2$ over $I_5$ are monomially equivalent if and only if their residue codes are equivalent.
\end{theorem}

\section{Computational Results}


In this section, we present the classification of self-orthogonal, quasi self-dual, and self-dual $I_5$-codes of length $n \leq 4$ using our construction methods. All of the computations were done using \texttt{MAGMA} \cite{MAGMA} and the classification is shown in Table \ref{tbl: inequiv SO codes over I5} below. Codes whose types are marked with $\dagger$  give the corrected versions of codes in Table 3 of \cite{MassFormula_Ip}.  For codes whose types are marked with $\S$, the complete list of codes can be accessed at \url{https://sites.google.com/view/marvinolavides/papers/classification-of-so-codes-over-i5}.

\begingroup
\setlength{\tabcolsep}{0.15pt}
{  \fontsize{8.65pt}{9.0pt}\selectfont
\begin{longtable}{ccccc@{\hspace{10pt}}cccc}
\caption*{Inequivalent self-orthogonal codes of length $n \leq 4$ over $I_5$}  
\label{tbl: inequiv SO codes over I5}  
\endfirsthead
\endhead
\thickhline
{   }    \\[-0.5em]
$\boldsymbol{n}$   
& \textbf{Type}    
& $\begin{array}{c} \textbf{Code}   \end{array}$       
& $\begin{array}{c} \textbf{Number of} \\ \textbf{Distinct Codes}   \end{array}$             
& $\boldsymbol{| \operatorname{Aut} (\mathcal{C}) |}$ 
& $\boldsymbol{d(\mathcal{C})}$
& $\begin{array}{c} \textbf{Weight} \\  \textbf{Distribution} \end{array}$                                                 
& \textbf{SO/QSD/SD}     \\

\phantom{text}  \\[-0.5em]
\thickhline   
\phantom{text}  \\[-0.5em]

$1$             
& $\{ 0 , 1 \}$    
& $( \mathtt{b} )$    
& $1$            
& $2$       
 & $1$                                           
& $[ \langle 0 , 1 \rangle , \langle 1 , 4 \rangle ]$                                                                      
& QSD                \\

\phantom{text}  \\[-0.75em]
\thickhline   
\phantom{text}  \\[-0.5em]

$2$     
&  $\{  0 , 1  \}$
&  $\left( \begin{array}{cc} \mathtt{b}  &  0   \end{array} \right)$         
& $1$ 
&  $4$
& $1$
&   $[ \langle  0 ,  1   \rangle ,   \langle   1  ,  4    \rangle   ]$   
&   SO  \\

\phantom{text}  \\[-0.5em]

{   }   
&  {    }
&  $\left( \begin{array}{cc} \mathtt{b}  &  \alpha   \end{array} \right)$          
&  $2$
&  $4$
& $2$
&   $[ \langle  0 ,  1   \rangle ,   \langle   2  ,  4    \rangle   ]$   
&   SO  \\
{   }    &   {   }    &   where $\alpha \in \{  \mathtt{b} , \mathtt{d}  \}$   \\

\phantom{text}  \\[-0.75em]
\cline{2-8}
\phantom{text}  \\[-0.5em]

{    }     
& $\{ 0 , 2 \}$    
& $\left( \begin{array}{cc} \mathtt{b}  &  0   \\   0   &   \mathtt{b}  \end{array} \right)$    
&  $1$                                                                                    
& $8$    
& $1$                                             
& $[ \langle  0  ,  1  \rangle   ,   \langle  1 ,  8  \rangle  ,  \langle  2  ,   16 \rangle ]$                            
& QSD                \\

\phantom{text}  \\[-0.75em]
\cline{2-8}
\phantom{text}  \\[-0.5em]

 {    }       
&  $\{  1 , 0  \}^\dagger$  
&  $\left( \begin{array}{cc} \mathtt{a}  &  \mathtt{e} \end{array} \right)$  
&   $1$ 
&  $4$           
& $2$                                  
&  $[ \langle 0 , 1 \rangle , \langle 2 , 24 \rangle ]$                                                               
& QSD                \\

\phantom{text}  \\[-0.5em]

{    }                
&    {    }              
& $\left( \begin{array}{cc} 
\mathtt{a}  &  \alpha
\end{array} \right)$     
&  $2$                                                  
& $2$   
& $2$                                              
& $[ \langle 0 , 1 \rangle , \langle 2 , 24 \rangle ]$                                                                     
& QSD                \\
{     }       &     {    }        &    where $\alpha \in \{  \mathtt{f}  ,  \mathtt{h}    \}$    \\

\phantom{text}  \\[-0.75em]
\cline{2-8}
\phantom{text}  \\[-0.5em]

 {    }     
& $\{  1 , 1  \}$  
& $\left(   \begin{array}{cc}  \mathtt{a}  &  \mathtt{e}    \\   0   &   \mathtt{b}  \end{array}   \right)$   
&  $1$             
& $4$      
& $1$                                           
& $[ \langle 0 , 1 \rangle  ,   \langle  1  ,  8  \rangle  ,  \langle  2 ,  116 \rangle  ]$                                
&   SD                 \\

\phantom{text}  \\[-0.75em]
\thickhline
\phantom{text}  \\[-0.5em]

$3$   
& $\{  0 ,  1  \}$ 
& $\left(   \begin{array}{ccc}  \mathtt{b}  & 0   &   0     \end{array}   \right)$           
&  $1$                     
& $16$     
& $1$                                           
& $[ \langle 0 , 1 \rangle  ,   \langle  1  ,  4  \rangle   ]$  
& SO          \\

\phantom{text}  \\[-0.5em]

{  }
& {   }
& $\left(   \begin{array}{ccc}  \mathtt{b}  & 0   &   \alpha     \end{array}   \right)$          
&  $2$                       
& $8$ 
& $2$                                               
& $[ \langle 0 , 1 \rangle  ,   \langle  2  ,  4  \rangle   ]$  
& SO          \\
{     }       &     {    }        &    where $\alpha \in \{  \mathtt{b}  ,  \mathtt{d}    \}$    \\

\phantom{text}  \\[-0.5em]

{  }
& {   }
& $\left(   
\begin{array}{ccc}  
\mathtt{b}  & \mathtt{b}    &   \mathtt{b}     
\end{array}   \right)$   
&  $1$                              
& $12$    
& $3$                                            
& $[ \langle 0 , 1 \rangle  ,   \langle  3  ,  4  \rangle   ]$  
& SO          \\

\phantom{text}  \\[-0.5em]

{  }
& {   }
& $\left(   
\begin{array}{ccc}  
\mathtt{b}  & \mathtt{b}    &   \mathtt{d}     
\end{array}   \right)$ 
&  $1$                                  
& $4$  
& $3$                                              
& $[ \langle 0 , 1 \rangle  ,   \langle  3  ,  4  \rangle   ]$  
& SO          \\

\phantom{text}  \\[-0.75em]
\cline{2-8}
\phantom{text}  \\[-0.5em]

{   }   
&  $\{  0 , 2   \}$
&  {$\left(  \begin{array}{ccc}    
0  &   \mathtt{b}  &    \mathtt{i}   \\
\mathtt{b} & \mathtt{d} & \mathtt{b}
\end{array}  \right)$}
&  $1$
&  $8$
& $1$
&  $[  \langle 0 , 1  \rangle , \langle  1 ,  4    \rangle  ,  \langle   2  ,  4   \rangle  ,  \langle  3 ,  16  \rangle   ]$
&  SO     \\

\phantom{text}  \\[-0.5em]

{   }   
&  
&  {$\left(  \begin{array}{ccc}    
\mathtt{b}  &   \mathtt{n}  &    \mathtt{b}   \\
\mathtt{b} & \mathtt{d} & \mathtt{b}
\end{array}  \right)$}
&  $1$
&  $8$
& $1$
&  $[  \langle 0 , 1  \rangle , \langle  1 ,  4    \rangle  ,  \langle   2  ,  4   \rangle  ,  \langle  3 ,  16  \rangle   ]$
&  SO     \\

\phantom{text}  \\[-0.5em]

{   }   
&  {  }
&  $\left(  \begin{array}{ccc}    
\mathtt{b}  &   0  &    0   \\
0 & \mathtt{b} & 0
\end{array}  \right)$
&  $1$
&  $16$
& $1$
&  $[  \langle 0 , 1  \rangle , \langle  1 ,  8    \rangle  ,  \langle   2  ,  16   \rangle  ]$
&  SO     \\

\phantom{text}  \\[-0.5em]

{   }   
&  {  }
&  {$\left(  \begin{array}{ccc}    
\mathtt{b}  &   0  &    \mathtt{n}   \\
\mathtt{b} & \mathtt{d} & \mathtt{b}
\end{array}  \right)$}
&  $1$
&  $12$
& $2$
&  $[  \langle 0 , 1  \rangle  ,  \langle   2  ,  12   \rangle  ,  \langle  3 ,  12  \rangle   ]$
&  SO     \\

\phantom{text}  \\[-0.5em]

{   }   
&  {  }
&  {$\left(  \begin{array}{ccc}    
\mathtt{b}  &   \mathtt{b}  &    \mathtt{d}   \\
\mathtt{b} & \mathtt{d} & \mathtt{b}
\end{array}  \right)$}
&  $1$
&  $4$
& $2$
&  $[  \langle 0 , 1  \rangle ,   \langle   2  ,  12   \rangle  ,  \langle  3 ,  12  \rangle   ]$
&  SO     \\

\phantom{text}  \\[-0.75em]
\cline{2-8}
\phantom{text}  \\[-0.5em]

{      }     
& $\{  0 ,  3  \}$ 
& $\left(   \begin{array}{ccc}  \mathtt{b}  & 0   &   0    \\   0   &   \mathtt{b}    &   0   \\   0   &   0   &   \mathtt{b}   \end{array}   \right)$      
&  $1$                             
& $48$   
& $1$                                             
& $[ \langle 0 , 1 \rangle  ,   \langle  1  ,  12  \rangle  ,  \langle  2 ,  48 \rangle  ,  \langle  3  ,  64  \rangle ]$  
& QSD                \\

\phantom{text}  \\[-0.75em]
\cline{2-8}
\phantom{text}  \\[-0.5em]

{   }    
&$ \{  1 , 0  \}$   
& $\left(   \begin{array}{ccc}  
\mathtt{a}  & \mathtt{b}   &   \alpha     
\end{array}     \right)$   
&  $2$                                                                                          
& $2$           
& $2$                                      
& $[ \langle 0 , 1 \rangle  ,   \langle  2 ,  4  \rangle ,  \langle  3  ,  20   \rangle   ]$                                
& SO                \\
{  }  &  {  }  &  where $\alpha \in \{  \mathtt{e} , \mathtt{f}     \}$  \\

\phantom{text}  \\[-0.5em]

{   }         
& {  }
& $\left(   \begin{array}{ccc}  
\mathtt{a}  & \alpha   &   \mathtt{b}    
\end{array}     \right)$    
&  $3$                                                                                           
& $2$            
& $2$                                     
& $[ \langle 0 , 1 \rangle  ,   \langle  2 ,  4  \rangle ,  \langle  3  ,  20   \rangle   ]$                                
& SO                \\
{  }  &  {  }  &  where $\alpha \in \{  \mathtt{h} , \mathtt{k} , \mathtt{s}     \}$  \\

\phantom{text}  \\[-0.5em]

{  }              
&  {   }
& $\left(   \begin{array}{ccc}  
\mathtt{a}  &    \mathtt{e}   &  0    
\end{array}     \right)$   
&  $1$                                                                                           
& $8$   
& $2$                                              
& $[ \langle 0 , 1 \rangle  ,   \langle  2  ,  24  \rangle   ]$                                
& SO                \\

\phantom{text}  \\[-0.5em]

{  }              
&  {   }
& $\left(   \begin{array}{ccc}  
\mathtt{a}  &    0   & \alpha
\end{array}     \right)$        
&  $2$                                                                                     
& $4$
& $2$                                                 
& $[ \langle 0 , 1 \rangle  ,   \langle  2  ,  24  \rangle   ]$                                
& SO                \\
{  }  &  {  }  &  where $\alpha \in \{  \mathtt{f} , \mathtt{h}    \}$  \\

\phantom{text}  \\[-0.75em]
\cline{2-8}
\phantom{text}  \\[-0.5em]

{    }				 
&  $\{  1 ,  1  \}^\dagger$   
&  {$ \left( \begin{array}{ccc}   
0  &  \mathtt{a}  & \mathtt{j}    \\   
\mathtt{b}  &  \mathtt{d}  &  \mathtt{b}    
\end{array} \right)$}       
&  $1$                                                          
& $8$    
& $1$                                             
& $[ \langle 0 , 1 \rangle  ,   \langle  1  ,  4  \rangle  ,  \langle  2 ,  24   \rangle  , \langle  3 ,  96   \rangle ]$                              
&      QSD                \\

\phantom{text}  \\[-0.5em]

{    }				
&         
&  {$\left( \begin{array}{ccc}   
\mathtt{a}  &  \mathtt{f}  &  \mathtt{i}    \\   
\mathtt{b}  &  \mathtt{d}  &  \mathtt{b}    
\end{array} \right)$}    
&  $1$                                                             
& $4$     
& $1$                                            
& $[ \langle 0 , 1 \rangle  ,   \langle  1  ,  4  \rangle  ,  \langle  2 ,  24   \rangle  , \langle  3 ,  96   \rangle ]$                              
&      QSD                \\ 

\phantom{text}  \\[-0.5em]

{    }				
&         
&  {$\left( \begin{array}{ccc}   
\mathtt{a}  &  \mathtt{h}  &  \mathtt{b}    \\   
\mathtt{b}  &  \mathtt{d}  &  \mathtt{b}    
\end{array} \right)$}    
&  $1$                                                             
& $4$     
& $1$                                            
& $[ \langle 0 , 1 \rangle  ,   \langle  1  ,  4  \rangle  ,  \langle  2 ,  24   \rangle  , \langle  3 ,  96   \rangle ]$                              
&      QSD                \\ 

\phantom{text}  \\[-0.5em]

{  	  }
&     {      }          
&  $\left( \begin{array}{ccc}   
\mathtt{a}  &  \mathtt{e}   & \mathtt{b}    \\   
0  &  \mathtt{b}  &  0     
\end{array} \right)$          
&  $1$                                                                 
& $2$     
& $1$                                           
& $[ \langle 0 , 1 \rangle  ,  \left\langle  1 , 8  \right\rangle  ,   \langle  2 ,  16   \rangle  , \langle  3 ,  100   \rangle ]$                              
& QSD                \\  

\phantom{text}  \\[-0.5em]

{  	  }
&     {    }         
&  $\left( \begin{array}{ccc}   
\mathtt{a}  &  \mathtt{e}   & 0    \\   
0  &  \mathtt{b}  &  0     
\end{array} \right)$        
&  $1$                                                                  
& $8$         
& $1$                                       
& $[ \langle 0 , 1 \rangle  ,  \left\langle  1 , 8  \right\rangle  ,   \langle  2 ,  116   \rangle   ]$                              
& QSD                \\  

\phantom{text}  \\[-0.5em]

{  	  }
&     {      }          
&  $\left( \begin{array}{ccc}   
\mathtt{a}  &  0   & \mathtt{e}    \\   
0  &  \mathtt{b}  &  \mathtt{b}     
\end{array} \right)$  
&  $1$                                                                        
& $2$    
& $2$                                            
& $[ \langle 0 , 1 \rangle  ,    \langle  2 ,  32   \rangle  , \langle  3 ,  92   \rangle ]$                              
& QSD                \\  

\phantom{text}  \\[-0.5em]

{  	  }
&     {      }          
&  $\left( \begin{array}{ccc}   
\mathtt{a}  &  \mathtt{e}   &  \mathtt{b}    \\   
0  &  \mathtt{b}  &  \alpha  
\end{array} \right)$     
&  $3$                                                                      
& $2$   
& $2$                                             
& $[ \langle 0 , 1 \rangle  ,    \langle  2 ,  32   \rangle  , \langle  3 ,  92   \rangle ]$                              
& QSD                \\  
{   }        &       {     }      &   where $\alpha \in \{ \mathtt{b} , \mathtt{i} , \mathtt{n}    \}$    \\

\phantom{text}  \\[-0.5em]

{  	  }
&     {      }          
&  $\left( \begin{array}{ccc}   
\mathtt{a}  &  \mathtt{e}   &  \mathtt{d}    \\   
0  &  \mathtt{b}  &  \mathtt{b}     
\end{array} \right)$  
&  $1$                                                                         
& $2$      
& $2$                                          
& $[ \langle 0 , 1 \rangle  ,    \langle  2 ,  32   \rangle  , \langle  3 ,  92   \rangle ]$                              
& QSD                \\  

\phantom{text}  \\[-0.75em]
\cline{2-8}
\phantom{text}  \\[-0.5em]

{    }   
&     $\{  1 , 2  \}$        
&  $\left( \begin{array}{ccc}   \mathtt{a}  &  \mathtt{e}  & 0    \\   
0  &  \mathtt{b}  &  0  \\
0 & 0 & \mathtt{b}
\end{array} \right)$  
&  $1$                                                                          
& $8$    
& $1$                                            
& $[ \langle 0 , 1 \rangle  ,   \langle  1  ,  12  \rangle  ,  \langle  2 ,  148   \rangle  , \langle  3 ,  464   \rangle ]$                              
& SO                \\

\phantom{text}  \\[-0.75em]
\thickhline   
\phantom{text}  \\[-0.5em]

$4$   
&  $\{  0 ,  1   \}$  
&    $\left(  \begin{array}{cccc}   \mathtt{b}  &  0 & 0  &  0      \end{array}  \right)$   
&  $1$
&$96$  
& $1$
&   $[  \langle  0 , 1  \rangle ,  \langle  1 , 4  \rangle ]$  
&   SO   \\

\phantom{text}  \\[-0.5em]

{   }
&  {    }
&    $\left(  \begin{array}{cccc}   
\mathtt{b}  &  \alpha & 0  &  0      \end{array}  \right)$   
&  $2$
&$32$  
& $2$
&   $[  \langle  0 , 1  \rangle ,  \langle  2 , 4  \rangle ]$  
&   SO   \\
{  }   &  {   }   &  where $\alpha \in \{ \mathtt{b}   , \mathtt{d}     \}$    \\

\phantom{text}  \\[-0.5em]

{   }
&  {    }  
&    $\left(  \begin{array}{cccc}   \mathtt{b}  &  \mathtt{b}  & \mathtt{b}  &  0      \end{array}  \right)$   
&  $1$
&$24$  
& $3$
&   $[  \langle  0 , 1  \rangle ,  \langle  3 , 4  \rangle ]$  
&   SO   \\

\phantom{text}  \\[-0.5em]

{   }
&  {    }  
&    $\left(  \begin{array}{cccc}   \mathtt{b}  &  \mathtt{b}  & \mathtt{d}  &  0      \end{array}  \right)$   
&  $1$
&$8$  
& $3$
&   $[  \langle  0 , 1  \rangle ,  \langle  3 , 4  \rangle ]$  
&   SO   \\

\phantom{text}  \\[-0.5em]

{   }
&  {    }  
&    $\left(  \begin{array}{cccc}   \mathtt{b}  &  \mathtt{b}  & \mathtt{b}  &  \mathtt{b}      \end{array}  \right)$   
&  $1$
&$48$  
& $4$
&   $[  \langle  0 , 1  \rangle ,  \langle  4 , 4  \rangle ]$  
&   SO   \\

\phantom{text}  \\[-0.5em]

{   }
&  {    }  
&    $\left(  \begin{array}{cccc}   \mathtt{b}  &  \mathtt{b}  & \mathtt{d}  &  \mathtt{d}      \end{array}  \right)$   
&  $1$
&$16$  
& $4$
&   $[  \langle  0 , 1  \rangle ,  \langle  4 , 4  \rangle ]$  
&   SO   \\

\phantom{text}  \\[-0.5em]

{   }
&  {    }  
&    $\left(  \begin{array}{cccc}   \mathtt{b}  &  \mathtt{b}  & \mathtt{b}  &  \mathtt{d}      \end{array}  \right)$   
&  $1$
&$12$ 
& $4$ 
&   $[  \langle  0 , 1  \rangle ,  \langle  4 , 4  \rangle ]$  
&   SO   \\

\phantom{text}  \\[-0.75em]
\cline{2-8}
\phantom{text}  \\[-0.5em]

{  }   
&  $\{ 0  ,  2 \}$   
&    $\left(  \begin{array}{cccc}  
\mathtt{b}  &  0 & 0 &  0  \\ 
0  & \mathtt{b} & \alpha & 0         
\end{array}  \right)$   
&  $2$
&$16$  
& $1$
&   $[  \langle  0 , 1  \rangle ,  \langle  1 , 4  \rangle,  \langle  2 , 4  \rangle ,  \langle   3  ,  16  \rangle  ]$  
&   SO   \\ 
{  }  &  {  }  &  where $\alpha  \in \{ \mathtt{b}  , \mathtt{d}   \}$   \\

\phantom{text}  \\[-0.5em]

{  }  
&  {    }  
&    $\left(  \begin{array}{cccc}  
\mathtt{b}  &  0 & 0 & 0   \\
0  & \mathtt{b} &\mathtt{b} & \mathtt{b}        
\end{array}  \right)$   
&  $1$
&$24$  
& $1$
&   $[  \langle  0 , 1  \rangle ,  \langle  1 , 4  \rangle,  \langle  3 , 4  \rangle ,  \langle   4  ,  16  \rangle  ]$  
&   SO   \\ 

\phantom{text}  \\[-0.5em]

{  }  
&  {    }  
&    $\left(  \begin{array}{cccc}  
\mathtt{b}  &  0 & 0 & 0  \\
0  & \mathtt{b} & \mathtt{b}  & \mathtt{d}        
\end{array}  \right)$   
&  $1$
&$8$  
& $1$
&   $[  \langle  0 , 1  \rangle ,  \langle  1 , 4  \rangle,  \langle  3 , 4  \rangle ,  \langle   4  ,  16  \rangle  ]$  
&   SO   \\ 

\phantom{text}  \\[-0.5em]

{  }     
&  {   }   
&    $\left(  \begin{array}{cccc}  \mathtt{b}  &  0 & 0  &  0  \\ 
0  & \mathtt{b} & 0 & 0         
\end{array}  \right)$
&  $1$   
&$64$  
& $1$
&   $[  \langle  0 , 1  \rangle ,  \langle  1 , 8 \rangle,  \langle  2 , 16  \rangle  ]$  
&   SO   \\

\phantom{text}  \\[-0.5em]

{  }  
&  {    }  
&    $\left(  \begin{array}{cccc}  
\mathtt{b}  & 0 & \mathtt{b} &  0  \\
0  & \mathtt{b} & \alpha & \mathtt{b}         
\end{array}  \right)$   
&  $2$
&$8$  
& $2$
&   $[  \langle  0 , 1  \rangle ,  \langle  2, 4  \rangle,   \langle   3  ,  8  \rangle  ,   \langle  4  ,  12  \rangle  ]$  
&   SO   \\   
{  }  &  {  }  &  where $\alpha  \in \{ \mathtt{b} , \mathtt{d}    \}$   \\

\phantom{text}  \\[-0.5em]

{  }  
&  {    }  
&    $\left(  \begin{array}{cccc}  
\mathtt{b}  & 0   & \mathtt{d} & 0  \\
0  & \mathtt{b} & \mathtt{b} & \alpha
\end{array}  \right)$   
&  $2$
&$4$  
& $2$
&   $[  \langle  0 , 1  \rangle ,  \langle  2, 4  \rangle,   \langle   3  ,  8  \rangle  ,   \langle  4  ,  12  \rangle  ]$  
&   SO   \\   
{  }  &  {  }  &  where $\alpha  \in \{ \mathtt{b}   , \mathtt{d}   \}$   \\

\phantom{text}  \\[-0.5em]

{  }  
&  {    }  
&    $\left(  \begin{array}{cccc}  
\mathtt{b}  & 0  & 0 &  \mathtt{b}  \\ 
0  & \mathtt{b} & \mathtt{d} & \mathtt{b}
\end{array}  \right)$ 
&  $1$  
&$4$  
& $2$
&   $[  \langle  0 , 1  \rangle ,  \langle  2, 4  \rangle,   \langle   3  ,  8  \rangle  ,   \langle  4  ,  12  \rangle  ]$  
&   SO   \\   

\phantom{text}  \\[-0.5em]

{  }  
&  {    }  
&    $\left(  \begin{array}{cccc} 
 \mathtt{b}  &  0  & \alpha &  0    \\ 
0  & \mathtt{b} & 0 & \alpha   
\end{array}  \right)$   
&  $2$
&$32$  
& $2$
&   $[  \langle  0 , 1  \rangle ,  \langle  2, 8  \rangle,   \langle   4  ,  16  \rangle  ]$  
&   SO   \\ 
{  }  &  {  }  &  where $\alpha  \in \{ \mathtt{b} , \mathtt{d}   \}$   \\

\phantom{text}  \\[-0.5em]

{  }  
&  {    }   
&    $\left(  \begin{array}{cccc}  
\mathtt{b} & 0  & \mathtt{b} & 0  \\
0  & \mathtt{b} & 0  &  \mathtt{d}
\end{array}  \right)$   
&  $1$
& $16$  
& $2$
&   $[  \langle  0 , 1  \rangle ,  \langle  2, 8  \rangle,   \langle   4  ,  16  \rangle  ]$  
&   SO   \\ 

\phantom{text}  \\[-0.5em]

{  }  
&  {    }  
&    $\left(  \begin{array}{cccc}  
\mathtt{b}  &  0 & \mathtt{b} &  0    \\ 
0  & \mathtt{b} & \mathtt{b} & 0         
\end{array}  \right)$  
&  $1$ 
&$24$  
& $2$
&   $[  \langle  0 , 1  \rangle ,    \langle   2  ,  12    \rangle  ,   \langle  3  ,  12    \rangle  ]$  
&   SO   \\     

\phantom{text}  \\[-0.5em]

{  }  
&  {    }  
&    $\left(  \begin{array}{cccc}  
\mathtt{b}  & 0 & \mathtt{d} &  \alpha   \\
0  & \mathtt{b} & \mathtt{b} & \alpha         
\end{array}  \right)$
&  $2$   
&$8$  
& $2$
&   $[  \langle  0 , 1  \rangle ,    \langle   2  ,  12    \rangle  ,   \langle  3  ,  12    \rangle  ]$  
&   SO   \\     
{  }  &  {  }  &  where $\alpha  \in \{  0 , \mathtt{b}   \}$   \\

\phantom{text}  \\[-0.5em]

{  }  
&  {    }  
&    $\left(  \begin{array}{cccc}  
\mathtt{b}  & 0  &  \mathtt{d}  &  \mathtt{b}  \\ 
0 & \mathtt{b}  & \mathtt{b}  &  \mathtt{d}
\end{array}  \right)$   
&  $1$
&$8$  
& $3$
&   $[  \langle  0 , 1  \rangle ,    \langle   3  ,  16    \rangle  ,   \langle  4  ,  8  \rangle  ]$  
&   SO   \\   

\phantom{text}  \\[-0.5em]

{  }  
&  {    }  
&    $\left(  \begin{array}{cccc}  
\mathtt{b} & 0 &  \mathtt{d} &   \mathtt{d}  \\ 
0 & \mathtt{b}  & \mathtt{b} & \mathtt{d}
\end{array}  \right)$   
&  $1$
&$4$  
& $3$
&   $[  \langle  0 , 1  \rangle ,    \langle   3  ,  16    \rangle  ,   \langle  4  ,  8  \rangle  ]$  
&   SO   \\   

\phantom{text}  \\[-0.75em]
\cline{2-8}
\phantom{text}  \\[-0.5em]

{    }    
&  $\{ 0  ,  3 \}$    
&    {$\left(  \begin{array}{cccc}  
0  &  \mathtt{b} & \mathtt{b}  & 0  \\
0  & 0 & 0  &  \mathtt{b}  \\
\mathtt{d} & \mathtt{n} & \mathtt{b} & 0
\end{array}  \right)$}   
&  $1$
&   $24$  
& $1$
&   $[  \langle  0 , 1  \rangle ,   \langle  1 , 4  \rangle ,    \langle  2 , 12  \rangle ,  \langle  3 , 60  \rangle ,  \langle  4  ,  48  \rangle  ]$  
&   SO   \\

\phantom{text}  \\[-0.5em]

{  }  
&    {    }  
&    {$\left(  \begin{array}{cccc}  
0  &  \mathtt{b} & \mathtt{b} & \mathtt{i}  \\
\mathtt{b}  & \mathtt{d} &  0  & \mathtt{b}     \\  
\mathtt{d} & \mathtt{n} & \mathtt{b} & 0
\end{array}  \right)$}   
&  $1$
&   $8$ 
& $1$ 
&   $[  \langle  0 , 1  \rangle ,  \langle  1 , 4  \rangle,  \langle  2 , 12  \rangle ,  \langle  3 , 60  \rangle ,  \langle  4  ,  48  \rangle  ]$  
&   SO   \\

\phantom{text}  \\[-0.5em]

{  }  
&    {    }   
&    {$\left(  \begin{array}{cccc}  
0  &  \mathtt{b} & \mathtt{i} & 0   \\
0  & 0 & 0 & \mathtt{b}     \\  
\mathtt{b} & \mathtt{d} & \mathtt{b} & 0
\end{array}  \right)$}   
&  $1$
&   $32$  
& $1$
&   $[  \langle  0 , 1  \rangle ,  \langle  1 , 8  \rangle,  \langle  2 , 20  \rangle ,  \langle  3 , 32  \rangle ,  \langle  4  ,  64  \rangle  ]$  
&   SO   \\   

\phantom{text}  \\[-0.5em]

{  }  
&    {    }   
&    $\left(  \begin{array}{cccc}  
\mathtt{b}  &  0 & 0 & 0   \\
0  & \mathtt{b} & 0 & 0     \\  
0 & 0 & \mathtt{b} & \mathtt{b}
\end{array}  \right)$   
&  $1$
&   $32$  
& $1$
&   $[  \langle  0 , 1  \rangle ,  \langle  1 , 8  \rangle,  \langle  2 , 20  \rangle ,  \langle  3 , 32  \rangle ,  \langle  4  ,  64  \rangle  ]$  
&   SO   \\   

\phantom{text}  \\[-0.5em]

{    }   
&  {    }    
&    {$\left(  \begin{array}{cccc}  
0  &  \mathtt{b} & 0  &  0  \\ 
0  & 0 & 0 & \mathtt{b}     \\  
0 & 0 & \mathtt{b} & 0  
\end{array}  \right)$}   
&  $1$
&   $96$
&   $1$  
&   $[  \langle  0 , 1  \rangle ,  \langle  1 , 12  \rangle,  \langle  2 , 48  \rangle ,  \langle  3 , 64  \rangle   ]$  
&   SO   \\

\phantom{text}  \\[-0.5em]

{  }  
&    {    }  
&    {$\left(  \begin{array}{cccc}  
0  &  \mathtt{b} & \mathtt{i} &  \mathtt{i}  \\
\mathtt{b}  & \mathtt{d} & 0  &  \mathtt{b}  \\
\mathtt{b} & \mathtt{d} & \mathtt{b} & 0
\end{array}  \right)$} 
&  $1$  
&   $48$  
& $2$
&   $[  \langle  0 , 1  \rangle ,    \langle  2 , 24  \rangle ,  \langle  3 , 48  \rangle ,  \langle  4  ,  52  \rangle  ]$  
&   SO   \\

\phantom{text}  \\[-0.5em]

{  }  
&    {    }   
&    $\left(  \begin{array}{cccc}  
\mathtt{b}  &  0 & 0 & \mathtt{b}    \\
0  & \mathtt{b} & 0  &  \mathtt{d}   \\
0 & 0 & \mathtt{b} & \mathtt{d}  
\end{array}  \right)$   
&  $1$
&   $16$  
& $2$
&   $[  \langle  0 , 1  \rangle ,    \langle  2 , 24  \rangle ,  \langle  3 , 48  \rangle ,  \langle  4  ,  52  \rangle  ]$  
&   SO   \\

\phantom{text}  \\[-0.5em]

{  }  
&    {    }   
&    {$\left(  \begin{array}{cccc}  
0  &  \mathtt{b} & \mathtt{b} & \mathtt{b}  \\
\mathtt{d}  & \mathtt{n} &  0  & \mathtt{b}     \\  
\mathtt{d} & \mathtt{n} & \mathtt{b} & 0  
\end{array}  \right)$}   
&  $1$
&   $12$  
& $2$
&   $[  \langle  0 , 1  \rangle ,    \langle  2 , 24  \rangle ,  \langle  3 , 48  \rangle ,  \langle  4  ,  52  \rangle  ]$  
&   SO   \\

\phantom{text}  \\[-0.75em]
\cline{2-8}
\phantom{text}  \\[-0.5em]

{  }    
&  $\{ 0  ,  4 \}$   
&    $\left(  \begin{array}{cccc}  \mathtt{b}  &  0 & 0  &  0  \\ 
0  & \mathtt{b} & 0 & 0  \\ 
0 & 0 & \mathtt{b} & 0  \\  0 & 0 & 0 & \mathtt{b}         \end{array}  \right)$   
&  $1$
& $384$  
& $1$
&   $[  \langle  0 , 1  \rangle ,  \langle  1 , 16 \rangle,  \langle  2 , 96  \rangle,  \langle 3  , 256  \rangle,  \langle  4  , 256   \rangle  ]$  
&   QSD   \\

\phantom{text}  \\[-0.75em]
\cline{2-8}
\phantom{text}  \\[-0.5em]

{  }
&  $\{  1 , 0  \}^\S$   
&   $\left(  \begin{array}{cccc}     
\mathtt{a}  &  \mathtt{a}  &  \mathtt{e}  &  \mathtt{e}
\end{array}  \right)$   
&  $1$
&     $16$   
& $4$  
&   $[   \langle 0 , 1  \rangle ,  \langle  4 , 24 \rangle  ]$    
&   SO  \\  

\phantom{text}  \\[-0.5em]

{  }
&  {   }
&   $\left(  \begin{array}{cccc}     
\mathtt{a}  &  \ovariant  & \mathtt{u}  &  \mathtt{s}
\end{array}  \right)$   
&     $1$ 
&     $8$     
& $4$
&   $[   \langle 0 , 1  \rangle ,  \langle  4 , 24 \rangle  ]$    
&   SO  \\  

\phantom{text}  \\[-0.5em]

{  }
&  {   }
&   $\left(  \begin{array}{cccc}     
\mathtt{a}  &  \ovariant  & \mathtt{x}  &  \mathtt{x}
\end{array}  \right)$   
&     $1$ 
&     $8$     
& $4$
&   $[   \langle 0 , 1  \rangle ,  \langle  4 , 24 \rangle  ]$    
&   SO  \\  

\phantom{text}  \\[-0.75em]
\cline{2-8}
\phantom{text}  \\[-0.5em]

{   }   
&  $\{  1  ,  1  \}^\S$ 
&   $\left(  \begin{array}{cccc}     
\mathtt{a}  &  \mathtt{a}  & \mathtt{w}  &  \mathtt{j}   \\
0  &  \mathtt{b}    & \mathtt{n}  &  \mathtt{n}
\end{array}  \right)$
&     $1$    
&     $2$     
&     $3$   
&   $[   \langle 0 , 1  \rangle ,   \langle  3 , 16  \rangle , \langle  4 , 108  \rangle   ]$  
&   SO  \\   

\phantom{text}  \\[-0.5em]

{   }
&  {    }
&   {$\left(  \begin{array}{cccc}     
0  &  \mathtt{b}  & \mathtt{d}  &  \mathtt{b}   \\
\mathtt{s}  &  \mathtt{q}    & \mathtt{a}  &  \mathtt{f}
\end{array}  \right)$}   
&     $1$ 
&     $2$     
&     $3$ 
&   $[   \langle 0 , 1  \rangle ,   \langle  3 , 16  \rangle , \langle  4 , 108  \rangle   ]$  
&   SO  \\   

\phantom{text}  \\[-0.5em]

{   }
&  {    }
&   $\left(  \begin{array}{cccc}     
\mathtt{a}  &  \mathtt{a}  & \mathtt{f}  &  \mathtt{e}   \\
0  &  \mathtt{b}    & \mathtt{i}  &  \mathtt{b}
\end{array}  \right)$  
&     $1$  
&     $2$ 
&     $3$     
&   $[   \langle 0 , 1  \rangle ,   \langle  3 , 16  \rangle , \langle  4 , 108  \rangle   ]$  
&   SO  \\   

\phantom{text}  \\[-0.75em]
\cline{2-8}
\phantom{text}  \\[-0.5em]

{    } 
&  $\{  1  ,  2  \}^\S$   
&   $\left(  \begin{array}{cccc}    
\mathtt{a}  &  \alpha   & 0  & \mathtt{d}  \\ 
0 & \mathtt{b} & 0 &  \mathtt{i}       \\ 
0 & 0 & \mathtt{b} & \mathtt{b}            
\end{array}  \right)$  
&   $2$   
&     $4$     
&     $2$     
&   $[   \langle 0 , 1  \rangle ,    \langle  2 , 44 \rangle ,  \langle  3 , 208 \rangle ,  \langle  4 ,  372 \rangle  ]$  
&   QSD  \\  
{  }   &   {   }  &  where $\alpha \in \{ \mathtt{e} , \mathtt{j}   \}$  \\

\phantom{text}  \\[-0.5em]

{ }  
&  { }  
&   $\left(  \begin{array}{cccc}    
\mathtt{a}  &  \mathtt{j}   & 0  & \mathtt{d}  \\ 
0 & \mathtt{b} & 0 & \mathtt{i}  \\ 
0 & 0 & \mathtt{b} & \mathtt{d}            
\end{array}  \right)$     
&   $1$
&     $2$     
&     $2$   
&   $[   \langle 0 , 1  \rangle ,    \langle  2 , 44 \rangle ,  \langle  3 , 208 \rangle ,  \langle  4 ,  372 \rangle  ]$  &   QSD  \\ 

\phantom{text}  \\[-0.5em]

{ }  
&  {   }  
&   $\left(  \begin{array}{cccc}    
\mathtt{a}  &  \mathtt{e}   & 0  & \mathtt{n}  \\ 
0 & \mathtt{b} & 0 & \mathtt{d}  \\ 
0 & 0 & \mathtt{b} & \mathtt{i}            
\end{array}  \right)$     
&   $1$
&     $2$     
&     $2$    
&   $[   \langle 0 , 1  \rangle ,    \langle  2 , 44 \rangle ,  \langle  3 , 208 \rangle ,  \langle  4 ,  372 \rangle  ]$  &   QSD  \\ 

\phantom{text}  \\[-0.75em]
\cline{2-8}
\phantom{text}  \\[-0.5em]

{   }   
&  $\{  1  ,  3  \}$   
&   $\left(  \begin{array}{cccc} 
\mathtt{a}  &  \mathtt{a}  &  \mathtt{e}  &  \mathtt{e}  \\
0   &  \mathtt{b}  &  0  &  0   \\
0   &  0  &  \mathtt{b}  & 0   \\
0  & 0 & 0 & \mathtt{b}
\end{array}  \right)$   
&  $1$
&     $16$     
&     $1$   
&   $[   \langle 0 , 1  \rangle ,  \langle  1 ,  16  \rangle  ,  \langle  2 , 96  \rangle ,  \langle  3 , 256 \rangle ,  \langle  4 ,  2756 \rangle  ]$  
&   SO  \\

\phantom{text}  \\[-0.5em]

{   }   
&  {   }  
&   $\left(  \begin{array}{cccc} 
\mathtt{a}  &  \mathtt{e}  &  0  &  0  \\
0   &  \mathtt{b}  &  0  &  0   \\
0   &  0  &  \mathtt{b}  & 0   \\
0  & 0 & 0 & \mathtt{b}
\end{array}  \right)$   
&  $1$
&     $32$  
&     $1$      
&   $[   \langle 0 , 1  \rangle ,  \langle  1 ,  16  \rangle  ,  \langle  2 , 196  \rangle ,  \langle  3 , 1056 \rangle ,  \langle  4 ,  1856 \rangle  ]$  
&   SO  \\

\phantom{text}  \\[-0.75em]
\cline{2-8}
\phantom{text}  \\[-0.5em]

{   }   
&   $\{ 2 , 0  \}^\S$  
&   {$\left(  \begin{array}{cccc}     
0  &  \mathtt{a}  &  0  &  \mathtt{j}   \\
\mathtt{f}  &  \mathtt{t}  &  \mathtt{a}  &  \mathtt{f}
\end{array}  \right)$}
&   $1$   
&     $8$     
&     $2$  
&   $[   \langle 0 , 1  \rangle ,   \langle  2 , 48 \rangle ,   \langle  4 ,  576 \rangle  ]$  
&   QSD  \\   

\phantom{text}  \\[-0.5em]

{   }
&   {     }
&   $\left(  \begin{array}{cccc}     
\mathtt{a}  &  0  &  0  &  \mathtt{x}   \\
0  &  \mathtt{a}  &  \mathtt{w}  &  0
\end{array}  \right)$   
&   $1$
&     $8$     
&     $2$    
&   $[   \langle 0 , 1  \rangle ,   \langle  2 , 48 \rangle ,   \langle  4 ,  576 \rangle  ]$  
&   QSD  \\   

\phantom{text}  \\[-0.5em]

{   }
&   {     }
&   {$\left(  \begin{array}{cccc}     
\mathtt{a}  &  \mathtt{g}  &  \mathtt{f}  &  \mathtt{f}   \\
\mathtt{k}  &  \mathtt{g}  &  \mathtt{a}  &  \mathtt{f}
\end{array}  \right)$}   
&   $1$
&     $4$    
&     $2$   
&   $[   \langle 0 , 1  \rangle ,   \langle  2 , 48 \rangle ,   \langle  4 ,  576 \rangle  ]$  
&   QSD  \\   

\phantom{text}  \\[-0.75em]
\cline{2-8}
\phantom{text}  \\[-0.5em]

{    }   
&   $\{  2  ,  1   \}^\S$ 
&   $\left(  \begin{array}{cccc}     
\mathtt{a}  &  0  &  0  &  \mathtt{x}  \\
0  &  \mathtt{a}  &  \mathtt{e}  &  0   \\
0  &  0  &  \mathtt{b}  &  \mathtt{n}
\end{array}  \right)$  
&   $1$ 
&     $2$     
&     $2$   
&   $[   \langle 0 , 1  \rangle ,    \langle  2 , 64  \rangle ,  \langle  3 , 368  \rangle  ,  \langle  4 ,  2692  \rangle  ]$  
&   SO  \\    

\phantom{text}  \\[-0.5em]

{    }
&   {    }
&   $\left(  \begin{array}{cccc}     
\mathtt{a}  &  0  &  0  &  \alpha  \\
0  &  \mathtt{a}  &  \mathtt{e}  &  \mathtt{n}   \\
0  &  0  &  \mathtt{b}  &  \mathtt{n}
\end{array}  \right)$   
&   $2$
&     $2$     
&     $2$  
&   $[   \langle 0 , 1  \rangle ,    \langle  2 , 64  \rangle ,  \langle  3 , 368  \rangle  ,  \langle  4 ,  2692  \rangle  ]$  
&   SO  \\    
{  }  &  {  }  &  where $\alpha \in \{ \mathtt{h} ,  \mathtt{s}   \}$   \\

\phantom{text}  \\[-0.5em]

{    }
&   {     }
&   $\left(  \begin{array}{cccc}     
\mathtt{a}  &  0  &  0  &  \mathtt{h}  \\
0  &  \mathtt{a}  &  \mathtt{j}  &  0   \\
0  &  0  &  \mathtt{b}  &  \mathtt{d}
\end{array}  \right)$  
&   $1$ 
&     $2$     
&     $2$   
&   $[   \langle 0 , 1  \rangle ,    \langle  2 , 64  \rangle ,  \langle  3 , 368  \rangle  ,  \langle  4 ,  2692  \rangle  ]$  
&   SO  \\    

\phantom{text}  \\[-0.75em]
\cline{2-8}
\phantom{text}  \\[-0.5em]

{    }   
&   $\{  2  ,  2   \}$ 
&   $\left(  \begin{array}{cccc}     
\mathtt{a}  &  0  &  \mathtt{e}  &  0  \\
0  &  \mathtt{a}  &  0  &  \mathtt{e}   \\
0  &  0  &  \mathtt{b}  &  0     \\
0  &  0  &  0  &  \mathtt{b}
\end{array}  \right)$   
&  $1$
&     $32$     
&     $1$   
&   $[   \langle 0 , 1  \rangle ,   \langle  1  ,  16  \rangle  ,  \langle  2 , 296  \rangle ,  \langle  3 , 1856  \rangle  ,  \langle  4 ,  13456   \rangle  ]$  
&   SD \\   

\phantom{text}  \\[-0.75em]
\thickhline

\vspace*{-25pt}

\end{longtable}

}
\endgroup

\section{Conclusion}
In this paper, we have constructed propagation rules for codes over the ring $I_5$ and used them to generate self-orthogonal codes, quasi self-dual codes, and self-dual codes over $I_5$. Using the building-up method together with the mass formula, we have classified all three categories of codes completely in length at most $n=4$ up to the monomial equivalence for all types $\{ k_1 , k_2 \}$. It is also a good problem to study other categories of codes over $I_p$, for instance, cyclic codes over $I_p$.

\section*{Declaration}
\textbf{Conflict of Interest} The authors declare that they have no conflict of interest in the manuscript.

\end{document}